%% file: doc.tex
\documentclass[a4paper, oneside,hidelinks]{book}

\usepackage[utf8]{inputenc}
\usepackage[english]{babel}

\usepackage[T1]{fontenc}

\usepackage{adjustbox}

\usepackage{subcaption}
\usepackage{graphicx,wrapfig,lipsum}
\usepackage[left=3.5cm,right=3cm,top=3cm,bottom=3cm]{geometry}
\usepackage{hyperref}
\usepackage{setspace}
\usepackage{shapepar}
\usepackage{listings}
\usepackage[table,xcdraw]{xcolor}
\usepackage{fonttable}
\usepackage{amssymb,amsmath,amsthm}

\newcommand{\floor}[1]{\lfloor {#1} \rfloor}

\usepackage[utf8]{inputenc}
\usepackage{mathtools}
\usepackage[table,xcdraw]{xcolor}
\usepackage{pgfplots}
\pgfplotsset{compat=newest}
\usepackage{float}

\usepackage{nicefrac,xfrac}

\usepackage{dsfont}

\usetikzlibrary{quantikz}

\usepackage{wrapfig} 

\usepackage{caption}
\usepackage{subcaption}

\usepackage{bbm}

\usepackage{fancyhdr}
\usepackage{stmaryrd}

\newtheorem*{theorem}{Theorem}

\newtheorem*{remark}{Remark}
\usepackage{optidef}

\usepackage[linesnumbered,ruled,vlined]{algorithm2e}

\usepackage{listings}
\colorlet{punct}{red!60!black}
\definecolor{background}{HTML}{EEEEEE}
\definecolor{delim}{RGB}{20,105,176}
\colorlet{numb}{magenta!60!black}
\definecolor{dkgreen}{rgb}{0,0.6,0}
\definecolor{gray}{rgb}{0.5,0.5,0.5}
\definecolor{mauve}{rgb}{0.58,0,0.82}

\hypersetup{
	colorlinks=false,
	linktoc=all, 
}

\usepackage[hyperref,backend=biber,backref,backrefstyle=none,sorting=none, safeinputenc]{biblatex}
\usepackage{fancyhdr}

\definecolor{ion}{RGB}{253, 202, 2}

\newcommand{\makelogo}{%
	\centering
        \phantom{cos}\\\vspace{4cm}
	\textsc{University of Naples Federico II}\\
	\textsc{University of Camerino}\\
	\textsc{National Research Council of Italy}
	\\[1.1cm]
	\includegraphics[width=0.2\textwidth]{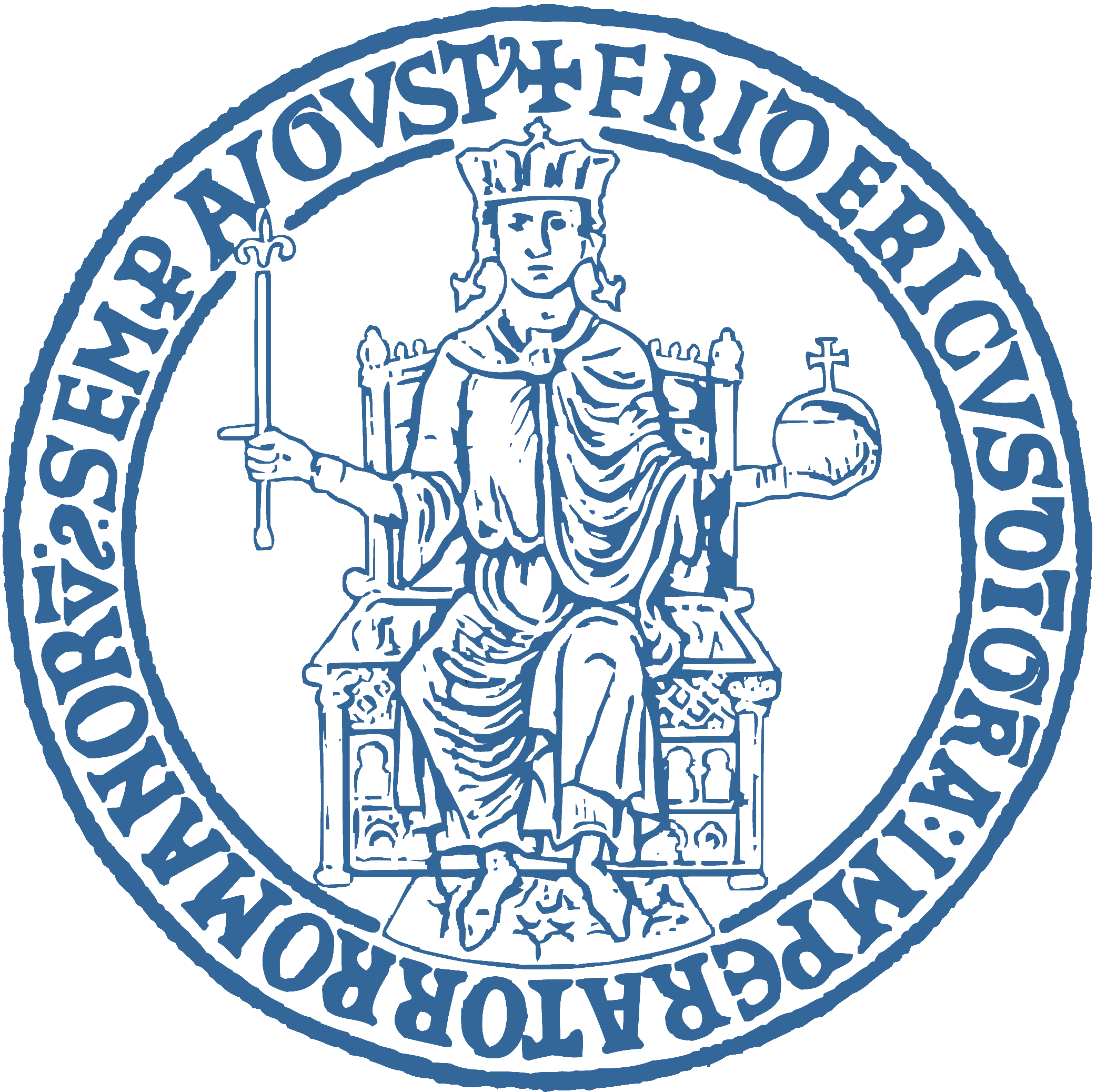}\hspace{1cm} 
	\includegraphics[width=0.155\textwidth]{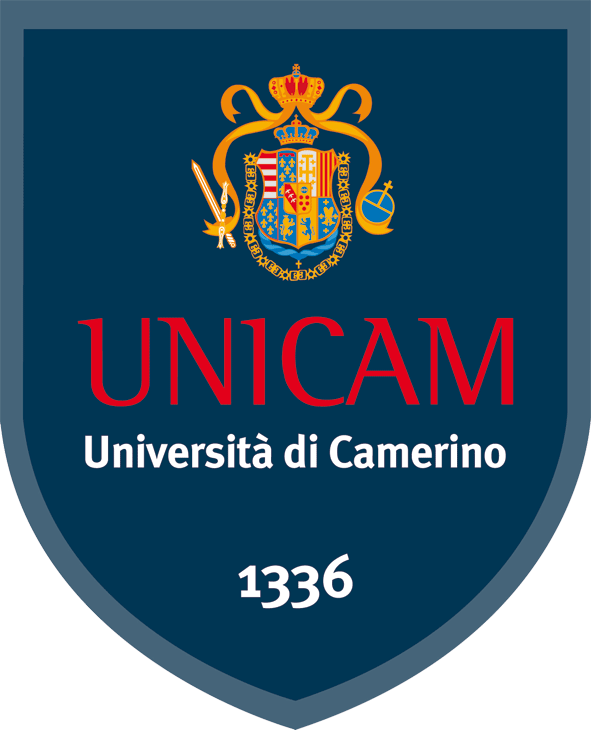}\hspace{0.8cm}
	\includegraphics[width=0.215\textwidth]{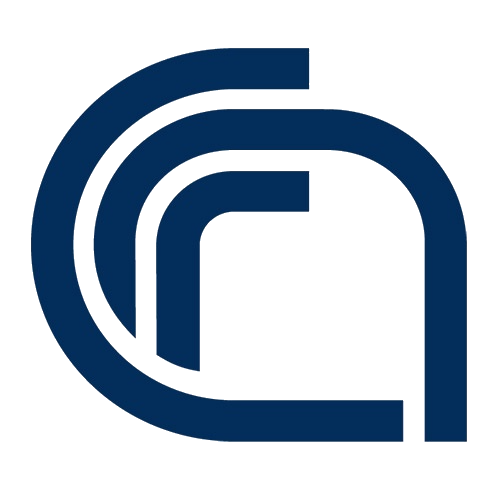}
	\\[1.1cm]
	\textsc{Doctor of Philosophy in Quantum Technologies}
	\\[4cm]
	\textbf{\href{http://www.fedoatd.unina.it/961/}{\large{Architectures and circuits for distributed quantum computing}}}
	\\[1cm]
	\textbf{\href{https://linktr.ee/daniele.cuomo}{\large{Daniele Cuomo}}}
        \vfill
	\par
}

\begin{document}
	\onehalfspacing
	\author{Daniele Cuomo}
	\title{Architectures and circuits for distributed quantum computing}

	{\frontmatter
	\let\cleardoublepage\clearpage
    \makelogo
    \thispagestyle{empty}

\include{abstract.tex}
            
		\tableofcontents
            
		\mainmatter
	}
	\let\cleardoublepage\clearpage{
            \pagenumbering{arabic}
		\include{tec-integration/tec-integration}
		\include{resources/resources}
		\include{logic/logical_computing}
		\include{compilation/compilation}
		\appendix
	}
		
	\backmatter
	
	
\end{document}

%% file: abstract.tex
\null\vfill

\section*{\centering\large Abstract}
\pagenumbering{roman}
\begin{flushleft}
This thesis treats networks providing \textit{quantum computation} based on \textit{distributed paradigms}.
Compared to architectures relying on one processor, a network promises to be more \textit{scalable} and less \textit{fault-prone}.

Developing a distributed system able to provide practical quantum computation comes with many challenges, each of which need to be faced with careful analysis in order to create a massive integration of several components properly engineered.

In accordance with hardware technologies, currently under construction around the globe, \textit{telegates} represent the fundamental inter-processor operations. Each telegate consists of several tasks: i) entanglement generation and distribution, ii) local operations, and iii) classical communications. Entanglement generation and distribution is an expensive resource, as it is time-consuming.

The main contribution of this thesis is on the definition of \textit{compilers} that minimize the impact of telegates on the overall \textit{fidelity}. Specifically, we give rigorous formulations of the subject problem, allowing us to identify the inter-dependence between computation and communication. With the support of some of the best tools for reasoning -- i.e. network optimization, circuit manipulation, group theory and \texttt{ZX}-calculus -- we found new perspectives on the way a distributed quantum computing system should evolve.
\end{flushleft}

\vfill\null

%% file: tec-integration/tec-integration.tex
\begin{refsection}
\chapter{Introduction}
\label{ch:techs}
\thispagestyle{empty}
\newpage

Distributed quantum computing is one of the most appealing applications in the panorama of quantum technologies. In fact, distributed architectures could be our bridge to step beyond the current NISQ era \cite{Preskill2018, cuomo2020towards, gibney2019quantum, gambetta2020ibm, CacCalTaf-19,VanDev-16}. This explains the wide interest for a large-scale integration of quantum technologies. By inter-connecting spatially distributed quantum processors, we would achieve a scalable architecture resistant to noise. The general trend \cite{gambetta2020ibm, eddins2022doubling, liu2022design, gold2021entanglement, kielpinski2002architecture, stephenson2020high} shows a common belief in distributed \textcolor{black}{(and quasi-distributed, or multi-core)} architectures as physical substrate, allowing a modular and horizontal scale-up of computing resources, rather than relying on vertical scale-up, coming from single hardware advancements.
On the flip side, by linking distributed quantum processors, several new challenges arise \cite{Kim-08,PirBra-16,DurLamHeu-17,WehElkHan-18,Cas-18,cuomo2020towards}.

Being up-to-date with technologies currently under construction around the globe is a mandatory step to create a realistic system. We hence begin -- in Sec. \ref{sec:tech} -- by reviewing some of the most promising hardware technologies. This provide us (and the reader) with a realistic perspective of the fundamental components belonging a distributed quantum computing system. Once done that we propose -- in Sec. \ref{sec:stack} -- a full-stack development meant to be modular and prone to future changes. 
 The stack is indeed already an extension of one of our first proposal, available in \cite{cuomo2020towards}.

In Ch. \ref{ch:essentials}, we abstract from hardware technologies, by introducing the reader to the fundamental tools concerning distributed computing paradigms. We use the most affirmed language to express quantum computation, i.e. the standard quantum circuit model. We start with some basics on unitary \textit{synthesis} and \textit{decomposition}, important also for local computational. We then extend the subject to work on distributed architectures.

For the sake of completeness, In Ch. \ref{ch:noise} we provide a framework concerning quantum noise. Handling noise is probably, to date, the hardest challenge we are facing, as any model struggle to be scalable for hardware. Hence the aim of this chapter is to give a perspective of the magnitude of the problem, which -- sooner or later -- will be part of a distributed computing system. The framework is strengthen by experimental results.

We conclude the thesis with what is our main contribution to research. Thanks to the knowledge gathered by us throughout the aforementioned chapters, we could minutely investigate one essential component for any system implementation of practical value. Such a component is commonly referred as \textit{compiler}, which we formulate with mathematical rigor in Ch. \ref{ch:compile}.
With the support of some of the best tools for reasoning -- i.e. network optimization, circuit manipulation, group theory and \texttt{ZX}-calculus -- we managed to give a first complete model for the compilation problem on distributed architectures. Every feature characterizing the problem is treated minutely. We could separate the problem in several parts, each of which is tackled with a dedicated optimizer.

We also use our model as benchmarker for different network topologies. It is indeed important, in the development of a practical system, that the project follows a \textit{co-design} line, where each component is designed to fit at its best with the other components.

\newpage
\section{Technologies for distributed quantum computing}
\subsection{Stationary-flying transduction}
\label{sec:tech}

Qubit-qubit interaction generally works by means of some \textit{transducer}. A transducer can be seen as a physical interface ``converting quantum signals from one form of energy to another'' \cite{lauk2020perspectives}. It is especially true, in a distributed setting, that a transducer is able to move an information stored into some \textit{stationary qubit} -- e.g. a trapped-ion, a transmon or a quantum dot -- into some flying object, usually photons. A photon is therefore an information carrier or medium, able to cover a long distance. Therefore, the medium can be used to make distant qubits interact.

The ability to engineer efficient transducers allows us to rethink at quantum architectures as to be scalable and modular. Depending on the transducer \cite{schutz2017universal,lauk2020perspectives,brubaker2022optomechanical,mckenna2020cryogenic,zeuthen2020figures,kielpinski2002architecture,fang2014one,ZhoWanZou-20,KraRanHam-21,gold2021entanglement}, different kinds of distributed architecture arise. For the sake of understanding qubit-qubit interaction in a distributed setting, we now consider distributed ion-traps architectures.

Scaling up a single ion-trap is challenging \cite{moody20222022}. On the other hand, they represent a promising technology for integration within a distributed architecture, as a result of high gate fidelity \cite{harty2014high,ballance2014high} and long life-time \cite{wehner2018quantum,wang2021single}. In what comes next, we consider a cavity-based integration.

\subsubsection{Cavities and photon emission}
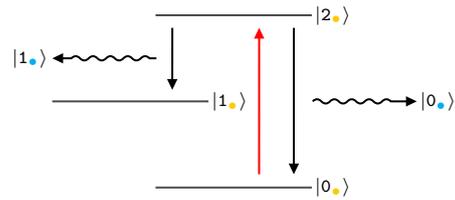
\begin{wrapfigure}{r}{7cm}
        \centering
        \input{tec-integration/figures/three-levels.tex}
        \caption{Simplified energy level structure of an ion ($\textcolor{ion}{\bullet}$) and relative photon emission ($\textcolor{cyan}{\bullet}$).}
        \label{fig:3-levels}
        \hrulefill
    \end{wrapfigure}
Considering the scenario of qubits stored on different processors, to couple them, the physical setting needs to \textit{scatter} quantum information outside a processor and reach the other one. This can be done by means of a single photon, canalized within an optical fiber.

In order to achieve such a configuration, we here consider ions able to be modeled as a three-levels system -- see Fig. \ref{fig:3-levels}. Such a system depicts the experimental set-up proposed in \cite{monroe2014large}\footnote{The interested reader can find other settings at Refs. \cite{gao2021optimisation,salmon2022gauge}.}. The specifics of the system comes from the \textit{ion species} selected to encode quantum information \cite{bruzewicz2019trapped}.
By placing such an ion within a \textit{cavity}, this creates an ion-cavity system, where now the ion interact with the {cavity mode}. The cavity has the role of collecting and scatter outside the system the photon emitted by the ion. Fig. \ref{fig:cavity} shows a pictorial representation of the ion-cavity able to do so.

    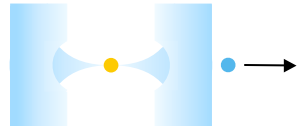
\begin{wrapfigure}{r}{7cm}
    \centering
    \input{tec-integration/figures/cavity.tex}
    \caption{Exemplary representation of an ion-cavity system emitting a photon. For technical details the reader may start from \cite{kaya2022cavity}.}
    \label{fig:cavity}
    \hrulefill
\end{wrapfigure}
The first step taking place is the excitation of the ion $\ket{\texttt{0}_{\textcolor{ion}{\bullet}}}\rightarrow \ket{\texttt{2}_{\textcolor{ion}{\bullet}}}$ -- i.e. the red arrow in Fig. \ref{fig:3-levels}. Ideally, a spontaneous decay of the ion brings its energy with equal probabilities to one of the two lowest (and computationally relevant) states -- i.e. $\ket{\texttt{0}}_{\textcolor{ion}{\bullet}}$ and $\ket{\texttt{1}_{\textcolor{ion}{\bullet}}}$. Furthermore, this happens with the emission of a photon which is coherent with the state of the ion. 


As we anticipated, the scattered photon can be canalized within an optical fiber; the final configuration of the ion-fiber system is in the superposition $\sfrac{1}{\sqrt{2}}(\ket{\texttt{0}_{\textcolor{ion}{\bullet}}\texttt{0}_{\textcolor{cyan}{\bullet}}} + \ket{\texttt{1}_{\textcolor{ion}{\bullet}}\texttt{1}_{\textcolor{cyan}{\bullet}}})$.



A pictorial representation of a single node of the distributed architecture is shown in Fig. \ref{fig:trap-cavity}. The cavity is pointing at one of the ions, which is coloured differently as an ion-trap may be composed by different ion species\footnote{Which brings to the classification in \textit{communication} and \textit{computation} (or data) qubits \cite{cuomo2021optimized,CacCalTaf-19,CalCac-20,cuomo2020towards}.}, depending on whether it is meant to perform computation or communication.


To achieve the \textit{non-local coupling} we need to consider two ion-traps generating and distributing the entanglement (at the same time), after which, a protocol called \textit{entanglement swapping} completes the process. A \textit{control system} will take care of accomplishing the task.

\begin{figure}[h]
    \centering
    \includegraphics[scale=0.24]{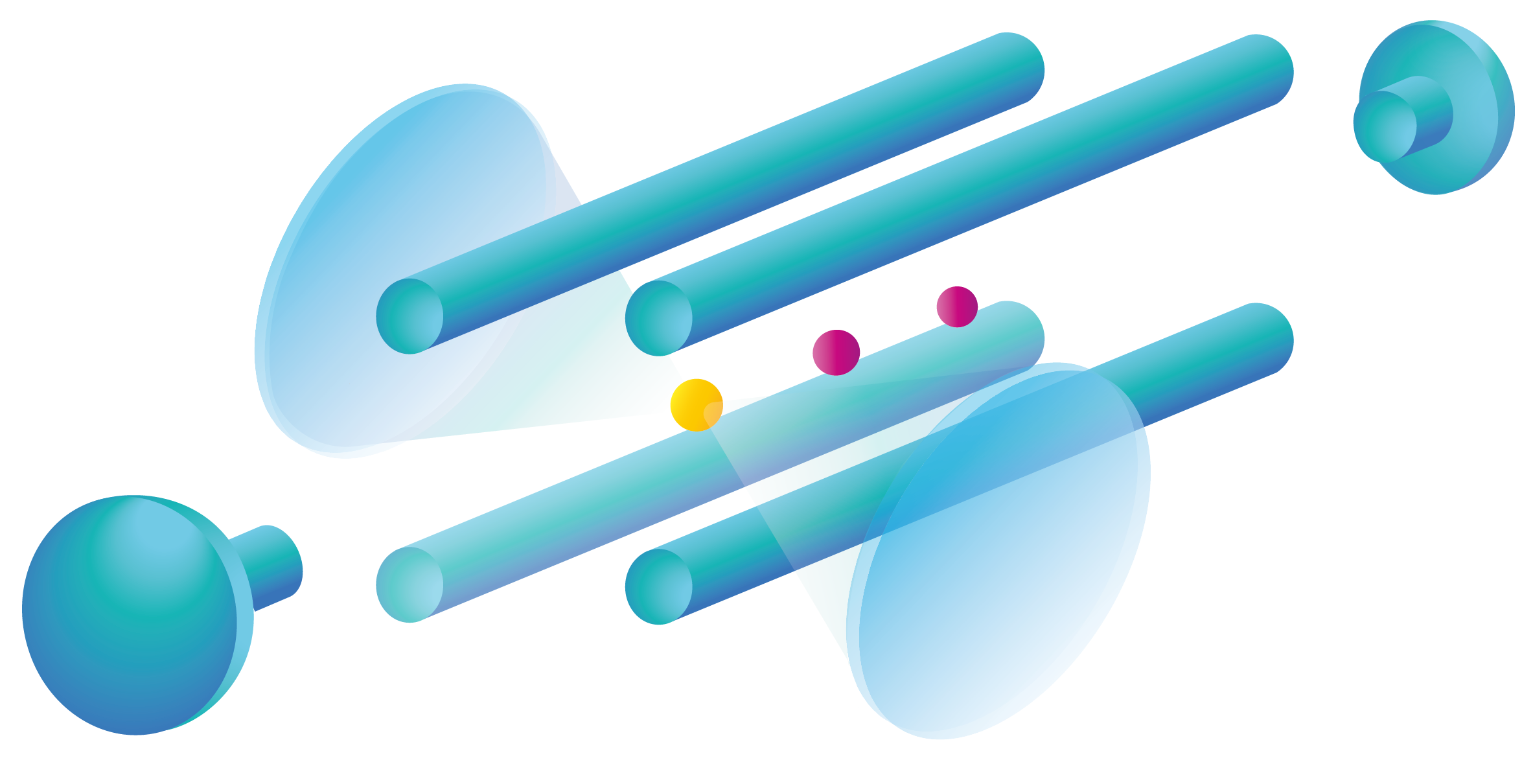}
    \caption{An ion-trap embedded with a cavity pointing at a single ion. Representation inspired by a linear design \cite{bruzewicz2019trapped,paul1990electromagnetic}.}
    \label{fig:trap-cavity}
\end{figure}


\subsection{Control system}
\label{sec:control-sys}
Summing up, to establish entanglement, the ions are simultaneously excited to an electronic state that spontaneously decays, during which a single photon each is emitted whose polarization is entangled with the ion's internal state. These photons are collected into optical fibres using free-space optics and sent to a common \textit{terminal}. 

The terminal take care of detecting the photons by means of a probabilistic \textit{Bell state measurement}. This projects the ions into a maximally entangled state, heralded by the coincident detection of the pair of photons -- see Sec. \ref{sec:analyser} for details. This is commonly referred as \textit{entanglement swapping}.

It is important that the ion-traps are synchronized, so that the photon reach the terminal at the same time\footnote{Without synchronization, one can retrieve the likelihood of each photon source. This makes the photons distinguishable, causing a loss in fidelity \cite{stephenson2019entanglement}.}. Classical synchronization protocol would take care of this by means of a master-clock. An experimental settings is available in Ref. \cite{nadlinger2022experimental}.

Ultimately, to achieve scalability we need to consider the case of several processors. In the most basic scenario, all the processors are \textit{centralized}, in the sense the all of them are wired to a common terminal to perform bell state measurement. Such a setting is the first example of scalable distributed architectures. The problems arising from such a setting is a \textit{scheduling} problem. A multiplexer taking deterministic choices would be enough to ensure that all the processors are carefully scheduled to not create overlaps. An experimental settings, where four processors are scheduled by means of a multiplexer, is available in Ref. \cite{oi2006scalable}.

\subsection{Bell state analyser}
\label{sec:analyser}

To keep the discussion easy we explained the main procedures by means of three-levels systems for the ions. However, each of this state should be split to create several possible configuration. This doesn't change the whole protocol, but it has consequences on the possible outcomes obtained by measuring the photons. Different ion species and kind of measurement lead to different configurations. In general, we distinguish three main outcomes:
\begin{itemize}
    \item One or both photons failed to be detected by the measurement. This means that the whole procedure is basically wasted time, as the nodes need to attempt again from scratch. This possibility can be a serious problem to practical computation, as a low \textit{success rate} leads to long waiting times.
    \item The protocol succeeded and the ions are in one of the four Bell states $\{\ket{\Phi^{\texttt{+}}},\ket{\Phi^{\texttt{-}}},\ket{\Psi^{\texttt{+}}},\ket{\Psi^{\texttt{-}}}\}$.
    \item The protocol partially succeeded. Namely, a superposition between two bell states has been created, e.g. $\ket{\Phi^{\pm}}$ or $\ket{\Psi^{\pm}}$. This scenario may occur for several reasons -- depending also from the employed physical settings -- it can be caused by \textit{dark measurements} \cite{stephenson2019entanglement}, a rare and negligible scenario. Otherwise, it may come, for example, from two \textit{clicks} coming from the same detector.
\end{itemize}

\begin{wrapfigure}{r}{5.1cm}
    \centering
    \input{tec-integration/figures/bsm.tex}
    \caption{A common setting for the Bell state measurement \cite{lutkenhaus1999bell}.}
    \label{fig:bsm}
    \hrulefill
\end{wrapfigure}
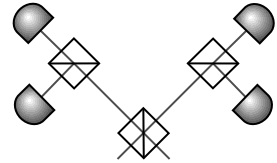
Different Bell measurement settings brings to different sets of heralded entanglements \cite{stephenson2020high,valivarthi2014efficient,lutkenhaus1999bell, mattle1996dense}. We here consider the quite general case of Ref. \cite{lutkenhaus1999bell}, as we think this case may be particularly efficient for distributed computation. A pictorial representation of the setting is reported in Fig. \ref{fig:bsm}. In fact such a configuration brings to a 50\% chance of success -- i.e. two clicks on different detectors -- and 50\% of partial success -- i.e. two clicks on the same detector. Namely, when the protocol succeeds, the final state is in $\{\ket{\Psi^{\texttt{-}}},\ket{\Psi^{\texttt{+}}}\}$, while in case of partial success, the output has an ambiguous phase: $\ket{\Phi^{\pm}}$.


The ideal result coming from performing entanglement generation and distribution followed by entanglement swapping is a maximally entangled state between distant qubits. In practice, this is not achievable as each of the complicated techniques we described are in general not perfect, resulting in a state slightly different from a Bell pair. One can evaluate the final distributed state in terms of \textit{fidelity} with some target Bell state. E.g.,
\begin{equation}
    \mathfrak{f} = \bra{\Phi^{\texttt{+}}}\sigma\ket{\Phi^{\texttt{+}}}.
\end{equation}
Where $\sigma$ is the generated state. For example, consider the experiments reported in Ref. \cite{lutkenhaus1999bell,stephenson2020high}. The author's proposal starts with the generation of a non-maximally entangled state ion-photon
\begin{equation}
    \sqrt{\frac{2}{3}}\ket{\texttt{0}_{\textcolor{ion}{\bullet}}\texttt{0}_{\textcolor{cyan}{\bullet}}} + \sqrt{\frac{1}{3}}\ket{\texttt{1}_{\textcolor{ion}{\bullet}}\texttt{1}_{\textcolor{cyan}{\bullet}}}.
\end{equation}
However, after the collection into the single mode fiber, the state gets projected to the Bell pair
\begin{equation}
    \frac{1}{\sqrt{2}}(\ket{\texttt{0}_{\textcolor{ion}{\bullet}}\texttt{0}_{\textcolor{cyan}{\bullet}}} + \ket{\texttt{1}_{\textcolor{ion}{\bullet}}\texttt{1}_{\textcolor{cyan}{\bullet}}}).
\end{equation}

Once performed the entanglement swap, the ion-ion average fidelity is $0.94$; a promising result. 
Unfortunately, in the perspective of practical computation, the fidelity needs to be some $\mathfrak{f} = 1 - \varepsilon$ with $\varepsilon$ small enough to keep the error rate \textit{manageable}, e.g. by means of \textit{error correction schemes} -- treated in Ch. \ref{ch:noise}. A possible solution is called \textit{entanglement distillation} \cite{rozpkedek2018optimizing,kalb2017entanglement,hu2021long} (or purification). However, choosing the best approach is not trivial and may very depend on the architecture specifics.



\section{Envisioning the full system}
\label{sec:stack}
We reported several important research fields deeply related to the implementation of a first distributed and scalable architecture. We introduced the required technologies to achieve qubit-qubit interaction when these are arbitrarily far apart.

\begin{wrapfigure}{r}{7.6cm}
    \centering
    \includegraphics[scale = 0.33]{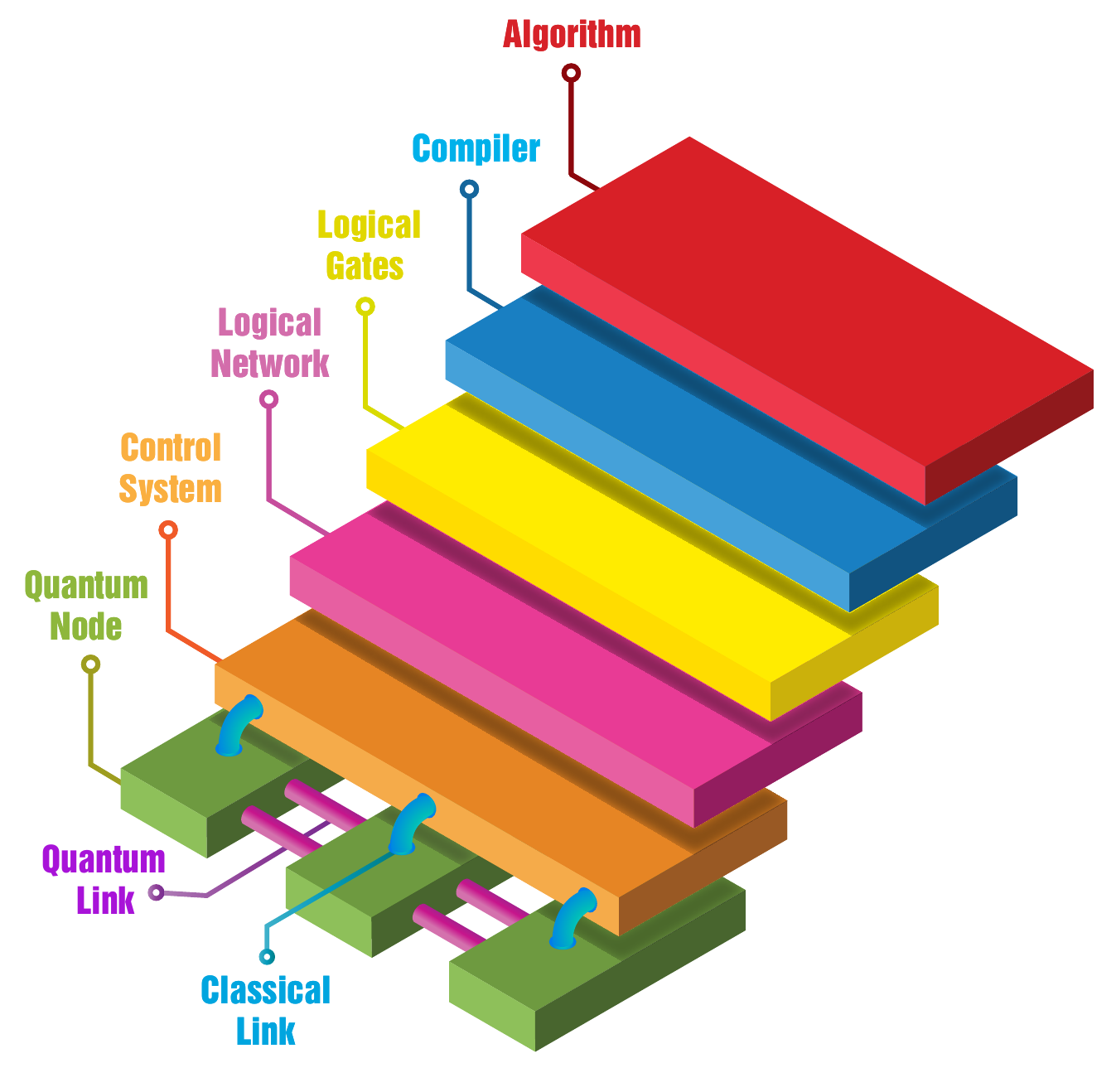}
    \caption{Full-stack development of a distributed quantum computing framework.}
    \label{fig:layers}
    \hrulefill
\end{wrapfigure}
The considered literature gives a perspective on how to integrate multiple quantum processors into a scalable architecture, able to perform distributed quantum computation.

Stemming from the above overview, we now need to extract a full-stack development, by identifying the most important roles -- and dependencies -- and we will need to engineer an \textit{ecosystem} \cite{cuomo2020towards}, providing a framework for distributed quantum computation. As we are facing the early stage of quantum computation and distributed architecture, it is wise to focus on the main challenges and assigning them to a proper \textit{entity}\footnote{As usual in the engineer terminology, an entity is something quite abstract, which needs to be defined in terms of its roles and relations with other entities.}. In fact, any proposal now is highly prone to changes, because of the continuous growing of the field \cite{gibney2019quantum} and the huge advancing in both technology and information theory.

We propose a design that starts from a bottom-up reasoning, stacking up a number of layers where the lower ones provide some \textit{resources} and flexibility which the upper layers can rely on. More precisely, consider Fig. \ref{fig:layers}: this shows a linear stack where each layer represents one of the fundamental subjects necessary to create a practical framework. 

\subsubsection{Level 1}
Not surprisingly, the first layer is a mere pictorial representation of the distributed hardware. For the sake of clarity, we depicted a network composed by three quantum nodes. We already discussed how such a network may be achieved by focusing on ion-traps integrated with cavities, Bell state analysers and multiplexers -- see Secs. \ref{sec:tech} and \ref{sec:analyser}. 

A quantum node refers to the full hardware set-up working locally, which includes, of course, the quantum processor. The nodes are inter-connected by \textbf{quantum links} wherein mediums carries quantum information.
Such a set-up can be achieved, for example, by means of \cite{oi2006scalable}:
\begin{itemize}
    \item ion-traps, cavity-based transducers as quantum nodes and
    \item optical fibers, multiplexers and bell-state analysers as quantum links.
\end{itemize}


\subsubsection{Level 2}
As explained in Sec. \ref{sec:control-sys}, as minimum requirement for the system to be \textit{operative}, the network needs to be carefully handled by a classical \textbf{control system}, which cares about synchronization and scheduling in real time.
Because of its role, the control system is mainly physical. In fact, it is directly connected to each quantum node by means of \textbf{classical links}. The linkage will be also used to gather classical information coming from \textit{measurement-based computation} \cite{nielsen2003quantum}.

Advancements in the physical network\footnote{Up to the development of a \textit{quantum internet} \cite{Kim-08,WehElkHan-18,van2022quantum,CacCalTaf-19,cuomo2020towards,KozWehVan-21,DurLamHeu-17,pant2019routing}} will allow to \textit{evolve} the layer, up to becoming a \textit{quantum control system} -- as envisioned for example in Refs. \cite{koudia2022deep,illiano2022quantum,avis2022analysis,wallnofer2019multipartite} --, a quantum control system will be able to optimize the efficiency of the network, having access to a wider spectrum of resources -- e.g. high-dimensional entangled states -- which can be manipulated to optimize the network efficiency, by means of communication protocol fundamentally based on quantum communication theory. A quantum control system may also be responsible for the definition of a (distributed) error correction scheme -- treated in Ch. \ref{ch:noise}.

\subsubsection{Level 3}
The control system deals with \textit{real time tasks} and, if engineered properly, it is able to guarantee a \textbf{logical network} with a time domain that can be \textit{discrete}. Specifically, a control system provides a \textit{topology} running in well-defined \textit{time slots}. This means that any kind of unexpected delay is negligible to the upper-layers. As central part of the stack, generating a logical network is a critical step. In fact, it is the layer where real-time tasks meet logical computation. Hence, a logical network should provide an array of logical resources which are meaningful to computation.

\subsubsection{Level 4}
As said above, the logical network provides an array of logical resources. Because of their importance a layer is dedicated to represent such resources. The set of \textbf{logical gates} will be the fundamental components upon which building a \textit{computational paradigm}, specifically designed to work on a server farm.

Logical gates can be local -- e.g. multi-qubit gates \cite{lu2019global,casanova2012quantum,ivanov2015efficient,martinez2016compiling} -- and non-local -- i.e. telegates --.  In fact, while local gates operates on logical nodes, telegates operates throughout the logical network. This makes the latter resource much more expensive, which deserve a dedicated analysis, especially considering the degree of novelty in the context of computational paradigms that may arise. Refer to Sec. \ref{sec:e-comp} for details.

\subsubsection{Level 5}
We entitled this layer \textbf{compiler} as its role is highly related to the already well-established branch of research in the context of local quantum computation \cite{madden2022best,hillmich2021exploiting,burgholzer2022limiting, MasFalMos-08,SirSanCol-18,WilBurZul-19,LiDinXie-19,ZulWil-19,ItoRayIma-19,ZhaZheZha-20,KarTezPet-20,MorParRes-21,MarMorRoc-21, BooDoBec-18,FerAmo-21}. A compiler has the crucial role to mask all the underlying stack to an \textit{algorithm designer}. In fact an algorithm is generally written to solve some problem which goes beyond the architecture meant to process it. This means that the designer works in an agnostic fashion, without considering the constraints coming from the stack. Which is why a compiler is the final necessary optimizer, able to transform an abstract algorithm into a logical algorithm, compliant with the logical resources given by the bottom layers. Our research project mainly focus on the \textit{standardization} of such a layer and we detailed our efforts in Ch. \ref{ch:compile}.

\subsubsection{Level 6}
The upper layer is simply an \textbf{algorithm}, an abstract input, which the framework takes charge of and carefully spread throughout the whole stack in order to be processed. 

Once the framework is ready-to-go, it will be able to accept some \textit{groups of algorithms}\footnote{Parallel-based algorithms -- e.g. see Ref. \cite{cleve2000fast} for a parallel algorithm solving the quantum Fourier transform -- are natively meant to work on distributed architectures.}, up to \textit{universal} groups. Refer to Ch. \ref{ch:essentials} for details.

\printbibliography[title=References,heading=subbibintoc]
\end{refsection}

%% file: tec-integration/figures/three-levels.tex
\tikzset{every picture/.style={line width=0.75pt}} 

\begin{tikzpicture}[x=0.65pt,y=0.65pt,yscale=-1,xscale=1]

\draw [color={rgb, 255:red, 74; green, 74; blue, 74 }  ,draw opacity=1 ][line width=0.75]    (210.33,300) -- (300.33,300) ;
\draw [color={rgb, 255:red, 74; green, 74; blue, 74 }  ,draw opacity=1 ][line width=0.75]    (210.33,200) -- (300.33,200) ;
\draw [color={rgb, 255:red, 74; green, 74; blue, 74 }  ,draw opacity=1 ][line width=0.75]    (150.83,250) -- (240.83,250) ;
\draw    (220,207.33) -- (219.97,240.33) ;
\draw [shift={(219.97,243.33)}, rotate = 270.05] [fill={rgb, 255:red, 0; green, 0; blue, 0 }  ][line width=0.08]  [draw opacity=0] (6.25,-3) -- (0,0) -- (6.25,3) -- cycle    ;
\draw [color={rgb, 255:red, 255; green, 0; blue, 0 }  ,draw opacity=1 ]   (270,210.33) -- (269.97,292.33) ;
\draw [shift={(270,207.33)}, rotate = 90.02] [fill={rgb, 255:red, 255; green, 0; blue, 0 }  ,fill opacity=1 ][line width=0.08]  [draw opacity=0] (6.25,-3) -- (0,0) -- (6.25,3) -- cycle    ;
\draw    (290,207.33) -- (289.97,289.33) ;
\draw [shift={(289.97,292.33)}, rotate = 270.02] [fill={rgb, 255:red, 0; green, 0; blue, 0 }  ][line width=0.08]  [draw opacity=0] (6.25,-3) -- (0,0) -- (6.25,3) -- cycle    ;
\draw    (300.83,250) .. controls (302.5,248.33) and (304.16,248.33) .. (305.83,250) .. controls (307.5,251.67) and (309.16,251.67) .. (310.83,250) .. controls (312.5,248.33) and (314.16,248.33) .. (315.83,250) .. controls (317.5,251.67) and (319.16,251.67) .. (320.83,250) .. controls (322.5,248.33) and (324.16,248.33) .. (325.83,250) .. controls (327.5,251.67) and (329.16,251.67) .. (330.83,250) .. controls (332.5,248.33) and (334.16,248.33) .. (335.83,250) .. controls (337.5,251.67) and (339.16,251.67) .. (340.83,250) .. controls (342.5,248.33) and (344.16,248.33) .. (345.83,250) -- (349.83,250) -- (357.83,250) ;
\draw [shift={(360.83,250)}, rotate = 180] [fill={rgb, 255:red, 0; green, 0; blue, 0 }  ][line width=0.08]  [draw opacity=0] (6.25,-3) -- (0,0) -- (6.25,3) -- cycle    ;
\draw    (153.83,225) -- (161.83,225) .. controls (163.5,223.33) and (165.16,223.33) .. (166.83,225) .. controls (168.5,226.67) and (170.16,226.67) .. (171.83,225) .. controls (173.5,223.33) and (175.16,223.33) .. (176.83,225) .. controls (178.5,226.67) and (180.16,226.67) .. (181.83,225) .. controls (183.5,223.33) and (185.16,223.33) .. (186.83,225) .. controls (188.5,226.67) and (190.16,226.67) .. (191.83,225) .. controls (193.5,223.33) and (195.16,223.33) .. (196.83,225) .. controls (198.5,226.67) and (200.16,226.67) .. (201.83,225) .. controls (203.5,223.33) and (205.16,223.33) .. (206.83,225) -- (210.83,225) -- (210.83,225) ;
\draw [shift={(150.83,225)}, rotate = 0] [fill={rgb, 255:red, 0; green, 0; blue, 0 }  ][line width=0.08]  [draw opacity=0] (6.25,-3) -- (0,0) -- (6.25,3) -- cycle    ;

\draw (300.33,300) node [anchor=west] [inner sep=0.75pt]    {$_{|\texttt{0}_{\textcolor{ion}{\bullet}}\rangle}$};
\draw (300.33,200) node [anchor=west] [inner sep=0.75pt]    {$_{|\texttt{2}_{\textcolor{ion}{\bullet}}\rangle}$};
\draw (240.83,250) node [anchor=west] [inner sep=0.75pt]    {$_{|\texttt{1}_{\textcolor{ion}{\bullet}}\rangle}$};
\draw (150.83,225) node [anchor=east] [inner sep=0.75pt]    {$_{|\texttt{1}_{\textcolor{cyan}{\bullet}}\rangle}$};
\draw (360.83,250) node [anchor=west] [inner sep=0.75pt]    {$_{|\texttt{0}_{\textcolor{cyan}{\bullet}}\rangle}$};

\end{tikzpicture}

%% file: tec-integration/figures/cavity.tex
  
\tikzset {_m63vhi5g5/.code = {\pgfsetadditionalshadetransform{ \pgftransformshift{\pgfpoint{0 bp } { 0 bp }  }  \pgftransformrotate{0 }  \pgftransformscale{2 }  }}}
\pgfdeclarehorizontalshading{_tutiilkno}{150bp}{rgb(0bp)=(0.63,0.86,1);
rgb(40bp)=(0.63,0.86,1);
rgb(49.25bp)=(0.8,0.92,1);
rgb(62.5bp)=(0.94,0.98,1);
rgb(100bp)=(0.94,0.98,1)}

  
\tikzset {_uio1of5eq/.code = {\pgfsetadditionalshadetransform{ \pgftransformshift{\pgfpoint{0 bp } { 0 bp }  }  \pgftransformrotate{0 }  \pgftransformscale{2 }  }}}
\pgfdeclarehorizontalshading{_j7mwgzsga}{150bp}{rgb(0bp)=(0.63,0.86,1);
rgb(40bp)=(0.63,0.86,1);
rgb(49.25bp)=(0.8,0.92,1);
rgb(62.5bp)=(0.94,0.98,1);
rgb(100bp)=(0.94,0.98,1)}

  
\tikzset {_l792xrlh4/.code = {\pgfsetadditionalshadetransform{ \pgftransformshift{\pgfpoint{0 bp } { 0 bp }  }  \pgftransformrotate{0 }  \pgftransformscale{2 }  }}}
\pgfdeclarehorizontalshading{_s9c60vty4}{150bp}{rgb(0bp)=(0.63,0.86,1);
rgb(40bp)=(0.63,0.86,1);
rgb(49.25bp)=(0.8,0.92,1);
rgb(62.5bp)=(0.94,0.98,1);
rgb(100bp)=(0.94,0.98,1)}

  
\tikzset {_j48dzssom/.code = {\pgfsetadditionalshadetransform{ \pgftransformshift{\pgfpoint{0 bp } { 0 bp }  }  \pgftransformrotate{0 }  \pgftransformscale{2 }  }}}
\pgfdeclarehorizontalshading{_hfeyozird}{150bp}{rgb(0bp)=(0.63,0.86,1);
rgb(40bp)=(0.63,0.86,1);
rgb(49.25bp)=(0.8,0.92,1);
rgb(62.5bp)=(0.94,0.98,1);
rgb(100bp)=(0.94,0.98,1)}

  
\tikzset {_27bb4olq4/.code = {\pgfsetadditionalshadetransform{ \pgftransformshift{\pgfpoint{0 bp } { 0 bp }  }  \pgftransformrotate{0 }  \pgftransformscale{2 }  }}}
\pgfdeclarehorizontalshading{_in77jy1o9}{150bp}{rgb(0bp)=(0.94,0.98,1);
rgb(37.5bp)=(0.94,0.98,1);
rgb(49.25bp)=(0.8,0.92,1);
rgb(62.5bp)=(0.63,0.86,1);
rgb(100bp)=(0.63,0.86,1)}

  
\tikzset {_elp34n4bv/.code = {\pgfsetadditionalshadetransform{ \pgftransformshift{\pgfpoint{0 bp } { 0 bp }  }  \pgftransformrotate{0 }  \pgftransformscale{2 }  }}}
\pgfdeclarehorizontalshading{_5fm5dxd46}{150bp}{rgb(0bp)=(0.94,0.98,1);
rgb(37.5bp)=(0.94,0.98,1);
rgb(49.25bp)=(0.8,0.92,1);
rgb(62.5bp)=(0.63,0.86,1);
rgb(100bp)=(0.63,0.86,1)}

  
\tikzset {_1wo9vzrg3/.code = {\pgfsetadditionalshadetransform{ \pgftransformshift{\pgfpoint{0 bp } { 0 bp }  }  \pgftransformrotate{0 }  \pgftransformscale{2 }  }}}
\pgfdeclarehorizontalshading{_likqzrrjj}{150bp}{rgb(0bp)=(0.94,0.98,1);
rgb(37.5bp)=(0.94,0.98,1);
rgb(49.25bp)=(0.8,0.92,1);
rgb(62.5bp)=(0.63,0.86,1);
rgb(100bp)=(0.63,0.86,1)}

  
\tikzset {_fleysy0s0/.code = {\pgfsetadditionalshadetransform{ \pgftransformshift{\pgfpoint{0 bp } { 0 bp }  }  \pgftransformrotate{0 }  \pgftransformscale{2 }  }}}
\pgfdeclarehorizontalshading{_gzmijusxk}{150bp}{rgb(0bp)=(0.94,0.98,1);
rgb(37.5bp)=(0.94,0.98,1);
rgb(49.25bp)=(0.8,0.92,1);
rgb(62.5bp)=(0.63,0.86,1);
rgb(100bp)=(0.63,0.86,1)}
\tikzset{every picture/.style={line width=0.75pt}} 

\begin{tikzpicture}[x=0.75pt,y=0.75pt,yscale=-1,xscale=1]

\draw [color={rgb, 255:red, 0; green, 0; blue, 0 }  ,draw opacity=1 ][line width=0.75]    (596.25,259.73) -- (616.65,259.93) ;
\draw  [draw opacity=0][shading=_tutiilkno,_m63vhi5g5] (505.15,251.82) .. controls (505.15,251.82) and (505.15,251.82) .. (505.15,251.82) .. controls (500.65,247.33) and (493.31,247.39) .. (488.74,251.96) .. controls (484.17,256.53) and (484.11,263.87) .. (488.6,268.37) .. controls (493.09,272.86) and (500.44,272.79) .. (505.01,268.22) .. controls (510.52,262.71) and (518.74,259.9) .. (529.68,259.81) .. controls (518.74,259.9) and (510.57,257.24) .. (505.15,251.82) -- cycle ;
\draw  [draw opacity=0][shading=_j7mwgzsga,_uio1of5eq] (498.05,228.74) .. controls (494.9,230.35) and (491.94,232.47) .. (489.31,235.11) .. controls (475.64,248.78) and (475.79,271.08) .. (489.63,284.93) .. controls (492.14,287.43) and (494.93,289.49) .. (497.89,291.1) -- (507.96,272.2) .. controls (506.82,271.56) and (505.74,270.76) .. (504.77,269.78) .. controls (499.28,264.3) and (499.14,255.55) .. (504.45,250.25) .. controls (505.49,249.2) and (506.67,248.37) .. (507.92,247.75) -- cycle ;
\draw  [draw opacity=0][shading=_s9c60vty4,_l792xrlh4] (479.41,290.62) -- (496.21,290.62) -- (496.21,247.42) -- (507.92,247.42) -- (507.92,228.82) -- (479.41,228.82) -- cycle ;
\draw  [draw opacity=0][shading=_hfeyozird,_j48dzssom] (479.41,230.1) -- (496.21,230.1) -- (496.21,273.3) -- (507.92,273.3) -- (507.92,291.1) -- (479.41,291.1) -- cycle ;
\draw  [draw opacity=0][shading=_in77jy1o9,_27bb4olq4] (554.32,251.84) .. controls (554.32,251.84) and (554.32,251.84) .. (554.32,251.84) .. controls (554.32,251.84) and (554.32,251.84) .. (554.32,251.84) .. controls (558.81,247.35) and (566.15,247.41) .. (570.72,251.98) .. controls (575.29,256.55) and (575.35,263.89) .. (570.86,268.39) .. controls (566.37,272.88) and (559.03,272.82) .. (554.46,268.25) .. controls (548.95,262.74) and (540.72,259.93) .. (529.79,259.84) .. controls (540.72,259.93) and (548.9,257.26) .. (554.32,251.84) -- cycle ;
\draw  [draw opacity=0][fill={rgb, 255:red, 253; green, 202; blue, 2 }  ,fill opacity=1 ] (533.46,259.84) .. controls (533.46,261.87) and (531.82,263.51) .. (529.79,263.51) .. controls (527.76,263.51) and (526.11,261.87) .. (526.11,259.84) .. controls (526.11,257.81) and (527.76,256.16) .. (529.79,256.16) .. controls (531.82,256.16) and (533.46,257.81) .. (533.46,259.84) -- cycle ;
\draw  [draw opacity=0][shading=_5fm5dxd46,_elp34n4bv] (561.41,228.77) .. controls (564.56,230.38) and (567.51,232.5) .. (570.15,235.13) .. controls (583.81,248.79) and (583.67,271.1) .. (569.83,284.94) .. controls (567.25,287.52) and (564.38,289.63) .. (561.32,291.25) -- (551.52,272.5) .. controls (552.72,271.84) and (553.86,271) .. (554.88,269.98) .. controls (560.46,264.4) and (560.6,255.49) .. (555.19,250.09) .. controls (554.13,249.03) and (552.94,248.18) .. (551.66,247.55) -- cycle ;
\draw  [draw opacity=0][shading=_likqzrrjj,_1wo9vzrg3] (580.06,290.64) -- (563.26,290.64) -- (563.26,247.44) -- (551.54,247.44) -- (551.54,228.84) -- (580.06,228.84) -- cycle ;
\draw  [draw opacity=0][shading=_gzmijusxk,_fleysy0s0] (580.04,229.3) -- (563.24,229.3) -- (563.24,272.5) -- (551.52,272.5) -- (551.52,291) -- (580.04,291) -- cycle ;
\draw  [draw opacity=0][fill={rgb, 255:red, 0; green, 0; blue, 0 }  ,fill opacity=1 ] (616.03,256.64) -- (616.03,256.64) -- (622.37,260.06) -- (616.03,263.48) -- (616.03,263.48) -- cycle ;
\draw  [draw opacity=0][fill={rgb, 255:red, 83; green, 183; blue, 236 }  ,fill opacity=1 ] (590.85,257.13) .. controls (592.28,258.56) and (592.28,260.89) .. (590.85,262.32) .. controls (589.41,263.76) and (587.09,263.76) .. (585.65,262.32) .. controls (584.22,260.89) and (584.22,258.56) .. (585.65,257.13) .. controls (587.09,255.69) and (589.41,255.69) .. (590.85,257.13) -- cycle ;

\end{tikzpicture}

%% file: tec-integration/figures/bsm.tex
  
\tikzset {_vr3acslw4/.code = {\pgfsetadditionalshadetransform{ \pgftransformshift{\pgfpoint{0 bp } { 0 bp }  }  \pgftransformrotate{0 }  \pgftransformscale{2 }  }}}
\pgfdeclarehorizontalshading{_qwjpneii6}{150bp}{rgb(0bp)=(0.94,0.98,1);
rgb(37.5bp)=(0.94,0.98,1);
rgb(49.25bp)=(0.8,0.92,1);
rgb(62.5bp)=(0.63,0.86,1);
rgb(100bp)=(0.63,0.86,1)}
\tikzset{_c4uexefq8/.code = {\pgfsetadditionalshadetransform{\pgftransformshift{\pgfpoint{0 bp } { 0 bp }  }  \pgftransformrotate{0 }  \pgftransformscale{2 } }}}
\pgfdeclarehorizontalshading{_1d7srl9yt} {150bp} {color(0bp)=(transparent!99);
color(37.5bp)=(transparent!99);
color(49.25bp)=(transparent!50);
color(62.5bp)=(transparent!44.99999999999999);
color(100bp)=(transparent!44.99999999999999) } 
\pgfdeclarefading{_qwlcmtedi}{\tikz \fill[shading=_1d7srl9yt,_c4uexefq8] (0,0) rectangle (50bp,50bp); } 

  
\tikzset {_kqwhemy0x/.code = {\pgfsetadditionalshadetransform{ \pgftransformshift{\pgfpoint{89.1 bp } { -108.9 bp }  }  \pgftransformscale{1.32 }  }}}
\pgfdeclareradialshading{_bgjffjk90}{\pgfpoint{-72bp}{88bp}}{rgb(0bp)=(1,1,1);
rgb(0bp)=(1,1,1);
rgb(25bp)=(0,0,0);
rgb(400bp)=(0,0,0)}

  
\tikzset {_b6z0jvn0g/.code = {\pgfsetadditionalshadetransform{ \pgftransformshift{\pgfpoint{0 bp } { 0 bp }  }  \pgftransformrotate{0 }  \pgftransformscale{2 }  }}}
\pgfdeclarehorizontalshading{_scwsrws83}{150bp}{rgb(0bp)=(0.94,0.98,1);
rgb(37.5bp)=(0.94,0.98,1);
rgb(49.25bp)=(0.8,0.92,1);
rgb(62.5bp)=(0.63,0.86,1);
rgb(100bp)=(0.63,0.86,1)}
\tikzset{_6u8z2oi3f/.code = {\pgfsetadditionalshadetransform{\pgftransformshift{\pgfpoint{0 bp } { 0 bp }  }  \pgftransformrotate{0 }  \pgftransformscale{2 } }}}
\pgfdeclarehorizontalshading{_w6o7t3gc8} {150bp} {color(0bp)=(transparent!99);
color(37.5bp)=(transparent!99);
color(49.25bp)=(transparent!50);
color(62.5bp)=(transparent!44.99999999999999);
color(100bp)=(transparent!44.99999999999999) } 
\pgfdeclarefading{_ga7qz4u6p}{\tikz \fill[shading=_w6o7t3gc8,_6u8z2oi3f] (0,0) rectangle (50bp,50bp); } 

  
\tikzset {_gg3nzbb1r/.code = {\pgfsetadditionalshadetransform{ \pgftransformshift{\pgfpoint{89.1 bp } { -108.9 bp }  }  \pgftransformscale{1.32 }  }}}
\pgfdeclareradialshading{_ydlslgru7}{\pgfpoint{-72bp}{88bp}}{rgb(0bp)=(1,1,1);
rgb(0bp)=(1,1,1);
rgb(25bp)=(0,0,0);
rgb(400bp)=(0,0,0)}

  
\tikzset {_rcaruaq5p/.code = {\pgfsetadditionalshadetransform{ \pgftransformshift{\pgfpoint{89.1 bp } { -108.9 bp }  }  \pgftransformscale{1.32 }  }}}
\pgfdeclareradialshading{_ip4gx3285}{\pgfpoint{-72bp}{88bp}}{rgb(0bp)=(1,1,1);
rgb(0bp)=(1,1,1);
rgb(25bp)=(0,0,0);
rgb(400bp)=(0,0,0)}

  
\tikzset {_0qzuw01gm/.code = {\pgfsetadditionalshadetransform{ \pgftransformshift{\pgfpoint{0 bp } { 0 bp }  }  \pgftransformrotate{0 }  \pgftransformscale{2 }  }}}
\pgfdeclarehorizontalshading{_dr3uqyfwt}{150bp}{rgb(0bp)=(0.94,0.98,1);
rgb(37.5bp)=(0.94,0.98,1);
rgb(49.25bp)=(0.8,0.92,1);
rgb(62.5bp)=(0.63,0.86,1);
rgb(100bp)=(0.63,0.86,1)}
\tikzset{_h6rtsectd/.code = {\pgfsetadditionalshadetransform{\pgftransformshift{\pgfpoint{0 bp } { 0 bp }  }  \pgftransformrotate{0 }  \pgftransformscale{2 } }}}
\pgfdeclarehorizontalshading{_gr5pe2u4w} {150bp} {color(0bp)=(transparent!99);
color(37.5bp)=(transparent!99);
color(49.25bp)=(transparent!50);
color(62.5bp)=(transparent!44.99999999999999);
color(100bp)=(transparent!44.99999999999999) } 
\pgfdeclarefading{_r7x9si0lr}{\tikz \fill[shading=_gr5pe2u4w,_h6rtsectd] (0,0) rectangle (50bp,50bp); } 

  
\tikzset {_v25ihwrd4/.code = {\pgfsetadditionalshadetransform{ \pgftransformshift{\pgfpoint{89.1 bp } { -108.9 bp }  }  \pgftransformscale{1.32 }  }}}
\pgfdeclareradialshading{_bfporumwl}{\pgfpoint{-72bp}{88bp}}{rgb(0bp)=(1,1,1);
rgb(0bp)=(1,1,1);
rgb(25bp)=(0,0,0);
rgb(400bp)=(0,0,0)}
\tikzset{every picture/.style={line width=0.75pt}} 

\begin{tikzpicture}[x=0.75pt,y=0.75pt,yscale=-1,xscale=1]

\draw [color={rgb, 255:red, 0; green, 0; blue, 0 }  ,draw opacity=0.7 ]   (567.39,224.84) -- (519.67,176.77) ;
\draw [color={rgb, 255:red, 0; green, 0; blue, 0 }  ,draw opacity=0.7 ]   (589.07,176.77) -- (610.47,199.01) ;
\draw [color={rgb, 255:red, 0; green, 0; blue, 0 }  ,draw opacity=0.7 ]   (589.07,176.77) -- (611.07,155.02) ;
\draw [color={rgb, 255:red, 0; green, 0; blue, 0 }  ,draw opacity=0.7 ]   (541.39,224.84) -- (589.07,176.77) ;
\path  [shading=_qwjpneii6,_vr3acslw4,path fading= _qwlcmtedi ,fading transform={xshift=2}] (554.29,199.27) -- (566.93,211.89) -- (554.49,224.41) -- (541.85,211.79) -- cycle ; 
 \draw   (554.29,199.27) -- (566.93,211.89) -- (554.49,224.41) -- (541.85,211.79) -- cycle ; 

\path  [shading=_bgjffjk90,_kqwhemy0x] (599.62,154.28) -- (604.78,148.9) .. controls (607.62,145.93) and (612.79,146.29) .. (616.31,149.72) .. controls (619.84,153.15) and (620.39,158.33) .. (617.54,161.31) -- (612.39,166.69) -- cycle ; 
 \draw   (599.62,154.28) -- (604.78,148.9) .. controls (607.62,145.93) and (612.79,146.29) .. (616.31,149.72) .. controls (619.84,153.15) and (620.39,158.33) .. (617.54,161.31) -- (612.39,166.69) -- cycle ; 

\path  [shading=_scwsrws83,_b6z0jvn0g,path fading= _ga7qz4u6p ,fading transform={xshift=2}] (589.02,189.26) -- (601.65,176.8) -- (589.11,164.28) -- (576.49,176.75) -- cycle ; 
 \draw   (589.02,189.26) -- (601.65,176.8) -- (589.11,164.28) -- (576.49,176.75) -- cycle ; 

\draw    (576.49,176.75) -- (601.65,176.8) ;

\path  [shading=_ydlslgru7,_gg3nzbb1r] (598.97,200.36) -- (604.26,205.5) .. controls (607.18,208.34) and (612.29,207.78) .. (615.67,204.25) .. controls (619.05,200.72) and (619.42,195.55) .. (616.5,192.71) -- (611.21,187.57) -- cycle ; 
 \draw   (598.97,200.36) -- (604.26,205.5) .. controls (607.18,208.34) and (612.29,207.78) .. (615.67,204.25) .. controls (619.05,200.72) and (619.42,195.55) .. (616.5,192.71) -- (611.21,187.57) -- cycle ; 

\draw [color={rgb, 255:red, 0; green, 0; blue, 0 }  ,draw opacity=0.7 ]   (519.67,176.77) -- (498.26,199.01) ;
\draw [color={rgb, 255:red, 0; green, 0; blue, 0 }  ,draw opacity=0.7 ]   (519.67,176.77) -- (497.66,155.02) ;
\path  [shading=_ip4gx3285,_rcaruaq5p] (509.26,153.75) -- (503.94,148.58) .. controls (501,145.73) and (495.85,146.26) .. (492.42,149.78) .. controls (488.99,153.29) and (488.59,158.46) .. (491.53,161.31) -- (496.84,166.48) -- cycle ; 
 \draw   (509.26,153.75) -- (503.94,148.58) .. controls (501,145.73) and (495.85,146.26) .. (492.42,149.78) .. controls (488.99,153.29) and (488.59,158.46) .. (491.53,161.31) -- (496.84,166.48) -- cycle ; 

\path  [shading=_dr3uqyfwt,_0qzuw01gm,path fading= _r7x9si0lr ,fading transform={xshift=2}] (519.72,189.34) -- (507.13,176.77) -- (519.62,164.2) -- (532.21,176.77) -- cycle ; 
 \draw   (519.72,189.34) -- (507.13,176.77) -- (519.62,164.2) -- (532.21,176.77) -- cycle ; 

\draw    (532.25,176.75) -- (507.09,176.8) ;
\path  [shading=_bfporumwl,_v25ihwrd4] (509.83,200.24) -- (504.55,205.46) .. controls (501.63,208.34) and (496.48,207.81) .. (493.04,204.29) .. controls (489.61,200.76) and (489.18,195.57) .. (492.1,192.69) -- (497.38,187.47) -- cycle ; 
 \draw   (509.83,200.24) -- (504.55,205.46) .. controls (501.63,208.34) and (496.48,207.81) .. (493.04,204.29) .. controls (489.61,200.76) and (489.18,195.57) .. (492.1,192.69) -- (497.38,187.47) -- cycle ; 

\draw    (554.29,199.27) -- (554.49,224.41) ;

\end{tikzpicture}

%% file: resources/resources.tex
\begin{refsection}
\chapter{Quantum logic essentials}
\label{ch:essentials}
\thispagestyle{empty}
\newpage

    \section{Quantum programming}
    \label{sec:programming}
        \subsection{Universality}
            \label{sec:qlogic}
            In this section we report some preliminaries that explain how a quantum processor can run any algorithm by means of a restricted set of gates. Such a restricted set is called \textit{universal} for this reason.
            We briefly get through the Boolean logic -- which governs classical computation -- and how we can express any Boolean function within the quantum framework. This gives provide the reader with a perspective of how quantum logic can express a wider group of functions. 
            \subsubsection{Boolean logic}
                \begin{wrapfigure}{r}{4.2cm}
                    \centering
                    \begin{quantikz}[thin lines, row sep={0.7cm,between origins}, column sep={0.4cm}]
                        \lstick{$_{\ket{\texttt{b}_1}}$}& \ctrl{2} & \qw\rstick{$_{\ket{\texttt{b}_1}}$}\\
                        \lstick{$_{\ket{\texttt{b}_2}}$}& \ctrl{1} & \qw\rstick{$_{\ket{\texttt{b}_2}}$}\\
                        \lstick{$_{\ket{\texttt{1}}}$}& \targ{} & \qw\rstick{$_{\ket{\lnot (\texttt{b}_1\land\texttt{b}_2)}}$}
                    \end{quantikz}
                    \caption{$\lnot(\texttt{b}_1\land\texttt{b}_2)$ operator by means of a Toffoli.}
                    \label{fig:toffoli}
                    \hrulefill
                \end{wrapfigure}
                Classical computation is deeply based on Boolean logic. Since classical technologies are really advanced and benefits from many years of research on physical implementations, several Boolean operators find direct implementation as classical gates.
                However, our interested is narrowed to how we can express any boolean function by means of quantum operators. Hence we can restrict the discussion to a single logical operator:  $\lnot(\texttt{b}_1\land\texttt{b}_2)$. In fact, for any Boolean variables $\texttt{b}_1,\texttt{b}_2 \in \{\texttt{0},\texttt{1}\}$, the $\lnot(\texttt{b}_1\land\texttt{b}_2)$ operator\footnote{The value is true iff no more than one variable is true.} is universal to Boolean logic \cite{vaughan1942william}, hence we need to find a quantum operator able to realize it. The Toffoli operator \cite{toffoli1980reversible} does the job. Formally, a classical state $\texttt{b}_1\cdot \texttt{b}_2$ gets encoded in the quantum state $\ket{\texttt{b}_1}\otimes\ket{\texttt{b}_2}\otimes\ket{\texttt{1}}$. By applying the Toffoli operator to the encoded system, the last qubit encodes the Boolean state $\lnot(\texttt{b}_1\land\texttt{b}_2)$. This is shown in Figure \ref{fig:toffoli}.
                
                To date, there is no direct physical implementation for the Toffoli operator. It needs to be expressed as a composition of quantum operators physically realizable.

            \subsubsection{Quantum logic}
            We here report a step-by-step introduction to quantum logic and what a quantum processors should have to provide universal computing.
            Quantum computing works with the logic of pure states. A Hilbert space $\mathbb{H}$ of dimension $d$ is closed under the \textit{unitary group} of degree $d$. This means that for any pure state $\ket{\varphi} \in \mathbb{H}$ and any unitary operator $\texttt{U}$, it results $\texttt{U}\ket{\varphi} \in \mathbb{H}$.
            
            A generic quantum algorithm can be expressed as a system initialized to $\ket{\texttt{0}}^{\otimes d}$ and a unitary $\texttt{U}$ operating on it. Figure \ref{fig:algo} gives a circuit representation.
                    \begin{figure}[h]
                        \centering
                        \begin{quantikz}[row sep={1cm,between origins},column sep=0.5cm]
                        	\lstick{$_{\ket{\texttt{0}}^{\otimes d}}$} &\gate[style={fill=violet!10}]{_{\texttt{U}}} & \qw \rstick{$_{\texttt{U}\ket{\texttt{0}}^{\otimes d}}$}
                        \end{quantikz}
                        \caption{Generic algorithm expressed as a unitary, operating over a $\ket{\texttt{0}}^{\otimes d}$ state.}
                        \label{fig:algo}
                    \end{figure}
            
            Since quantum processors do not supply as primitive operator a generic unitary, this is subject to one or more steps of \textit{decomposition} or \textit{synthesis}. A generic decomposition is showed in circuit of Figure \ref{fig:uk}.
            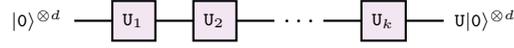
\begin{figure}[h]
                \centering
                        \begin{quantikz}[row sep={1cm,between origins},column sep=0.5cm]
                        	\lstick{$_{\ket{\texttt{0}}^{\otimes d}}$} &\gate[style={fill=violet!10}]{_{\texttt{U}_1}} &\gate[style={fill=violet!10}]{_{\texttt{U}_2}} & \ \ldots\ \qw&\gate[style={fill=violet!10}]{_{\texttt{U}_k}}  & \qw \rstick{$_{\texttt{U}\ket{\texttt{0}}^{\otimes d}}$}
                        \end{quantikz}
                \caption{Generic decomposition of $\texttt{U}$ into $k$ unitaries.}
                \label{fig:uk}
            \end{figure}
            
            A universal operator set should be \textit{efficient}, in the sense that the overhead caused by the decomposition from a $d$-degree unitary to the operator set is upper-bounded by some polynomial function. 
            There are several (historically) important results showing how such a requirement is achievable.
            We start from one coming from the work done in \cite{barenco1995universal,divincenzo1995two}, showing that, given a generic $2$-degree unitary $\texttt{U}$ -- i.e. it operates over a single qubit --, the following operator is universal:
            \begin{equation}
                \label{eq:cu}
                \land(\texttt{U}) \equiv \ket{\texttt{0}}\bra{\texttt{0}}\otimes \mathds{1} + \ket{\texttt{1}}\bra{\texttt{1}}\otimes\texttt{U}
            \end{equation}
            
            \begin{wrapfigure}{r}{4.5cm}
                \centering
                    \begin{quantikz}[thin lines, row sep={0.7cm,between origins}, column sep={0.4cm}]
                        \qw & \ctrl{1} & \qw\\
                    	\qw & \gate[style={fill=violet!10}]{_{\texttt{U}}}  & \qw
                    \end{quantikz}
                \caption{Circuit representation for the $\land(\texttt{U})$ operator.}
                \label{fig:cu}
                \hrulefill
            \end{wrapfigure}
            \noindent{}This is a controlled-\texttt{U} operator in Feynmann's notation \cite{feynman1985quantum}. We will use this notation throughout this thesis because of its versatility. Even if it is a bit outdated, we think it helps to highlight the logic behind the operators; it also helps us to keep consistency throughout the chapters. In the circuit model the subject operator appears as in Figure \ref{fig:cu}.
            
            One can notice that the group of equation \eqref{eq:cu} is pretty compact, as it involves only 2-qubit operators where one of them act always as control.
            Nevertheless, some decomposition step is necessary to get closer to what real processors can actually offer. To this aim, we need to introduce some further operators. The first one is known as \textit{special unitary}. We refer to this operator as $\texttt{V}_{\alpha, \beta, \delta}$ and it is defined as follows \cite{barenco1995elementary}:
            \begin{equation}
                \texttt{V}_{\alpha, \beta, \delta} \equiv \begin{pmatrix}
                        e^{\sfrac{\texttt{i}(\alpha \texttt{+} \beta)}{2}}\cos{\sfrac{\delta}{2}} & e^{\sfrac{\texttt{i}(\alpha \texttt{-} \beta)}{2}}\sin{\sfrac{\delta}{2}}\\
                        -e^{\sfrac{\texttt{i}(\texttt{-}\alpha \texttt{+} \beta)}{2}}\sin{\sfrac{\delta}{2}} & e^{\sfrac{\texttt{i}(\texttt{-}\alpha \texttt{-} \beta)}{2}}\cos{\sfrac{\delta}{2}}
                    \end{pmatrix}
            \end{equation}
            The second operator represent a \textit{global phase shift}:
            \begin{equation}
                \texttt{G}_{\gamma} \equiv \begin{pmatrix}
                        e^{\texttt{i}\gamma} & 0\\
                        0 & e^{\texttt{i}\gamma}
                    \end{pmatrix}
            \end{equation}
            It results that special unitaries composed with global phase shifts characterizes the group of unitaries. It follows the same statement in matrix form:
            \begin{equation}
                \label{eq:u=gv}
                \texttt{U}_{\gamma, \alpha, \beta, \delta} \equiv \texttt{G}_{\gamma}\texttt{V}_{\alpha, \beta, \delta}
            \end{equation}
            As consequence, the first decomposition we can apply comes from the circuit equivalence of Figure \ref{fig:cgcv}.
            \begin{figure}[h]
                \centering
                \begin{quantikz}[thin lines, row sep={0.7cm,between origins}, column sep={0.4cm}]
                    \qw & \ctrl{1} & \qw\\
                	\qw & \gate[style={fill=violet!10}]{_{\texttt{U}_{\gamma, \alpha, \beta, \delta}}} & \qw
                \end{quantikz}
                $\ \ \equiv$
                \begin{quantikz}[thin lines, row sep={0.7cm,between origins}, column sep={0.4cm}]
                    \qw & \ctrl{1} & \ctrl{1} & \qw\\
                    \qw & \gate[style={fill=violet!10}]{_{\texttt{V}_{\alpha, \beta, \delta}}} & \gate[style={fill=violet!10}]{_{\texttt{G}_{\gamma}}} & \qw
                \end{quantikz}
                \caption{First decomposition as a result of equivalence \eqref{eq:u=gv}.}
                \label{fig:cgcv}
            \end{figure}
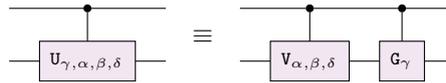
            
            A conditioned global phase shift is equivalent to a \textit{relative phase shift}. Thus, let $\texttt{R}_{\gamma}$ such an operator, defined as follows:
            \begin{equation}
                \texttt{R}_{\gamma} \equiv \begin{pmatrix}
                        1 & 0\\
                        0 & e^{\texttt{i}\gamma}
                    \end{pmatrix}
            \end{equation}
            
            \begin{wrapfigure}{r}{5.5cm}
                \centering
                \begin{quantikz}[thin lines, row sep={0.3cm,between origins}, column sep={0.4cm}]
                    \qw & \ctrl{2} & \qw & & & \gate[style={fill=violet!10}]{_{\texttt{R}_{\gamma}}} & \qw \\
                    & & &\equiv& & &\\
                    \qw & \gate[style={fill=violet!10}]{_{\texttt{G}_{\gamma}}} & \qw & & & \qw & \qw
                \end{quantikz}
                \caption{Circuit representation of the equivalence $\land(\texttt{G}_{\gamma}) \equiv \texttt{R}_{\gamma}\otimes\mathds{1}$.}
                \label{fig:cg=r}
                \hrulefill
            \end{wrapfigure}
            \noindent{}The above statement allows us to \textit{synthesize} $\land(\texttt{G}_{\gamma})$ into $\texttt{R}_{\gamma}\otimes\mathds{1}$. Figure \ref{fig:cg=r} shows the corresponding circuit equivalence.
            
            The final step to achieve universal computing through physically realizable operators is to decompose $\land(\texttt{V}_{\alpha, \beta, \delta})$. To this aim, let us introduce the rotational operators over the Pauli axes. Namely, $\texttt{X}_{\gamma} \equiv \texttt{e}^{\texttt{-i}\texttt{X}\sfrac{\gamma}{2}}$, $\texttt{Y}_{\gamma} \equiv \texttt{e}^{\texttt{-i}\texttt{Y}\sfrac{\gamma}{2}}$ and $\texttt{Z}_{\gamma} \equiv \texttt{e}^{\texttt{-i}\texttt{Z}\sfrac{\gamma}{2}}$.
            Notice also that
            \begin{equation}
                \texttt{e}^{\texttt{-i}\texttt{E}\sfrac{\gamma}{2}} = \cos{\sfrac{\gamma}{2}} - \texttt{i}\texttt{E}\sin{\sfrac{\gamma}{2}},
            \end{equation}
            holds for any $\texttt{E} \in \{\texttt{X},\texttt{Y},\texttt{Z}\}$.
            Therefore, whenever $\gamma = \pi$, each rotational operator relates to the corresponding Pauli operator \texttt{X}, \texttt{Y} or \texttt{Z}. Formally, they are equivalent up to the global phase $\texttt{G}_{\sfrac{\texttt{-}\pi}{2}} \equiv -\texttt{i}$. E.g. $\texttt{X}_{\pi} \cong \texttt{X}$.
        
            Now we proceed by reporting the results coming from \cite{barenco1995elementary}. Let \texttt{A}$_{\alpha,\delta}$, \texttt{B}$_{\delta,\alpha,\beta}$, \texttt{C}$_{\alpha,\beta}$ be a triplet of special unitaries defined as follows:
            \begin{itemize}
                \item \texttt{A}$_{\alpha,\delta} = \texttt{Y}_{\sfrac{\delta}{2}}\texttt{Z}_{\alpha}$;
                \item \texttt{B}$_{\delta,\alpha,\beta} = \texttt{Y}_{\sfrac{\texttt{-}\delta}{2}}\texttt{Z}_{\sfrac{\texttt{-}(\alpha\texttt{+}\beta)}{2}}$;
                \item \texttt{C}$_{\alpha,\beta} = \texttt{Z}_{\sfrac{(\beta\texttt{-}\alpha)}{2}}$.
            \end{itemize}
            With such a triplet, together with $\land(\texttt{X})$ we are able to decompose $\land(\texttt{V}_{\alpha, \beta, \delta})$, according to the circuit equivalence shown in Figure \ref{fig:cv_dec}.
            \begin{figure}[h]
                \centering
                \begin{quantikz}[thin lines, row sep={0.7cm,between origins}, column sep={0.4cm}]
                    \qw & \ctrl{1} & \qw\\
                	\qw & \gate[style={fill=violet!10}]{_{\texttt{V}_{\alpha, \beta, \delta}}} & \qw
                \end{quantikz}
                $\ \ \equiv$
                \begin{quantikz}[thin lines, row sep={0.7cm,between origins}, column sep={0.4cm}]
                    \qw & \qw & \ctrl{1} & \qw & \ctrl{1} & \qw & \qw\\
                    \qw & \gate[style={fill=violet!10}]{_{\texttt{C}_{\alpha,\beta}}} & \targ{} & \gate[style={fill=violet!10}]{_{\texttt{B}_{\delta,\alpha,\beta}}} & \targ{} & \gate[style={fill=violet!10}]{_{\texttt{A}_{\alpha,\delta}}} & \qw
                \end{quantikz}
                \caption{Decomposition of $\land(\texttt{V}_{\alpha, \beta, \delta})$.}
                \label{fig:cv_dec}
            \end{figure}
            
            We can finally show the universality by composing the two results and getting a decomposition for $\land(\texttt{U}_{\gamma, \alpha, \beta, \delta})$ -- see Figure \ref{fig:cu_dec}.
            \begin{figure}[h]
                \centering
                \begin{quantikz}[thin lines, row sep={0.3cm,between origins}, column sep={0.4cm}]
                    \qw & \ctrl{2} & \qw & & & \qw & \ctrl{2} & \qw & \ctrl{2} & \qw &\gate[style={fill=violet!10}]{_{\texttt{R}_{\gamma}}}& \qw\\
                    & & &\ \equiv\ & & &\\
                    \qw & \gate[style={fill=violet!10}]{_{\texttt{U}_{\gamma, \alpha, \beta, \delta}}} & \qw & & & \gate[style={fill=violet!10}]{_{\texttt{C}_{\alpha,\beta}}} & \targ{} & \gate[style={fill=violet!10}]{_{\texttt{B}_{\delta,\alpha,\beta}}} & \targ{} & \gate[style={fill=violet!10}]{_{\texttt{A}_{\alpha,\delta}}} & \qw & \qw
                \end{quantikz}
                \caption{Decomposition of $\land(\texttt{U}_{\gamma,\alpha, \beta, \delta})$.}
                \label{fig:cu_dec}
            \end{figure}

            \begin{theorem}
                The \texttt{IBM} gate set \cite{ibm-gateset,ibm-resonance,rzx-gate,rigetti2010fully,mckay2017efficient} is universal.
            \end{theorem}
            \begin{proof}
                All the processors supplied by the \texttt{IBM} cloud have gate set $\{\texttt{Z}_{\gamma}, \texttt{X}_{\sfrac{\pi}{2}}, \texttt{X}, \land(\texttt{X})\}$. To prove the universality of such a set, notice that $\texttt{R}_{\gamma} \cong \texttt{Z}_{\gamma}$. Furthermore, any special unitary $\texttt{V}_{\alpha,\beta,\delta}$ can be factorized as follows \cite{barenco1995elementary}:
                \[\texttt{V}_{\alpha,\beta,\delta} \equiv \texttt{Z}_{\alpha}\texttt{Y}_{\delta}\texttt{Z}_{\beta}\]
                Since an \texttt{IBM} processor can run natively a generic $\texttt{Z}_{\gamma}$ gate, we only need to synthesise $\texttt{Y}_{\delta}$. 
                This can be done, by using operations from the gate set only:
                \[\texttt{Y}_{\delta} \equiv \texttt{X}_{\sfrac{\pi}{2}} \texttt{Z}_{\delta\texttt{+}\pi}\texttt{X}_{\sfrac{3\pi}{2}}\]
                In conclusion, since the gate set also provides $\land(\texttt{X})$, any $\land(\texttt{U})$ can be realized through a composition of operations coming from the \texttt{IBM} gate set.
            \end{proof}
            Not all the existing processors are universal. \texttt{D-Wave} ones are an example. In fact, the company goal is to build specific-purpose processors, meant to explore the field of optimization problems through \textit{quantum annealing} procedures \cite{dwave-opt}.

    \subsubsection{Relation between classical and quantum logic}
    \label{sec:decomp}
    It has been shown \cite{shi2003both} that, to achieve quantum universality, one can start from an operator universal in the Boolean functions, i.e. the Toffoli, and by adding only a single 1-qubit operator, the operator set is quantum universal. The subject operator can be expressed as $\texttt{X}_{\sfrac{\pi}{2}}\texttt{Z}_{\sfrac{\pi}{2}}\texttt{X}_{\sfrac{\pi}{2}}$\footnote{As well as $\texttt{Y}_{\sfrac{\pi}{2}}\texttt{X}$ or as the more common \textit{Hadamard} gate \texttt{H}. We are not going to use the latter in this thesis as we opted to keep the treating closer to real gate implementations.}. As stated in \cite{aharonov2003simple}, this ``[...] can be interpreted as 
    saying that Fourier transform is really all there is to quantum computation on top of classical'', since $\texttt{X}_{\sfrac{\pi}{2}}\texttt{Z}_{\sfrac{\pi}{2}}\texttt{X}_{\sfrac{\pi}{2}}$ corresponds to a Fourier transform.
    

    \subsection{The Clifford group}
    \label{sec:cliff}
    The Clifford group $\mathbb{C}$ dues its importance to its implication in fault-tolerant computation \cite{anderson2014fault}, simulation \cite{gottesman1998theory} and benchmarking \cite{magesan2012characterizing}. Such a group is generated by 3 operators:
    \begin{equation}
        \label{eq:cliff}
        \mathbb{C} \equiv \langle \land(\texttt{X}), \texttt{X}_{\sfrac{\pi}{2}}, \texttt{Z}_{\sfrac{\pi}{2}}\rangle.
    \end{equation}
    $\mathbb{C}$ can be efficiently simulated by a classical computer \cite{gidney2021stim}. This has as comeback that one can evaluate fault-tolerant protocols classically. As drawback, it is not universal. However to achieve universality while at the same time providing an operator set which can realize any quantum evolution efficiently one need to add a single 1-qubit operator, usually assumed to be $\texttt{R}_{\sfrac{\pi}{4}}$\footnote{Commonly referred as the \texttt{T} gate.} or the corresponding in rotational terms $\texttt{Z}_{\sfrac{\pi}{4}}$. In fact, any unitary of dimensionality $2$ can be approximated efficiently and with arbitrary precision \cite{dawson2005solovay, kliuchnikov2013synthesis,vatan2004optimal}. Formally, by properly composing single qubit operators coming from $\mathbb{C}$ -- i.e.  $\texttt{X}_{\sfrac{\pi}{2}}, \texttt{Z}_{\sfrac{\pi}{2}}$ -- together with $\texttt{Z}_{\sfrac{\pi}{4}}$ one can achieve universality in the \textit{dense} sense, by only introducing a polynomial overhead\footnote{Other extensions of $\mathbb{C}$ may be of interest. E.g. in Ref. \cite{glaudell2021optimal} authors consider $\land(\texttt{Z}_{\sfrac{\pi}{2}})$ in their generator set.}.

    Since we are facing with a discrete operator set composed by even fractions of $\pi$ only, we can re-state the nomenclature. Namely, the same universal operator set can be expressed in the following intuitive way:
    \begin{equation}
        \label{eq:univ_sqrt}
        \mathbb{C}^{\texttt{+}} \equiv \langle \land(\texttt{X}), \texttt{X}^{\sfrac{1}{2}}, \texttt{Z}^{\sfrac{1}{2}}, \texttt{Z}^{\sfrac{1}{4}}\rangle.
    \end{equation}
    This nomenclature stresses the logic behind the relation Pauli-rotational gates, as the power function degree says how many times one needs to apply the rotational gate to simulate a Pauli operator.



\subsection{Programming in higher order framework}
\label{sec:logical-ifo}
        \subsubsection{Time-ordered framework}
        So far we have implicitly assumed that a quantum evolution undergoes a time-ordered definition. Formally, consider two unitaries $\texttt{U}$ and $\texttt{V}$ and a state $\ket{\vartheta}$. Then, any time-ordered framework force us to chose whether $\texttt{U}$ operates on $\ket{\vartheta}$ before or after $\texttt{V}$. In circuit representation, these two cases are respectively
            \begin{center}
                    \begin{quantikz}[thin lines, row sep={0.3cm,between origins}, column sep={0.4cm}]
                    	\lstick[]{$_{\ket{\vartheta}}$}&\gate[style={fill=violet!10}]{_{\texttt{U}}} & \gate[style={fill=violet!10}]{_{\texttt{V}}} & \qw \rstick[]{$_{\texttt{VU}\ket{\vartheta}}$}
                    \end{quantikz}
            \end{center}
            and
            \begin{center}
                    \begin{quantikz}[thin lines, row sep={0.3cm,between origins}, column sep={0.4cm}]
                    	\lstick[]{$_{\ket{\vartheta}}$}&\gate[style={fill=violet!10}]{_{\texttt{V}}} & \gate[style={fill=violet!10}]{_{\texttt{V}}} & \qw \rstick[]{$_{\texttt{UV}\ket{\vartheta}}$}
                    \end{quantikz}
            \end{center}
            
        However quantum mechanics allows to think of more general frameworks, where, not only states, but also operators can be superposed to create more complex systems. Such systems may bring new non-classical advantages. Attempts in formally designing a higher-order framework is an active branch of research \cite{chiribella2008quantum,ChiKri-19}, but this is out of scope. Rather, in the following, we work with a higher-order \textit{oracle} implementing an \textit{indefinite casual order}.
        
        \subsubsection{Indefinite causal orders}
        \label{sec:soco_comp}
        The \textit{indefinite causal order} is an interesting property of quantum mechanics. In brief, it is a quantum evolution where two or more operations occur, but the order in which they occur is causally ordered by an extra quantum system. This creates a superposition of causal orders among those operations. Such a resource can be used for several purposes. In Refs. \cite{felce2020quantum,simonov2022work,guha2020thermodynamic}, the indefinite causal orders is considered for thermalization protocols. A big line of investigation deals with enhancing quantum communication \cite{EblSalChi-18, SalEblChi-18,CalCac-20, koudia2022deep}; we gave the first experimental witness for such a resource \cite{cuomo2021experiencing}\footnote{We report this in Ch. \ref{ch:noise}.}. As regards computation, indefinite causal order would speed up tasks that involves permuting operations \cite{Chi-12, AraCosBru-14, colnaghi2012quantum}. In Ref. \cite{procopio2015experimental} authors implement a photonic-based discrimination protocol solved by \textit{superposing unitaries}.

       
            The indefinite causal orders for computation consists essentially on superposing two or more unitaries. The superposition is then coherently conditioned to an auxiliary qubit, called control qubit. 
            For the case of 2 unitaries $\texttt{U},\texttt{V}$, the superposing unitary operator $\texttt{S}$ can be defined as follows:
            \begin{equation}
                \texttt{S} =
                \begin{bmatrix} 
                \texttt{UV} & \mathbf{0}\\
                \mathbf{0} & {\texttt{VU}} 
            \end{bmatrix}.
            \end{equation}
            By treating the operator $\texttt{S}$ as an oracle -- i.e. it takes a unit of time to run -- it brings the advantage of evaluating two different orders at the same time. This advantage comes from the fact that we are extending the standard framework -- which originally could only superpose quantum states -- to being able to superpose unitaries. This can be used, for example, in Information Processing to distinguish between different evolution by means of a single check \cite{procopio2015experimental}. This is unthinkable with standard frameworks where the order of execution of the unitaries must be defined.
            
            However, implementing $\texttt{S}$ does not necessarily reflect the advantage coming from theory. It is necessary to make a distinction to what is a real implementation of $\texttt{S}$ and what is more like a \textit{simulation}. In other words, implementing $\texttt{S}$ means that the physical settings preserve the theoretical advantages. Only in this case one can treat $\texttt{S}$ as an oracle. It is an open question if such a real implementation will ever be possible \cite{vilasini2022embedding, RubRozFei-17}. In the attempt of finding a solution, a framework meant to superpose \textit{gravitational fields} has been proposed \cite{paunkovic2020causal}.
            
            Stemming from the above, what we can do now is to realize $\texttt{S}$ by means of simulation. Namely, an equivalent evolution which does not preserve the speed-up advantage. An example is shown in Figure \ref{fig:ico-sim}.
            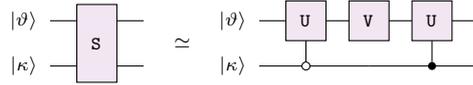
\begin{figure}[h]
                \centering
                \begin{quantikz}[thin lines, row sep={0.3cm,between origins}, column sep={0.4cm}]
                    \\
                    \lstick{$_{\ket{\vartheta}}$}&\gate[3, nwires = {2},style={fill=violet!10}]{_{\texttt{S}}} &\qw\\
                    &&&\\
                    \lstick{$_{\ket{\kappa}}$}&& \qw\\
                \end{quantikz}
                $_{\simeq}$
                \begin{quantikz}[thin lines, row sep={0.3cm,between origins}, column sep={0.4cm}]
                \\
            	    \lstick{$_{\ket{\vartheta}}$}&\gate[style={fill=violet!10}]{_{\texttt{U}}} & \gate[style={fill=violet!10}]{_{\texttt{V}}} &\gate[style={fill=violet!10}]{_{\texttt{U}}} & \qw \\
                    &&&&\\
            	    \lstick{$_{\ket{\kappa}}$}& \octrl{-2} & \qw & \ctrl{-2} & \qw\\
                \end{quantikz}
                \caption{Simulation of oracle $\texttt{S}$; losing the theoretical advantage.}
                \label{fig:ico-sim}
            \end{figure}
            This example immediately shows the theoretical loss, as it requires 2 use for one of the unitaries \cite{ChiDarPer-13}. Specifically, since $\texttt{U}$ operators run under complementary conditions, one of them does not run, meaning that one time step is always dedicated to perform an identity operation $\mathds{1}$ -- an idle time.
            
            \begin{wrapfigure}{r}{5cm}
                \centering
                \begin{quantikz}[thin lines, row sep={0.6cm,between origins}, column sep={0.4cm}]
                	    \lstick{$_{\ket{\vartheta}}$}&\targX{}&\gate[style={fill=violet!10}]{_{\texttt{U}}} & \gate[style={fill=violet!10}]{_{\texttt{V}}}&\qw\\
                	    \lstick{$_{\ket{\texttt{0}}}$}&\swap{-1}&\gate[style={fill=violet!10}]{_{\texttt{V}}} & \gate[style={fill=violet!10}]{_{\texttt{U}}}&\qw\\
                	    \lstick{$_{\ket{\kappa}}$}&\ctrl{-2}&\qw&\qw&\qw
                \end{quantikz} 
                \caption{Simulation of oracle $\texttt{S}$ by means an auxiliary qubit.}
                \label{fig:ico_aux}
                \hrulefill
            \end{wrapfigure}
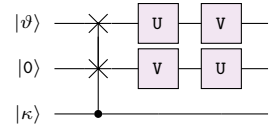
            A simulation which looks more adapt to reflect the natural behaviour of $\texttt{S}$ is represented in Figure \ref{fig:ico_aux}. This circuit makes use of an auxiliary qubit initialized to $\ket{\texttt{0}}$. A swap conditioned on $\ket{\kappa}$ is the only potentially entangling gate between the three states. However, $\ket{\texttt{0}}$ has no impact on the overall system\footnote{I.e., no phase shift leaks in nor out of state $\ket{\texttt{0}}$.}, up to the permutation caused by the swap operation\footnote{Which can be retrieved through a non-destructive measurement $\texttt{Z}\otimes\texttt{Z}$. Non-destructive measurements are introduced in Chapter \ref{ch:noise}.}. It follows that by tracing out such a sub-system, the subject circuit implements $\texttt{S}$. Thanks to the auxiliary qubit, the two different orders can be expressed in parallel.
        
        Nevertheless notice that it make use of a complex gate at the beginning -- i.e. a controlled swap. However, to date, there is no technology providing such a gate natively, hence it is necessary to consider it as an oracle to not lose the advantage. The best we can do now is providing an efficient decomposition for it, aware of the fact that we have partial knowledge of the input -- i.e. auxiliary qubit being in state $\ket{\texttt{0}}$. Figure \ref{fig:cswap_opt} shows an optimized decomposition w.r.t. the standard one \cite{shende2009cnot}.
        \begin{figure}[h]
            \centering
            \begin{quantikz}[thin lines, row sep={0.6cm,between origins}, column sep={0.3cm}]
    	            & \qw & \ctrl{1} & \qw & \qw& \qw & \ctrl{1} & \qw & \qw& \qw & \qw & \targ{} & \qw\\
                	\lstick{$_{\ket{\texttt{0}}}$}&\gate[style={fill=red!10}]{_{{\texttt{X}^{\texttt{-}\sfrac{1}{2}}}}} & \targ{} & \gate[style={fill=green!10}]{_{{\texttt{Z}^{\texttt{-}\sfrac{1}{4}}}}} & \targ{} & \gate[style={fill=green!10}]{_{{\texttt{Z}^{\sfrac{1}{4}}}}} & \targ{} & \gate[style={fill=green!10}]{_{{\texttt{Z}^{\texttt{-}\sfrac{1}{4}}}}} & \targ{} &\gate[style={fill=green!10}]{_{{\texttt{Z}^{\sfrac{1}{4}}}}}& \gate[style={fill=red!10}]{_{{\texttt{X}^{\sfrac{1}{2}}}}} & \ctrl{-1}&\qw \rstick{$\ {\cong}$}\\
                	& \qw & \qw & \qw & \ctrl{-1} & \qw & \qw  & \qw& \ctrl{-1} & \qw&\qw & \qw & \qw
                \end{quantikz}
                \begin{quantikz}[thin lines, row sep={0.6cm,between origins}, column sep={0.3cm}]
                \\
                &\swap{2}&\qw\\
                \lstick{$_{\ket{\texttt{0}}}$}&\targX{}&\qw\\
                &\ctrl{-2}&\qw\\
                \end{quantikz}
                \caption{Optimized decomposition of the controlled swap operation.}
                \label{fig:cswap_opt}
            \end{figure}
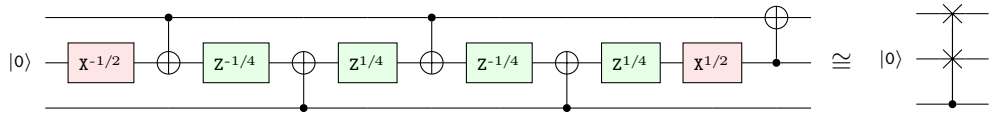
            
        Unfortunately, as long as native controlled swap are not physically implemented, the speed-up advantage -- w.r.t. a time-ordered framework -- is still just theoretical and it cannot be witnessed on real implementation. This is especially true when trying to superpose more than two unitaries. Some investigations are available in Refs. \cite{renner2022computational, colnaghi2012quantum}, where the computational speed up is preserved as long as the swaps are considered oracles.
        
        Nevertheless, a non-native implementation of indefinite causal orders may still bring some sort of quantum advantage, as it can be used for \textit{magnitude amplification} \cite{zavatta2011high,kim2008scheme,zavatta2009experimental}\footnote{Useful, e.g., to implement Grover's algorithm oracle \cite{zalka1999grover}}. In Section \ref{sec:soco}, we treat this concept -- also experimentally -- to enhance communication capacity.

    \section{Entanglement-based computation}
    \label{sec:e-comp}
    Entanglement is probably the most fascinating property coming from quantum mechanics. It is also the most promising resource.
    
    To our purpose we focus on one of the four Bell states. These states are fully entangled, meaning that their correlation is maximal and the system is said to be close.
    Let us introduce the $\ket{\Phi^{\texttt{+}}}$ state, defined as follow:
    \[\ket{\Phi^{\texttt{+}}} = \frac{1}{\sqrt{2}}(\ket{\texttt{00}} + \ket{11}).\]
    Since the two system in $\ket{\Phi^{\texttt{+}}}$ present a \textit{non-local} correlation, this state can be used to perform non-local operations. Physically speaking, this means that one need to perform what is called \textit{entanglement generation and distribution} \cite{cacciapuoti2019quantum, cacciapuoti2020entanglement, cuomo2020towards}:
    \begin{itemize}
        \item generation; despite the non-local correlation, the generation happens between system which are in proximity one another. Generating a maximally entangled state with high fidelity is generally hard and time-consuming.
        \item distribution; once that the entanglement is ready, it is possible to relocate the two systems. The entanglement is, in principle, preserved.
    \end{itemize}
    Below we outline some of the most promising resources for distributed quantum computing, all exploiting entanglement. Most of what follows comes from the work published in \cite{cuomo2021optimized}.
    Before proceeding, we need to introduce a formalism for \textit{measurement-based computation}.
    
    \subsubsection{Measurement-based computation}
    A computing paradigm is called measurement-based whenever the quantum computation is interleaved by measurements acting on a sub-system \cite{nielsen2003quantum}. The output of a measurement is then used to perform classical conditioned quantum operations.
        
    A measurement over a Pauli axis $\texttt{E}$ generates a Boolean $\texttt{b}$:
        \begin{center}
                \begin{quantikz}[thin lines,row sep={0.7cm,between origins}, column sep={0.4cm}]
                    &\measure[style={fill=gray!10}]{_{\langle \texttt{E}\rangle,\texttt{b}}}
                \end{quantikz}
        \end{center}
        The output can then be used to choose whether performing or not a unitary operation $\texttt{U}$ over another qubit:
        \begin{center}
                \begin{quantikz}[thin lines,row sep={0.7cm,between origins}, column sep={0.4cm}]
                    &\gate[style={fill=gray!10}]{_{\texttt{U}^{\texttt{b}}}}&\qw
                \end{quantikz}
        \end{center}

Whenever to a Pauli gate follows a Pauli measurement, there is no need to apply the former. Precisely, instead of applying the quantum gate followed by the measurement, one can always perform the measurement and then applying a classical correction on the output, corresponding to the Pauli gate. Such a technique may be referred as \textit{pushing} technique as, intuitively, one can picture this manipulation as pushing the gate beyond the measurement. See circuits in Figure \ref{circ:push} for an example.

        \begin{figure}[h]
            \centering
             \begin{subfigure}[b]{0.48\textwidth}
                \centering
                \begin{quantikz}[thin lines,row sep={0.7cm,between origins}, column sep={0.4cm}]
                &\gate[style={fill=red!10}]{_{\texttt{X}}} & \measure[style={fill=gray!10}]{_{\langle \texttt{Z}\rangle,\texttt{b}_1}}
            \end{quantikz}
            $_{\ \ \equiv}$
            \begin{quantikz}[thin lines,row sep={0.7cm,between origins}, column sep={0.4cm}]
                & \measure[style={fill=gray!10}]{_{\langle \texttt{Z}\rangle,\lnot\texttt{b}_1}}
            \end{quantikz}
             \end{subfigure}
             \hfill
             \begin{subfigure}[b]{0.48\textwidth}
                \centering
                \begin{quantikz}[thin lines,row sep={0.7cm,between origins}, column sep={0.4cm}]
                &\gate[style={fill=green!10}]{_{\texttt{Z}}} & \measure[style={fill=gray!10}]{_{\langle \texttt{X}\rangle,\texttt{b}_1}}
            \end{quantikz}
            $_{\ \ \equiv}$
            \begin{quantikz}[thin lines,row sep={0.7cm,between origins}, column sep={0.4cm}]
                & \measure[style={fill=gray!10}]{_{\langle \texttt{X}\rangle,\lnot\texttt{b}_1}}
            \end{quantikz}
             \end{subfigure}
             \caption{Pushing technique to avoid quantum gates.}
             \label{circ:push}
        \end{figure}
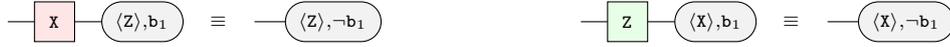

        \subsection{Teleportation}
        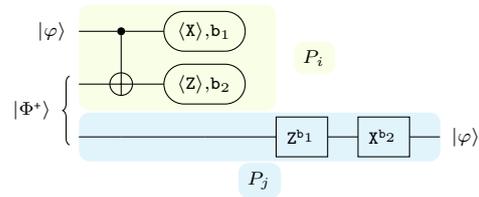
\begin{wrapfigure}{r}{7cm}
                \centering
                \begin{quantikz}[thin lines, row sep={0.7cm,between origins}, column sep={0.4cm}]
            		\lstick[]{$_{\ket{\varphi}}$}& \ctrl{1}\gategroup[wires=2,steps=2,style={draw=none,rounded corners,inner xsep=8pt,inner ysep=-1pt, fill=lime!10}, background, label style={rounded corners,label position=above, xshift = 1.8cm,  yshift=-0.85cm, fill=lime!10}, background]{$_{P_i}$} & \measure[style={fill=lime!10}]{_{\langle {\texttt{X}}\rangle,\texttt{b}_1}} \\
            		\lstick[2]{$_{\ket{\Phi^{\texttt{+}}}}$}& \targ{} & \measure[style={fill=lime!10}]{_{\langle {\texttt{Z}}\rangle,\texttt{b}_2}} \\
            		& \qw\gategroup[wires=1,steps=4,style={draw=none,rounded corners,inner xsep=8pt,inner ysep=-1pt, fill=cyan!10}, background,label style={rounded corners,label position=below, yshift=-0.37cm, fill=cyan!10}, background]{$_{P_j}$} & \qw & \gate[style={fill=cyan!10}]{_{\texttt{Z}^{\texttt{b}_1}}} & \gate[style={fill=cyan!10}]{_{\texttt{X}^{\texttt{b}_2}}} & \qw \rstick[]{$_{\ket{\varphi}}$}
            	\end{quantikz}
                \caption{Quantum teleportation protocol between two parties.}
                \label{fig:teleport}
                \hrulefill
            \end{wrapfigure}
        The first resource we report is called \textit{teleportation} \cite{cuomo2020towards}. It is a protocol that, by fact, teleport quantum information from a system to another, by means of entanglement. Figure \ref{fig:teleport} shows the protocol steps in circuit representation. The wires are grouped by color, representing different processor, or memories. A generic processor is referred as $P_i$. In the case of teleportation 2 processors $P_i,P_j$ are involved, which has the roles of \textit{sender} and \textit{receiver}. At the beginning of the protocol, the quantum information $\ket{\varphi}$ is located within the sender, together with half of the Bell state $\ket{\Phi^{\texttt{+}}}$. The receiver need to store the other part of the Bell state. At this point the receiver need to wait that the receiver perform a few operations meant to ``inject'' the information within its half part of the entangled state. Because of the entanglement, now the quantum information is already spread all over the system, which means that $P_j$, the receiver, also has partial knowledge of it.
        
        The end of this part consists on measuring the entire system in the orthogonal basis $\texttt{X}\otimes\texttt{Z}$. This is a necessary step in order to ensure that the receiver will get exactly $\ket{\varphi}$. In fact, the measurement will produce two boolean values, $\texttt{b}_1,\texttt{b}_2$, that the receiver will need to correct its state.

        To understand how such a resource can be used in computation, consider that it is often the case where two states need to interact but they are stored in such a position that does not allow them to do so, unless some middle step is performed to approach them. This means that a \textit{routing} protocol would take care of moving those states up to a couple of qubits which are able to interact one another. This can be done by means of teleportation \cite{hillmich2021exploiting}. However, in local quantum computation, it is more common to find some smart criteria to re-arrange the state storage when necessary, by means of \textit{swapping} protocols \cite{zulehner2019compiling, o2019generalized, madden2022best}. 
        Even if this is the most common approach, it may be smart as well to instead consider teleportation as an alternative to swapping protocols. 
        
        When considering distributed architecture, this are generally assumed to deeply exploit entanglement for their interconnection. Hence, it is more likely to see in the future proposal exploiting teleportation, rather than swapping.
        

        \subsection{Non-local operations}
        \label{sec:rcx}
        \begin{figure}[h]
                \centering
                \begin{quantikz}[thin lines, row sep={0.7cm,between origins}, column sep={0.4cm}]
            		\lstick[]{$_{\ket{\varphi}}$}& \ctrl{1}\gategroup[wires=2,steps=3,style={draw=none,rounded corners,inner xsep=8pt,inner ysep=-1pt, fill=lime!10}, background, label style={rounded corners,label position=above,  yshift=0.08cm, fill=lime!10}, background]{$_{P_i}$} &\qw & \gate[style={fill=lime!10}]{_{\texttt{Z}^{\texttt{b}_2}}} &  \qw \\
            		\lstick[2]{$_{\ket{\Phi^{\texttt{+}}}}$}& \targ{} & \measure[style={fill=lime!10}]{_{\langle {\texttt{Z}}\rangle,\ \texttt{b}_1}} \\
            		& \ctrl{1}\gategroup[wires=2,steps=3,style={draw=none,rounded corners,inner xsep=8pt,inner ysep=-1pt, fill=cyan!10}, background,label style={rounded corners,label position=below, yshift=-0.37cm, fill=cyan!10}, background]{$_{P_j}$} & \measure[style={fill=cyan!10}]{_{\langle {\texttt{X}}\rangle,\ \texttt{b}_2}}\\
            		\lstick[]{$_{\ket{\vartheta}}$}  & \targ{} & \qw & \gate[style={fill=cyan!10}]{_{\texttt{X}^{\texttt{b}_1}}} & \qw
            	\end{quantikz}
            	$\ \ \equiv$
            	\begin{quantikz}[thin lines, row sep={0.7cm,between origins}, column sep={0.4cm}]
                    \lstick[]{$_{\ket{\varphi}}$}&\ctrl{1}\gategroup[wires=1,steps=1,style={draw=none,rounded corners,inner xsep=8pt,inner ysep=5pt, fill=lime!10}, background, label style={rounded corners,label position=above,  yshift=0.08cm, fill=lime!10}, background]{$_{P_i}$} & \qw\\
                    \lstick[]{$_{\ket{\vartheta}}$}& \targ{}\gategroup[wires=1,steps=1,style={draw=none,rounded corners,inner xsep=8pt,inner ysep=2.5pt, fill=cyan!10}, background,label style={rounded corners,label position=below, yshift=-0.37cm, fill=cyan!10}, background]{$_{P_j}$} & \qw
                \end{quantikz}
                \caption{Tele-gate performing $\land(\texttt{X})$ between qubits belonging different processors.}
                \label{fig:rcx}
            \end{figure}
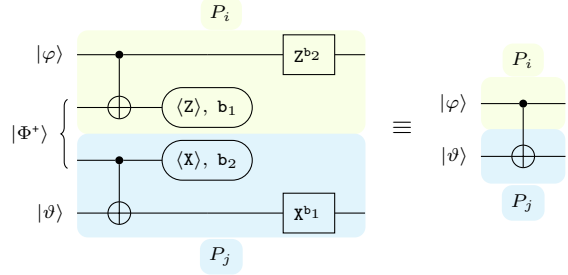
        The teleportation protocol results to be a basic example of a wider class of teleporting protocols. It is in fact possible to, not only teleport state, but entire operations. For this reason the procedure we report here can be also referred as \textit{tele-gate}.
        
        We already discussed the importance of the $\land(\texttt{X})$ operator for quantum computation in Section \ref{sec:qlogic}. We here report a way to perform such an operation between states belonging to different processors by means of non-local operations. In fact, since we are under the assumption the processors inter-connectivity is fundamentally based on distributed entangled states, a $\land(\texttt{X})$ operator can be implemented within a few steps. These steps are shown in Figure \ref{fig:rcx}.

        Notice that performing a non-local operation is not limited to the consumption of a Bell state $\ket{\Phi^{+}}$. Rather, as discussed in Sec. \ref{sec:importance}, one can use any Bell state. For the sake of simplicity, we will always refer to $\ket{\Phi^{+}}$ states.

         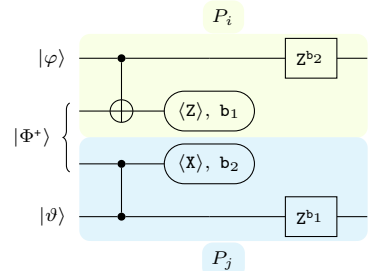
\begin{wrapfigure}{r}{7cm}
                \centering
                \begin{quantikz}[thin lines, row sep={0.7cm,between origins}, column sep={0.4cm}]
            		\lstick[]{$_{\ket{\varphi}}$}& \ctrl{1}\gategroup[wires=2,steps=3,style={draw =none,rounded corners,inner xsep=8pt,inner ysep=-1pt, fill=lime!10}, background, label style={rounded corners,label position=above,  yshift=0.08cm, fill=lime!10}, background]{$_{P_i}$} &\qw & \gate[style={fill=lime!10}]{_{\texttt{Z}^{\texttt{b}_2}}} &  \qw \\
            		\lstick[2]{$_{\ket{\Phi^{\texttt{+}}}}$}& \targ{} & \measure[style={fill=lime!10}]{_{\langle {\texttt{Z}}\rangle,\ \texttt{b}_1}} \\
            		& \ctrl{1}\gategroup[wires=2,steps=3,style={draw=none,rounded corners,inner xsep=8pt,inner ysep=-1pt, fill=cyan!10}, background,label style={rounded corners,label position=below, yshift=-0.37cm, fill=cyan!10}, background]{$_{P_j}$} & \measure[style={fill=cyan!10}]{_{\langle {\texttt{X}}\rangle,\ \texttt{b}_2}}\\
            		\lstick[]{$_{\ket{\vartheta}}$}  & \control{} & \qw & \gate[style={fill=cyan!10}]{_{\texttt{Z}^{\texttt{b}_1}}} & \qw
            	\end{quantikz}
                \caption{Tele-gate performing $\land(\texttt{Z})$ between qubits belonging to different processors.}
                \label{fig:rcz}
                \hrulefill
            \end{wrapfigure}
        Similarly to the teleportation protocol -- see Figure \ref{fig:teleport} --, there is a step meant two inject the control and target states within the entangled state. Thanks to this step, some quantum information is exchanged between the two parties. However, to ensure the equivalence with the subject operator, a measurement step is necessary. Hence, the state that was originally fully entangled in $\ket{\Phi^{\texttt{+}}}$ is measured in
        the orthogonal basis $\texttt{Z}\otimes\texttt{X}$. The output $\texttt{b}_1,\texttt{b}_2$ is then subject to a cross communication through classical channels and eventually used to perform Pauli corrections, i.e., the last step of the procedure: $\texttt{Z}^{\texttt{b}_2}$ over the control qubit and $\texttt{X}^{\texttt{b}_1}$ over the target qubit.

            Mindful of Section \ref{sec:qlogic}, we don't need further non-local operations in order to achieve universal computation. However, it is useful to know that other tele-gates are certainly possible. E.g. Figure \ref{fig:rcz} shows the procedure to perform a $\land(\texttt{Z})$ non-local operator.

        \subsection{Entanglement swap}
        Here we report a protocol known as \textit{entanglement swap}. It is a promising procedure as it models scalable distributed architecture. Figure \ref{fig:eswap} shows the parties involved and their actions. Specifically, assume that two processors $P_i,P_j$ cannot rely on a direct inter-connection through an entanglement pair. However they both share an interconnection with a \textit{middle} processor $P_k$. The middle processor has therefore stored in his memory two half of entangled pairs. By means of a local $\land(\texttt{X})$ on his system, $P_k$ inter-connects $P_i$ with $P_j$. As usual, an orthogonal measurement $\texttt{X}\otimes\texttt{Z}$ is necessary to apply eventual corrections over $P_i$ and $P_j$, which at the end of the procedure share a fully entangled state, ensuring the new inter-connection.
        
            \begin{wrapfigure}{r}{7cm}
                \centering
                \begin{quantikz}[thin lines, row sep={0.7cm,between origins}, column sep={0.4cm}]
            		\lstick[2]{$_{\ket{\Phi^{\texttt{+}}}}$}& \qw\gategroup[wires=1,steps=3,style={draw=none,rounded corners,inner xsep=8pt,inner ysep=-1pt, fill=lime!10}, background, label style={rounded corners,label position=above,  yshift=0.08cm, fill=lime!10}, background]{$_{P_i}$} & \qw  & \gate[style={fill=lime!10}]{_{\texttt{Z}^{\texttt{b}_2}}} & \qw\rstick[4]{$_{\ket{\Phi^{\texttt{+}}}}$} \\
            		 & \ctrl{1}\gategroup[wires=2,steps=2,style={draw=none,rounded corners,inner xsep=8pt,inner ysep=-1pt, fill=orange!10}, background, label style={rounded corners,label position=above, xshift = 1.8cm,  yshift=-0.85cm, fill=orange!10}, background]{$_{P_k}$}  & \measure[style={fill=orange!10}]{_{\langle {\texttt{X}}\rangle,\ \texttt{b}_1}} \\
            		\lstick[2]{$_{\ket{\Phi^{\texttt{+}}}}$}& \targ{} & \measure[style={fill=orange!10}]{_{\langle {\texttt{Z}}\rangle,\ \texttt{b}_2}}\\
            		& \qw\gategroup[wires=1,steps=3,style={draw=none,rounded corners,inner xsep=8pt,inner ysep=-1pt, fill=cyan!10}, background,label style={rounded corners,label position=below, yshift=-0.37cm, fill=cyan!10}, background]{$_{P_j}$} & \qw & \gate[style={fill=cyan!10}]{_{\texttt{X}^{\texttt{b}_1}}} & \qw
            	\end{quantikz}
                \caption{The entanglement swap protocol inter-connects two processors by means of an intermediate one, which has direct connection with both of them.}
                \label{fig:eswap}
                \hrulefill
            \end{wrapfigure}
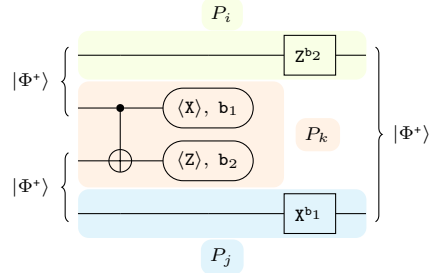
        Depending on the time model adopted when dealing with this procedure, the inter-connectivity among the processors find different treating. For example, the work done in \cite{ferrari2021compiler} has a more dynamic-like approach, making a distinction between \textit{link} and \textit{virtual link}. Such a choice probably comes from an interest in modeling a network of quantum technologies, where it is more common to deal with \textit{online} combinatorial problems \cite{hentenryck2006online}. Instead, our focus here is to smartly model distributed architecture meant to perform algorithms. This brings us to see the inter-connectivity within a more static time model; which translates into a simpler modeling for the connectivity. We explain this in detail within next sub-section \ref{sec:e-path}.
        \subsection{Entanglement paths}
        \label{sec:e-path}
        As our focus is on the treating of distributed quantum computing, we here provides a non-local $\land(\texttt{X})$, which makes use of entanglement swaps.
        The circuit in Figure \ref{fig:rcx_swap} comes from the combination of the basic implementation of a non-local $\land(\texttt{X})$ -- see Section \ref{sec:rcx} -- with the entanglement swap protocol.
            \begin{figure}[h]
                \centering
                \begin{quantikz}[thin lines, row sep={0.7cm,between origins}, column sep={0.4cm}]
                    \lstick[]{$_{\ket{\varphi}}$} & \ctrl{1}\gategroup[wires=2,steps=3,style={draw=none,rounded corners,inner xsep=8pt,inner ysep=-1pt, fill=lime!10}, background, label style={rounded corners,label position=above,  yshift=0.08cm, fill=lime!10}, background]{$_{P_i}$}& \qw &\gate[style={fill=lime!10}]{_{\texttt{Z}^{\texttt{b}_2\oplus \texttt{b}_4}}} & \qw\\
            		\lstick[2]{$_{\ket{\Phi^{\texttt{+}}}}$}& \targ{} & \measure[style={fill=lime!10}]{_{\langle {\texttt{Z}}\rangle,\ \texttt{b}_1}} \\
            		& \ctrl{1}\gategroup[wires=2,steps=2,style={draw=none,rounded corners,inner xsep=8pt,inner ysep=-1pt, fill=orange!10}, background, label style={rounded corners,label position=above, xshift = 1.8cm,  yshift=-0.85cm, fill=orange!10}, background]{$_{P_k}$} & \measure[style={fill=orange!10}]{_{\langle {\texttt{X}}\rangle,\ \texttt{b}_2}}\\
            		\lstick[2]{$_{\ket{\Phi^{\texttt{+}}}}$}& \targ{}  & \measure[style={fill=orange!10}]{_{\langle {\texttt{Z}}\rangle,\ \texttt{b}_3}}\\
            		& \ctrl{1}\gategroup[wires=2,steps=3,style={draw=none,rounded corners,inner xsep=8pt,inner ysep=-1pt, fill=cyan!10}, background,label style={rounded corners,label position=below, yshift=-0.37cm, fill=cyan!10}, background]{$_{P_j}$} & \measure[style={fill=cyan!10}]{_{\langle {\texttt{X}}\rangle,\ \texttt{b}_4}}\\
            		\lstick[]{$_{\ket{\vartheta}}$}& \targ{} & \qw & \gate[style={fill=cyan!10}]{_{\texttt{X}^{\texttt{b}_1\oplus \texttt{b}_3}}} & \qw
            	\end{quantikz}
                $\ \ \equiv$
            	\begin{quantikz}[thin lines, row sep={0.7cm,between origins}, column sep={0.4cm}]
                    \lstick[]{$_{\ket{\varphi}}$}&\ctrl{1}\gategroup[wires=1,steps=1,style={draw=none,rounded corners,inner xsep=8pt,inner ysep=5pt, fill=lime!10}, background, label style={rounded corners,label position=above,  yshift=0.08cm, fill=lime!10}, background]{$_{P_i}$} & \qw\\
                    \lstick[]{$_{\ket{\vartheta}}$}& \targ{}\gategroup[wires=1,steps=1,style={draw=none,rounded corners,inner xsep=8pt,inner ysep=2.5pt, fill=cyan!10}, background,label style={rounded corners,label position=below, yshift=-0.37cm, fill=cyan!10}, background]{$_{P_j}$} & \qw
                \end{quantikz}
                \caption{Tele-gate $\land(\texttt{X})$ by means of entanglement swap.}
                \label{fig:rcx_swap}
            \end{figure}
        It is important to notice that all the measurements happen at the same time. Even more important to know is that this result can be generalized to any number of middle processors $P_{k_1},P_{k_2},\dots,P_{k_m}$. For this reason we refer to $\{P_i, P_{k_1},P_{k_2},\dots,P_{k_m}, P_j\}$ as an \textit{entanglement path} of length $m+1$. We now give an inductive proof for this result.
        
        \begin{theorem}
        \label{th:path}
            An entanglement path $\{P_{i_1}, P_{i_2}, \dots, P_{i_m}\}$ has an implementation with depth 4.
        \end{theorem}
        \begin{proof}
           Consider an entanglement path of length 2. A naive realization consists on putting in strict sequence two entanglement swaps:
           \begin{center}
               \begin{quantikz}[thin lines,row sep={0.65cm,between origins}, column sep ={0.2cm}]
        		 \lstick[2]{$_{\ket{\Phi^{\texttt{+}}}}$}& \qw & \qw  & \gate[style={fill=green!10}]{_{\texttt{Z}^{\texttt{b}_1}}} & \qw  & \qw &\gate[style={fill=green!10}]{_{\texttt{Z}^{\texttt{b}_3}}}\\
        		& \ctrl{1} &  \measure[style={fill=gray!10}]{_{\langle {\texttt{X}}\rangle, \texttt{b}_1}} \\
        	    \lstick[2]{$_{\ket{\Phi^{\texttt{+}}}}$}& \targ{} &  \measure[style={fill=gray!10}]{_{\langle {\texttt{Z}}\rangle, \texttt{b}_2}}\\
        		&\qw& \qw &\gate[style={fill=red!10}]{_{\texttt{X}^{\texttt{b}_2}}} & \ctrl{1}  & \measure[style={fill=gray!10}]{_{\langle {\texttt{X}}\rangle, \texttt{b}_3}}\\
        		 \lstick[2]{$_{\ket{\Phi^{\texttt{+}}}}$} &\qw&\qw &\qw & \targ{} & \measure[style={fill=gray!10}]{_{\langle {\texttt{Z}}\rangle, \texttt{b}_4}}\\
        		 &\qw & \qw &\qw & \qw & \qw   & \gate[style={fill=red!10}]{_{\texttt{X}^{\texttt{b}_4}}}
        	\end{quantikz}
           \end{center}
           Pauli gates are the only ones we are going to optimize; since the others are independent and no optimization can be applied. What follows is the base case for the induction:
            \begin{center}
                \begin{quantikz}[thin lines,row sep={0.65cm,between origins}, column sep ={0.2cm}]
            		& \gate[style={fill=green!10}]{_{\texttt{Z}^{\texttt{b}_{1}}}} & \qw & \qw & \gate[style={fill=green!10}]{_{\texttt{Z}^{\texttt{b}_{3}}}} & \qw\\
            		&\gate[style={fill=green!10}]{_{\texttt{X}^{\texttt{b}_{2}}}} & \ctrl{1} & \measure[style={fill=gray!10}]{_{\langle {\texttt{X}}\rangle, \texttt{b}_3}}\\
            		& \qw  &\targ{} & \measure[style={fill=gray!10}]{_{\langle {\texttt{Z}}\rangle, \texttt{b}_4}}\\
            		& \qw & \qw & \qw  & \gate[style={fill=red!10}]{_{\texttt{X}^{\texttt{b}_{4}}}} & \qw
            	\end{quantikz}
            	$\equiv$
            	\begin{quantikz}[thin lines,row sep={0.65cm,between origins}, column sep ={0.2cm}]
            		& \qw & \qw & \gate[style={fill=green!10}]{_{\texttt{Z}^{\texttt{b}_1\oplus \texttt{b}_3}}}&\\
            		& \ctrl{1} &\measure[style={fill=gray!10}]{_{\langle {\texttt{X}}\rangle, \texttt{b}_3}}\\
            		& \targ{} & \measure[style={fill=gray!10}]{_{\langle {\texttt{Z}}\rangle, \texttt{b}_4}}\\
            		& \qw & \qw & \gate[style={fill=red!10}]{_{\texttt{X}^{\texttt{b}_2\oplus \texttt{b}_4}}} &
            	\end{quantikz}
                \end{center} 
            The r.h.s. of the above equation has post-processing composed by $\texttt{Z}^{\texttt{b}_1 \oplus \texttt{b}_3}$ on first qubit and $\texttt{X}^{\texttt{b}_2 \oplus \texttt{b}_4}$ on last qubit. Notice that the measurements are independent from other operations.
           
            By assuming that such a shape is preserved in the inductive step, we show that this transformation can be applied to any length $m$:
            \begin{center}
                \begin{quantikz}[thin lines,row sep={0.65cm,between origins}, column sep ={0.2cm}]
        		& \gate[style={fill=green!10}]{_{\texttt{Z}^{\texttt{b}_1 \oplus \texttt{b}_3 \oplus\cdots \oplus \texttt{b}_{2m\texttt{-}3}}}}& \qw&\qw &\gate[style={fill=green!10}]{_{\texttt{Z}^{\texttt{b}_{2m\texttt{-}1}}}} &\\
        		& \gate[style={fill=red!10}]{_{\texttt{X}^{\texttt{b}_1 \oplus \texttt{b}_3 \oplus\cdots \oplus \texttt{b}_{2m\texttt{-}2}}}} & \ctrl{1}& \measure[style={fill=gray!10}]{_{\langle {\texttt{X}}\rangle,\texttt{b}_{2m\texttt{-}1}}}\\
        		& \qw & \targ{} & \measure[style={fill=gray!10}]{_{\langle {\texttt{Z}}\rangle, \texttt{b}_{2m\phantom{\texttt{-}1}}}}\\
        		& \qw & \qw &\qw &\gate[style={fill=red!10}]{_{\texttt{X}^{\texttt{b}_{2m\phantom{\texttt{-}1}}}}} &
        	\end{quantikz}
        	$\equiv$
        	\begin{quantikz}[thin lines,row sep={0.65cm,between origins}, column sep ={0.2cm}]
        		& \qw&\qw &\gate[style={fill=green!10}]{_{\texttt{Z}^{\texttt{b}_1 \oplus \texttt{b}_3 \oplus\cdots \oplus \texttt{b}_{2m\texttt{-}1}}}} &\\
        		& \ctrl{1}& \measure[style={fill=gray!10}]{_{\langle {\texttt{X}}\rangle, \texttt{b}_{2m\texttt{-}1}}}\\
        		& \targ & \qw & \measure[style={fill=gray!10}]{_{\langle {\texttt{Z}}\rangle, \texttt{b}_{2m\phantom{\texttt{-}1}}}}\\
        		& \qw & \qw  &\gate[style={fill=red!10}]{_{\texttt{X}^{\texttt{b}_1 \oplus \texttt{b}_3 \oplus\cdots \oplus \texttt{b}_{2m\phantom{\texttt{-}1}}}}} &
        	\end{quantikz}
            \end{center}
            This proves that we can always consider an entanglement path $\{P_{i_1}, P_{i_2}, \dots, P_{i_m}\}$ to have circuit depth 4.
        \end{proof}
        
        We just showed an efficient implementation for the entanglement path. Now we do one last step to exploit such a result and performing a generalized remote operation efficiently. 
        
        \begin{theorem}
            A tele-gate of entanglement path $\{P_{i_1}, P_{i_2}, \dots, P_{i_{m\texttt{+}2}}\}$ has depth 4.
        \end{theorem}
        \begin{proof}
        The theorem above allows us to assume that, to perform a remote operation by using a path of length $m$, the computing qubits interact only with two communications qubits and depend only by Pauli operations $\texttt{Z}^{\texttt{b}_1 \oplus \texttt{b}_3 \oplus\cdots \oplus \texttt{b}_{2m-1}}$ and $\texttt{X}^{\texttt{b}_2 \oplus \texttt{b}_4 \oplus\cdots \oplus \texttt{b}_{2m}}$. We can further \textit{propagate} such operations as follows:
        \begin{center}
            \begin{quantikz}[thin lines,row sep={0.65cm,between origins}, column sep ={0.2cm}]
        		& \qw& \ctrl{1} & \qw &\gate[style={fill=green!10}]{_{\texttt{Z}^{\texttt{b}_{2m\texttt{+}2}}}} & \\
        		&\gate[style={fill=green!10}]{_{\texttt{Z}^{\texttt{b}_1 \oplus \texttt{b}_3 \oplus\cdots \oplus \texttt{b}_{2m\texttt{-}1}}}}& \targ{} & \measure[style={fill=gray!10}]{_{\langle {\texttt{Z}}\rangle, \texttt{b}_{2m\texttt{+}1}}}\\
        		&\gate[style={fill=red!10}]{_{\texttt{X}^{\texttt{b}_2 \oplus \texttt{b}_4 \oplus\cdots \oplus \texttt{b}_{2m\phantom{\texttt{-}1}}}}}& \ctrl{1} & \measure[style={fill=gray!10}]{_{\langle {\texttt{X}}\rangle, \texttt{b}_{2m\texttt{+}2}}}\\
        		&\qw& \targ{} & \qw & \gate[style={fill=red!10}]{_{\texttt{X}^{\texttt{b}_{2m\texttt{+}1}}}} &
        	\end{quantikz}
        	$\equiv$
        	\begin{quantikz}[thin lines,row sep={0.65cm,between origins}, column sep ={0.2cm}]
        		& \ctrl{1} & \qw  & \gate[style={fill=green!10}]{_{\texttt{Z}^{\texttt{b}_1 \oplus \texttt{b}_3 \oplus\cdots \oplus \texttt{b}_{2m\texttt{-}1} \oplus \texttt{b}_{2m\texttt{+}2}}}} & \\
        		& \targ{} & \measure[style={fill=gray!10}]{_{\langle {\texttt{Z}}\rangle, \texttt{b}_{2m\texttt{+}1}}}\\
        		& \ctrl{1} & \measure[style={fill=gray!10}]{_{\langle {\texttt{X}}\rangle, \texttt{b}_{2m\texttt{+}2}}}\\
        		& \targ{} & \qw & \gate[style={fill=red!10}]{_{\texttt{X}^{\texttt{b}_2 \oplus \texttt{b}_4 \oplus\cdots \oplus \texttt{b}_{2m\phantom{\texttt{+}1}} \oplus \texttt{b}_{2m\texttt{+}1}}}} &
        	\end{quantikz}
        \end{center}
        In this way the measurements are independent and the depth of the circuit is not increased.
        \end{proof}

\subsection{Entanglement Trees}
\label{sec:rcxx}
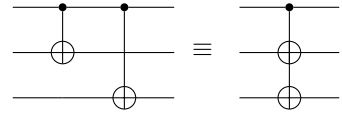
\begin{wrapfigure}{r}{5cm}
    \centering
    \begin{quantikz}[thin lines,row sep={0.6cm,between origins}]
        &\ctrl{1} & \ctrl{2} & \qw\\
        &\targ{} & \qw & \qw\\
        & \qw & \targ{} & \qw 
    \end{quantikz}
    $\ \equiv$
    \begin{quantikz}[thin lines,row sep={0.6cm,between origins}]
        &\ctrl{2} & \qw\\
        &\targ{} & \qw\\
        & \targ{} & \qw 
    \end{quantikz}
    \caption{Circuit representation of the equivalence \eqref{eq:cxx}.}
    \label{fig:cxx}
	\hrulefill
\end{wrapfigure}
Driven again by the aim of finding smart strategies to perform non-local gates with low consumption of entanglement links; we now investigate a way to perform multiple operations \textbf{by means of the same link}. To get a first intuition of what we are going to show, consider the following equivalence:
\begin{equation}
    \label{eq:cxx}
    \land(\mathds{1}\otimes\texttt{X})\cdot\land(\texttt{X}\otimes\mathds{1}) \equiv \land(\texttt{X}\otimes\texttt{X});
\end{equation}
which has circuit representation reported in Figure \ref{fig:cxx}. Behind its simplicity, equivalence \eqref{eq:cxx} gives us a different perspective to implement multiple non-local operations by means of one entanglement link. This happens for the case of two processors $P_i,P_j$, but it can be generalized. In fact, as showed in \cite{andres2019automated,yimsiriwattana2004generalized}, any $\land(\texttt{X}^{\otimes m})$ -- with system of dimensionality $n>m$ and control qubit being stored in a different processor w.r.t. the target one -- admits an implementation that needs one entanglement link only. 

\begin{figure}[h]
    \centering
    \begin{quantikz}[transparent,thin lines, row sep={0.7cm,between origins}, column sep={0.4cm}]
            \lstick[]{$_{\ket{\varphi}}$} & \ctrl{1}\gategroup[wires=2,steps=3,style={draw=none,rounded corners,inner xsep=8pt,inner ysep=-1pt, fill=lime!10}, background, label style={rounded corners,label position=above, xshift = 2.5cm,  yshift=-0.85cm, fill=lime!10}, background]{$_{P_i}$} & \qw  & \gate{_{\texttt{Z}^{\texttt{b}_2\oplus\texttt{b}_4}}} & \qw\\
	    \lstick[2]{$_{\ket{\Phi^{\texttt{+}}}}$} & \targ{}  & \measure[style={fill=lime!10}]{_{\langle \texttt{Z}\rangle,\texttt{b}_1}} &\\
		& \ctrl{2}\gategroup[wires=3,steps=3,style={draw=none,rounded corners,inner xsep=8pt,inner ysep=-1pt, fill=orange!10}, background, label style=        {rounded corners,label position=above, xshift = 2.5cm,  yshift=-1.2cm, fill=orange!10}, background]{$_{P_k}$} & \measure[style={fill=orange!10}]          {_{\langle \texttt{X}\rangle,\texttt{b}_2}} & \\
		\lstick[]{$_{\ket{\psi}}$} &  \targ{}  &\qw & \gate{_{\texttt{X}^{\texttt{b}_1}}} & \qw\\
		\lstick[2]{$_{\ket{\Phi^{\texttt{+}}}}$}& \targ{}  &\measure[style={fill=orange!10}]{_{\langle \texttt{Z}\rangle,\texttt{b}_3}} &\\
		& \ctrl{1}\gategroup[wires=2,steps=3,style={draw=none,rounded corners,inner xsep=8pt,inner ysep=-1pt, fill=cyan!10}, background, label style={rounded corners,label position=above, xshift = 2.5cm,  yshift=-0.85cm, fill=cyan!10}, background]{$_{P_j}$}  & \measure[style={fill=cyan!10}]{_{\langle \texttt{X}\rangle,\texttt{b}_4}} &\\
		\lstick[]{$_{\ket{\vartheta}}$} & \targ{}  & \qw  & \gate{_{\texttt{X}^{\texttt{b}_1 \oplus \texttt{b}_3}}} & \qw
    \end{quantikz}
    $\ \ \ \equiv$
    \begin{quantikz}[thin lines, row sep={0.7cm,between origins}, column sep={0.4cm}]
        \lstick[]{$_{\ket{\varphi}}$}&\ctrl{2}\gategroup[wires=1,steps=1,style={draw=none,rounded corners,inner xsep=8pt,inner ysep=5pt, fill=lime!10}, background, label style={rounded corners,label position=above,  xshift = 1cm, yshift = -0.5cm, fill=lime!10}, background]{$_{P_i}$} & \qw\\
        \lstick[]{$_{\ket{\psi}}$}& \targ{}\gategroup[wires=1,steps=1,style={draw=none,rounded corners,inner xsep=8pt,inner ysep=2.5pt, fill=orange!10}, background,label style={rounded corners,label position=below, xshift = 1cm, yshift = 0.225cm, fill=orange!10}, background]{$_{P_k}$} & \qw\\
        \lstick[]{$_{\ket{\vartheta}}$}& \targ{}\gategroup[wires=1,steps=1,style={draw=none,rounded corners,inner xsep=8pt,inner ysep=2.5pt, fill=cyan!10}, background,label style={rounded corners,label position=below, xshift = 1cm, yshift = 0.225cm, fill=cyan!10}, background]{$_{P_j}$} & \qw
    \end{quantikz}
    \caption{Remote implementation of $\land(\texttt{X}\otimes\texttt{X})$.}
    \label{fig:2rcxx}
\end{figure}
Stemming from equation \eqref{eq:cxx}, the two target systems are essentially independent, up to the common control qubit. Hence, there is no reason to restrict them to be part of the same processor. For the case under consideration  -- i.e. $\land(\texttt{X}\otimes\texttt{X})$ --, the maximum number of processors is three. Figure \ref{fig:2rcxx} shows the circuit protocol to perform $\land(\texttt{X}\otimes\texttt{X})$, where the system is spread over three processors.

Thanks to the entanglement path definition given in Sec. \ref{sec:e-path}, such a protocol runs efficiently also for system linked only by an entanglement path.

We will use this technique in Ch. \ref{ch:compile} to minimize entanglement links consumption, up to gain optimal solutions in some interesting scenarios.

\printbibliography[title=References,heading=subbibintoc]
\end{refsection}

%% file: logic/logical_computing.tex
\begin{refsection}
\chapter{Quantum noise and how to handle it}
\label{ch:noise}
\thispagestyle{empty}
\newpage

\section{Quantum noise}
    Assuming a good characterizing model for the noise affecting quantum information can be challenging. A highly generic one can be described as the evolution $\mathcal{N}$ on an $d$-dimensional system, where a state of interest $\sigma$ evolves together with the environment in a state $\rho$ \cite{nielsen2010quantum}:
            \[\mathcal{N}(\sigma) = \texttt{Tr}_{\texttt{env}}\big(\texttt{U}(\sigma \otimes \rho)\texttt{U}^{\dagger}\big).\]
            Let $\{\ket{e_k}\}_{k}$ be an orthonormal basis for the environment. One can assume the environment system as being in the state $\rho = \ket{e_0}\bra{e_0}$. This assumption is non-restrictive as the basis is generic and if the considered environment was not pure, one can always introduce an extra \textit{reference} system to purify $\rho$. It follows that
            \[\mathcal{N}(\sigma) = \sum_{k}\bra{e_k}\texttt{U}(\sigma \otimes \ket{e_0}\bra{e_0})\texttt{U}^{\dagger}\ket{e_k}.\]
            By defining $\texttt{N}_k = \bra{e_k}\texttt{U}\ket{e_0}$, known as \textit{Kraus} operator, $\mathcal{N}$ becomes
            \begin{equation}
                \label{eq:kraus}
                \mathcal{N}(\sigma) = \sum_{k}\texttt{N}_k\sigma \texttt{N}_k^{\dagger}.
            \end{equation}
            Decoherence is the most problematic noise affecting information. This can be represented with the Kraus formalism in terms of Pauli operators acting on each qubit independently \cite{dennis2002topological,mancini2020quantum, berg2022probabilistic}, hence the set $\{\texttt{N}_k\}_{k}$ as the form $\{\sqrt{\texttt{p}_k}\ \texttt{E}_k\}_{k}$, such that $\sum_{k} \texttt{p}_k = 1$ and $\texttt{E}_k$ belongs to the \textit{Pauli group} $\mathbb{P} \equiv \{\pm \mathds{1},\pm \texttt{i}\mathds{1},\pm \texttt{X},\pm \texttt{i}\texttt{X},\pm \texttt{Y},\pm \texttt{i}\texttt{Y},\pm \texttt{Z},\pm \texttt{i}\texttt{Z}\}^{\otimes d}$.
    
    A further refinement is achievable. In fact, from an information perspective, the global phase has no relevance. Hence, by considering the ``quotient'' group\footnote{Being aware of the equivalence $\texttt{XZ} \equiv -\texttt{iY}$.} 
    \begin{equation}
    \label{eq:quot-pauli}
        \mathbb{E} \equiv \{\mathds{1},\texttt{X},\texttt{Z}, \texttt{XZ}\}^{\otimes d} \subset \mathbb{P},
    \end{equation}
    one can assume $\{\texttt{E}_k\}_{k} \subset \mathbb{E}$.

    A comparison among Pauli noise evolution for a 2-dimensional state is given in Fig. \ref{fig:pauli-bloch}.
    \begin{figure}[h]
            \centering
            \begin{subfigure}{0.3\textwidth}
                \centering
                \includegraphics[scale=0.1]{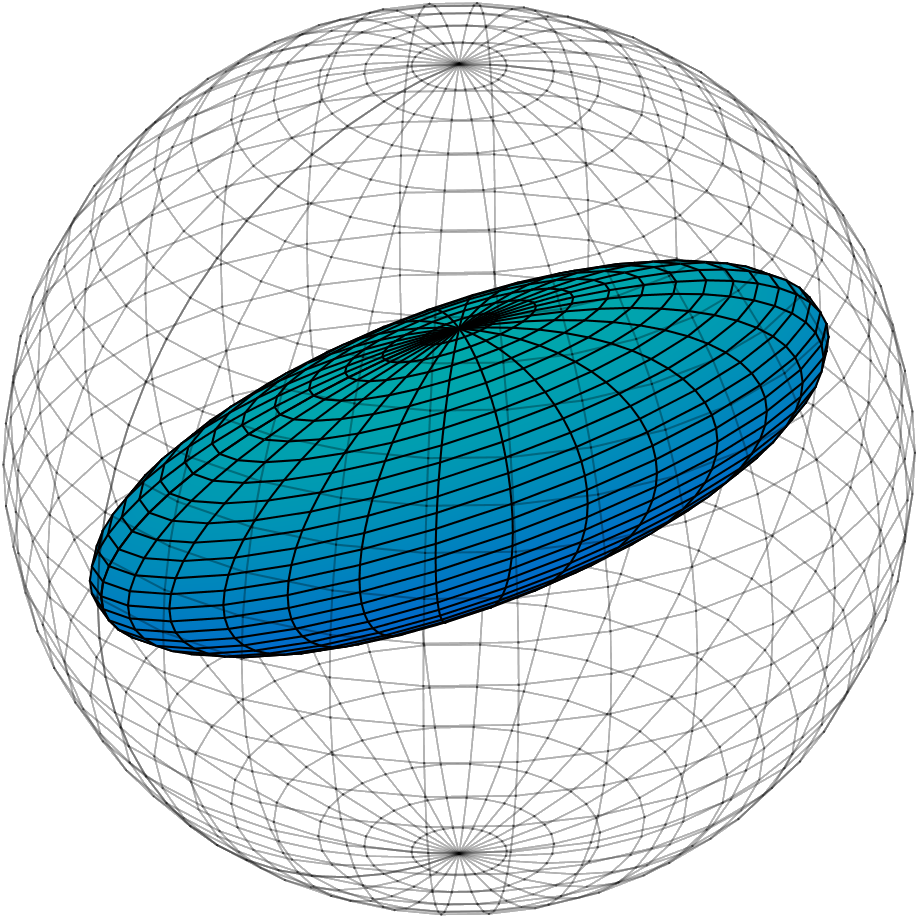}
                \caption*{$\{\texttt{X},\mathds{1}\}$ and $\texttt{p} = 0.3$.}
            \end{subfigure}
            \begin{subfigure}{0.3\textwidth}
                \centering
                \includegraphics[scale=0.1]{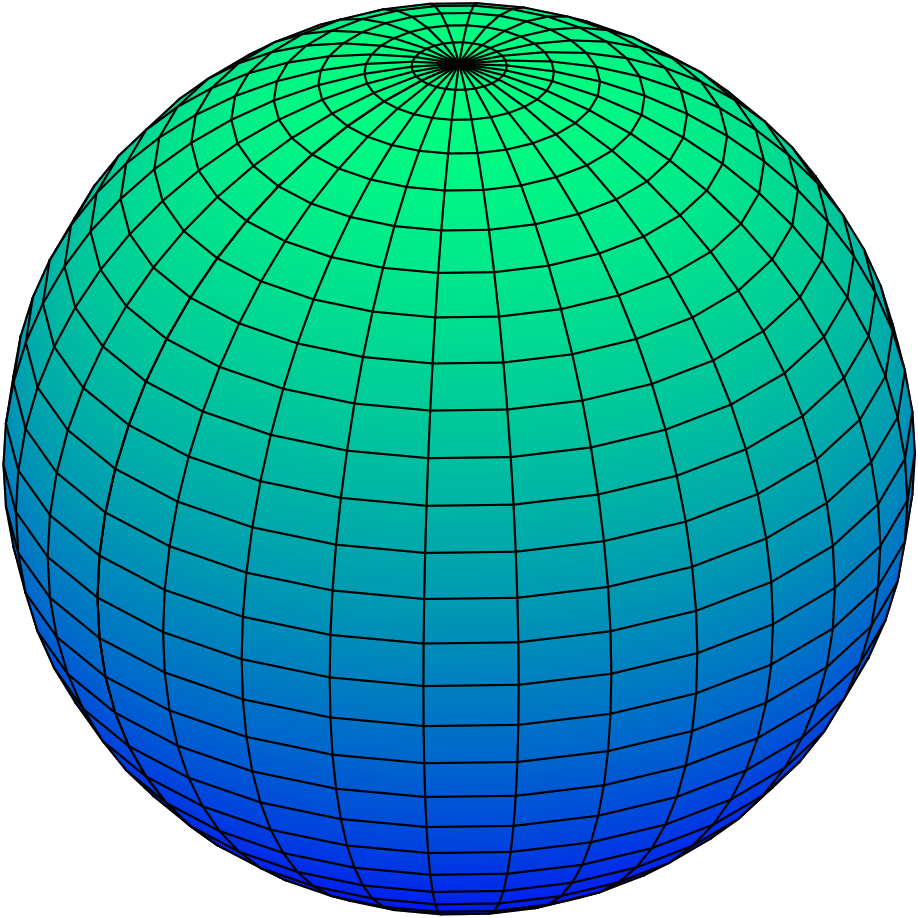}
                \caption*{$\{\mathds{1}\}$, i.e., a noiseless evolution.}
            \end{subfigure} 
           \begin{subfigure}{0.3\textwidth}
           \centering
                \includegraphics[scale=0.1]{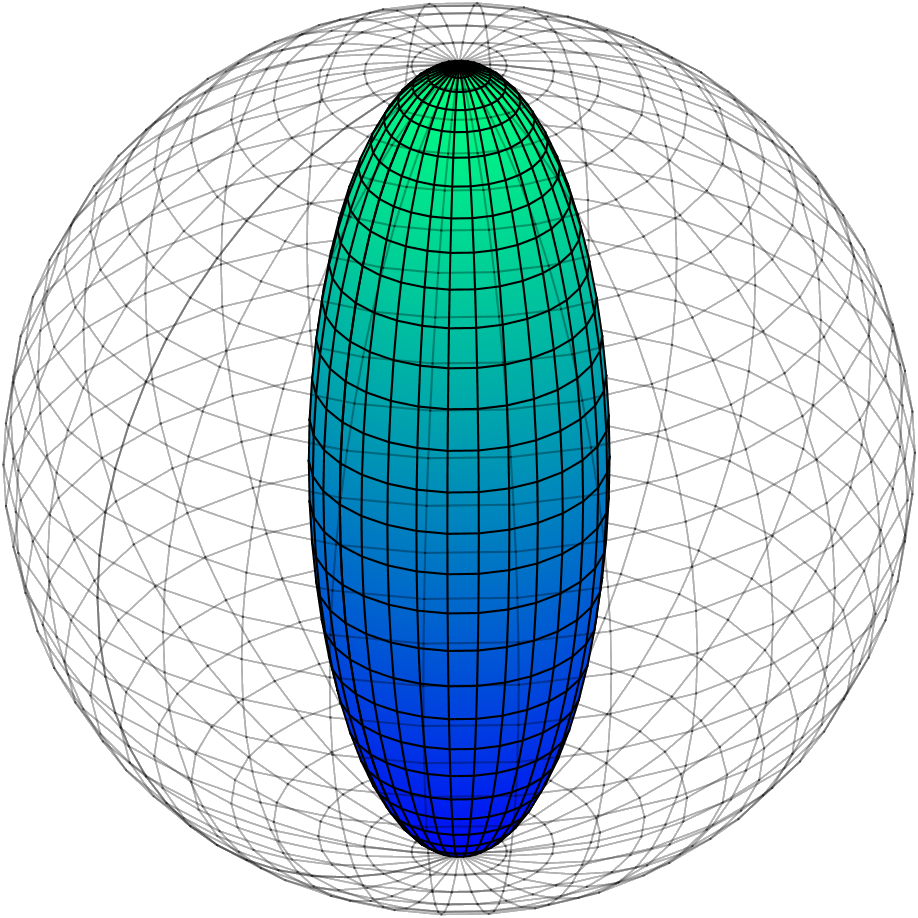}
                \caption*{$\{\texttt{Z},\mathds{1}\}$ and $\texttt{p} = 0.3$.}
            \end{subfigure}
            \caption{Example of Pauli noise for 2-dimensional Hilbert space in Bloch sphere representation.}
            \label{fig:pauli-bloch}
    \end{figure}
        
    \section{Estimating an evolution}
        \label{sec:qpt}
            For a given evolution $\mathcal{N}$, another useful representation of its action on a state $\sigma$ is through its \textit{Choi matrix} \cite{choi1975completely} $\varsigma_{_\mathcal{N}}$. When the Kraus operators of $\mathcal{N}$ are known -- i.e. $\{\sqrt{\texttt{p}_k}\ \texttt{E}_k\}_k$ -- the Choi matrix is immediately defined as \cite{de2019fault}:
            \begin{equation}
                \label{eq:choi}
                \varsigma_{_\mathcal{N}} = \sum_{n,m}\sqrt{\texttt{p}_{n}\texttt{p}_{m}} |\texttt{E}_n\rangle\rangle\langle\langle \texttt{E}_m|,
            \end{equation}
            with $|\texttt{E}\rangle\rangle$ being the \textit{vectorization} of $\texttt{E}$.
            When the evolution is \textit{unitary} -- i.e. $\mathcal{U}(\sigma) = \texttt{U}\sigma\texttt{U}^\dagger$ --, equation \eqref{eq:choi} simplifies to 
            \begin{equation}
                \label{eq:u_choi}
            \varsigma_{_\mathcal{U}} = |\texttt{U}\rangle\rangle\langle\langle \texttt{U}|.
            \end{equation}
            
            Working with Choi matrices allows to compute the fidelity of some unknown evolution w.r.t. a target one. For example, given a target unitary evolution $\mathcal{U}$ and an \textit{experimental} evolution $\mathcal{E}$, both operating on $d$-dimensional Hilbert space. The following function is a proper fidelity for the two channels:
            \begin{equation}
                \mathfrak{f}(\mathcal{U},\mathcal{E}) = \frac{1}{d} \cdot \texttt{Tr}(\sqrt{\sqrt{\varsigma_{_\mathcal{U}}}\sqrt{\varsigma_{_\mathcal{E}}}})^2.
            \end{equation}
            
            Estimating $\varsigma_{_\mathcal{E}}$ is an expensive procedure based on a orthogonal set of measurements. A sufficient and common one is the Pauli group. In the circuit model we refer to this procedure as follows
            \begin{center}
                \begin{quantikz}[row sep={0.7cm,between origins}, column sep={0.4cm}]
                    \lstick[]{$_{\sigma}$}&\gate[style={fill=lime!10}]{_{\mathcal{E}}}&\measure[style={fill=lime!10}]{_{\langle {\texttt{X}}, {\texttt{Y}}, {\texttt{Z}}\rangle}}
                \end{quantikz}
            \end{center}
            The circuit above represent a \textit{state tomography} -- i.e., how $\mathcal{E}$ affects $\sigma$. It works by running $k$ sets of experiments to obtain an estimator for $\{\texttt{p}_k\}_k$.
            
            Similarly, a \textit{process tomography} starts from an orthogonal set of input states in order to fully characterize $\varsigma_{\mathcal{E}}$. In the circuit model we express such an estimation as
            \begin{center}
                \begin{quantikz}[row sep={0.7cm,between origins}, column sep={0.4cm}]
                    \meterD[style={fill=lime!10}]{_{\ket{\texttt{0},\texttt{1},\texttt{+},\texttt{i}}}}&\gate[style={fill=lime!10}]{_{\mathcal{E}}} &\measure[style={fill=lime!10}]{_{\langle {\texttt{X}}, {\texttt{Y}},{\texttt{Z}}\rangle}}
                \end{quantikz}
            \end{center}
            Tomography methods are intractable when considering high-dimensional systems.
            

    \section{Noise canceling through indefinite causal orders}
    \label{sec:soco}
    Similarly to what we have done in Sec. \ref{sec:logical-ifo}, we now report the possible advantages coming from the indefinite causal order framework, applied to non-unitaries. Namely, noisy channels are superposed in order to increase the overall capacity \cite{EblSalChi-18, SalEblChi-18,CalCac-20, koudia2022deep}
    
     There are several ways to physically implement an indefinite order. Most of realizations are photonic-based \cite{ProMogAra-15, RubRozFei-17, GuoHuHou-20, FitJonVed-15, GosGiaKew-18}, but it is not the only way. Indeed, within this Section we go over the work in Ref. \cite{cuomo2021experiencing}; presenting an implementation with a programmable technology, based on superconductors \cite{ibmq, Cas-17}.
     

        The Indefinite Causal Orders for two evolution $\mathcal{N} \mapsto \{\texttt{N}_n\}_n$ and ${\mathcal{M}} \mapsto \{\texttt{M}_m\}_m$, is given by \cite{ChiDarPer-13}:
        \begin{equation}
            \label{Eq:ico}
            \mathcal{S}(\sigma, \varkappa) = \sum_{nm} \texttt{S}_{nm} (\sigma \otimes \varkappa) \texttt{S}_{nm}^\dagger,
        \end{equation}
        where $\varkappa$ is a \textit{control state} and $\{\texttt{S}_{nm}\}_{nm}$ denotes the set of Kraus operators such that
        \begin{equation}
        \texttt{S}_{nm} = \texttt{N}_n \texttt{M}_m \otimes \ket{\texttt{0}}\bra{\texttt{0}} + {\texttt{M}}_m \texttt{N}_n  \otimes \ket{\texttt{1}}\bra{\texttt{1}}    
        \end{equation}
        Therefore, under the assumption of errors in $\mathbb{E}$ -- see Eq. \ref{eq:quot-pauli} --, it is also true that
        \begin{equation}
        \texttt{S}_{nm} = \sqrt{\texttt{p}_{n}\texttt{p}_{m}}
            \begin{bmatrix} 
                \texttt{E}_{nm} & \mathbf{0}\\
                \mathbf{0} & {\texttt{E}}_{mn} 
            \end{bmatrix}.
        \end{equation}
        As $\texttt{N}_n \texttt{M}_m = \sqrt{\texttt{p}_{n}\texttt{p}_{m}}\ \texttt{E}_{nm}$ and $\texttt{M}_m \texttt{N}_n = \sqrt{\texttt{p}_{m}\texttt{p}_{n}}\ \texttt{E}_{mn}$, with $\texttt{E}_{nm},\texttt{E}_{mn}$ being in the Pauli group.
    
        As discussed in Sec. \ref{sec:programming}, higher-ordered circuit frameworks are not presented within this thesis. An \textit{oracle} is sufficient to our purpose. 
        \begin{center}
            \begin{quantikz}[thin lines, row sep={0.5cm,between origins}, column sep={0.4cm}]
                \lstick{$_{\sigma}$} & \gate[2,style={fill=violet!10}]{_{\mathcal{S}}} & \qw \rstick[2]{$_{\sum_{nm}\texttt{S}_{nm}(\sigma \otimes \varkappa) \texttt{S}_{nm}^{\dagger}}$}\\
                \lstick{$_{\varkappa}$}&  & \qw
            \end{quantikz}
        \end{center}
        
        
             Our aim here is to report our work published as in \cite{cuomo2021experiencing}. Namely, experiencing and evaluating the indefinite causal order within a \textit{Noisy Intermediate-Scale Quantum} (NISQ) architecture \cite{Preskill2018}, based on superconductors. NISQ architectures are widespread and they promise to be resources of practical interest in the next future. Furthermore, the design is likely to rapidly evolve, also by considering the Indefinite Causal Order as resource. Our hope is to enrich the knowledge on the capabilities of current quantum technologies, with the long-term goal of contributing to shape future architecture designs.
            
            The experiment set-up is meant to witness the communication advantage, resulting from a specific case of the subject evolution. 
            According to \textit{quantum Shannon theory} \cite{Wil-13}, the capacity is a metric to quantify the ability for a noisy channel to convey quantum information, without destroying it. A channel with null capacity destroys the \textit{coherence} of the quantum information, meaning that it is not possible to retrieve the original information. By superposing two or more null capacity channels, the result is a new channel with a not-null capacity, an interesting behaviour from a practical point of view \cite{SalEblChi-18,CalCac-20}.

            In the following we consider the \textit{bit-flip} channel $\mathcal{X} \mapsto \{\texttt{X}, \mathds{1}\}$ and the \textit{phase-flip} channel ${\mathcal{Z}} \mapsto \{\texttt{Z}, \mathds{1}\}$, having noise probability, respectively, \texttt{p} and \texttt{q}.
            Ultimately, by preparing the control state to be $\varkappa = \ket{\texttt{+}}\bra{\texttt{+}}$, it occurs the evolution represented in Figure \ref{fig:ico_comm} \cite{SalEblChi-18}.
            \begin{figure}[h]
                \centering
                \label{eq:qs_xz}
                \begin{quantikz}[thin lines, row sep={0.5cm,between origins}, column sep={0.4cm}]
                    \lstick{$_{\sigma}$} & \gate[2,style={fill=violet!10}]{_{\mathcal{S}}} & \qw \rstick[2]{$_{(\ \cdots\ ) \otimes \ket{\texttt{+}}\bra{\texttt{+}}\ +\  \texttt{pq} ( \texttt{Y} \sigma \texttt{Y} ) \otimes \ket{\texttt{-}}\bra{\texttt{-}}}$}\\
                    \meterD[style={fill=cyan!10}]{_{\ket{\texttt{+}}}}&  & \qw
                \end{quantikz}
                \caption{Superposition of causal orders for bit-flip and phase-flip.}
                \label{fig:ico_comm}
            \end{figure}
            
            According to the bottleneck inequality \cite{Wil-13}, given the composite operation $\mathcal{Z}\circ\mathcal{X}$ and let $\mathfrak{C}(\cdot)$ be the quantum capacity, the following upper-bound holds \cite{SalEblChi-18}:
            \begin{equation}
                \mathfrak{C}(\mathcal{Z}\circ\mathcal{X}) \leq 1 - \texttt{max}\{\mathfrak{h}_{\texttt{2}}(\texttt{p}), \mathfrak{h}_{\texttt{2}}(\texttt{q})\},
            \end{equation}
            where $\mathfrak{h}_{\texttt{2}}$ denotes the binary Shannon entropy. Also, the same inequality holds for $\mathcal{X}\circ\mathcal{Z}$. 
            
            Whenever both $\texttt{p}$ and $\texttt{q}$ are equal to $\sfrac{1}{2}$, we have that both the configurations are characterized by a null capacity, i.e., $\mathfrak{C}(\mathcal{Z} \circ \mathcal{X}) = \mathfrak{C}(\mathcal{X} \circ \mathcal{Z}) = 0$.
            
            Let us now superpose the two noises. Accordingly, with probability $\texttt{pq}$ the output of circuit \ref{fig:ico_comm} is given by the second addendum, namely, $\big( \texttt{Y} \sigma \texttt{Y} \big) \otimes \ket{\texttt{-}}\bra{\texttt{-}}$.
            As consequence the capacity of $\mathcal{S}$ is lower-bounded by $\mathfrak{C}(\mathcal{S}) \geq \sfrac{1}{4}$, as shown in \cite{CalCac-20}.
            Specifically, $\sigma$ occurs to pass through $\mathcal{Y}(\sigma) = \texttt{Y}\sigma\texttt{Y}$, coherently with control state being $\ket{\texttt{-}}\bra{\texttt{-}}$. Therefore, it is possible to exploit the control state to gain a heralded unitary evolution $\mathcal{Y}$ via post-selection through the occurrence of $\ket{\texttt{-}}\bra{\texttt{-}}$. Since $\mathcal{Y}$ is unitary, it is also reversible, therefore we can restore the information, gaining a perfect transmission of $\sigma$, i.e.,
            \begin{equation}
                \label{eq:ppy}
                \mathcal{Y}\circ\mathcal{Y}(\sigma) = \texttt{Y}\texttt{Y}\sigma\texttt{Y}\texttt{Y} = \sigma.
            \end{equation}

            \subsection{Quantum simulation}
            In this section we present our steps to realize $\mathcal{S}$ of Figure \ref{fig:ico_comm}.
            We already discussed that the state of the art on quantum technologies doesn't allow any native implementation of the indefinite causal orders.
            Hence our goal is to witness the communication advantage through quantum simulation.
            
            \begin{wrapfigure}{r}{6cm}
                \centering
                \begin{quantikz}[thin lines, row sep={0.6cm,between origins}, column sep={0.4cm}]
                    \lstick{$_{\ket{\vartheta}}$} &\gate[style={fill=violet!10}]{_{\{\texttt{E},\mathds{1}\}}} & \qw  \rstick{$\ _{\mapsto}$}
                \end{quantikz}
                \begin{quantikz}[thin lines, row sep={0.6cm,between origins}, column sep={0.4cm}]
    	            \lstick{$_{\ket{\vartheta}}$} &\gate[style={fill=violet!10}]{_{\texttt{E}}} & \qw \\
		            \meterD[style={fill=cyan!10}]{_{\ket{\texttt{p}}}}& \ctrl{-1} & \qw
                \end{quantikz}
                \caption{Stinespring dilation simulating a generic Pauli noise \texttt{E}.}
                \label{fig:dilation}
                \hrulefill
             \end{wrapfigure}
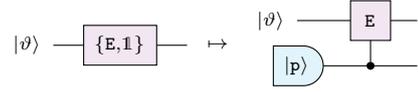
            When considering a non-unitaries, as in the case of our interest, the overhead grows up w.r.t. superposing unitaries. The reason is that non-unitaries are naturally harder to engineer and need to be simulated as well.
            Specifically, for any Hilbert space $\mathbb{H}$, the corresponding set of density states lies in the convex map $\mathcal{C}(\mathbb{H})$.
            
            According to the \textit{Stinespring dilation} \cite{Sti-55}, one can always associate to a non-unitary evolution $\mathcal{N}: \mathcal{C}(\mathbb{H}) \rightarrow \mathcal{C}(\mathbb{H})$ a unitary one, $\mathcal{A}_{\mathcal{N}}$, defined as follows:
            \begin{equation}
                \label{Eq:3.1}
                \mathcal{A}_{\mathcal{N}}: \mathcal{C}(\mathbb{H})\otimes\mathcal{C}(\mathbb{A}) \rightarrow \mathcal{C}(\mathbb{H})\otimes\mathcal{C}(\mathbb{A})
            \end{equation}
            where $\mathcal{C}(\mathbb{A})$ is an auxiliary system with associated basis $\big\{\ket{a_v}\bra{a_w}\big\}_{vw}$. Since $\mathcal{A}_{\mathcal{N}}$ is unitary, it has direct realization with the circuit algorithm.
            To simulate $\mathcal{N}$ from a realization of $\mathcal{A}_{\mathcal{N}}$, one need to \textit{discard} the auxiliary system. In terms of operations, discarding the auxiliary system means applying a \textit{partial trace} $\texttt{Tr}_2: \mathcal{C}(\mathbb{H})\otimes\mathcal{C}(\mathbb{A}) \rightarrow \mathcal{C}(\mathbb{H})$. Specifically, for a generic state $\sigma = \sum_{ijvw}c_{ijvw}\big(\ket{\vartheta_i}\bra{\vartheta_j}\otimes\ket{a_v}\bra{a_w}\big)$, the partial trace outputs the following \cite{RiePol-11}:
            \begin{equation}
                \label{Eq:3.2}
                \texttt{Tr}_2(\sigma) = \sum_{ijvw}c_{ijvw}\ket{\vartheta_i}\bra{\vartheta_j}\langle{a_w}|{a_v}\rangle.
            \end{equation}
            Since $\mathbb{H}$ and $\mathbb{A}$ are taken to be generic systems, equation \eqref{Eq:3.2} has a direct generalization to the form $\texttt{Tr}_{i_1,\dots,i_k}$, tracing out subsystems indexed by ${i_1},\dots,{i_k}$.
            
            In summary, we just outlined a method to realize an operation $\mathcal{N}$, involving two steps:
            \begin{enumerate}
                \item realizing the circuit $\mathcal{A}_{\mathcal{N}}$;
                \item discarding the auxiliary system with a partial trace $\texttt{Tr}_2$.
            \end{enumerate}

            To our purpose, we apply this method, restricted to evolution $\{\texttt{E}, \mathds{1}\}$. In circuit representation this is shown in Figure \ref{fig:dilation}.
            To superpose them we need an extra qubit, which encodes the control system. The final realization is shown in Figure \ref{circ:ico_ij}.
            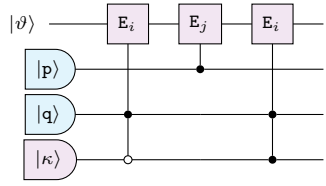
\begin{figure}[h]
                \centering
                \begin{quantikz}[thin lines, row sep={0.6cm,between origins}, column sep={0.4cm}]
                	\lstick{$_{\ket{\vartheta}}$} &\gate[style={fill=violet!10}]{_{\texttt{E}_i}} & \gate[style={fill=violet!10}]{_{\texttt{E}_j}} &\gate[style={fill=violet!10}]{_{\texttt{E}_i}} & \qw \\
                    \meterD[style={fill=cyan!10}]{_{\ket{\texttt{p}}}} & \qw & \ctrl{-1} & \qw & \qw\\
            		\meterD[style={fill=cyan!10}]{_{\ket{\texttt{q}}}}& \ctrl{-2} & \qw & \ctrl{-2} & \qw\\
            		\meterD[style={fill=violet!10}]{_{\ket{\kappa}}} & \octrl{-1} & \qw & \ctrl{-1} & \qw
                \end{quantikz}
                \caption{Circuit implementing an indefinite causal orders between two Pauli noise $\texttt{E}_i, \texttt{E}_j$.}
                \label{circ:ico_ij}
            \end{figure}

            \subsection{Physical setting}
            Starting from circuit \ref{circ:ico_ij}, we can do a step closer towards the real physical setting meant for us to witness the communication advantage of the indefinite causal order.
            To this aim we need to simulate the evolution from Figure \ref{fig:ico_comm}, for the case of i.i.d. probabilities -- i.e. $p = q = \sfrac{1}{2}$.
            Figure \ref{circ:ico_xz} shows the physical setting we used for our experiments.
            Notice that we added specific state preparations and measurements for information and control qubits. This update express a process tomography settings, explained in detail in Section \ref{sec:qpt}.
            Second and third qubits are not measured, rather, their final output is ignored, which naturally express trace out over those systems -- i.e. $\texttt{Tr}_{2,3}$.
            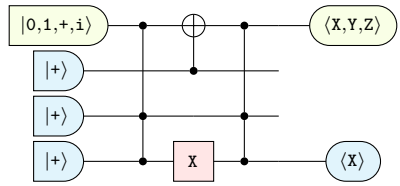
\begin{figure}[h]
                \centering
                \begin{quantikz}[thin lines,row sep={0.6cm,between origins}, column sep={0.4cm}]
                	\meterD[style={fill=lime!10}]{_{\ket{\texttt{0},\texttt{1},\texttt{+},\texttt{i}}}}&\control{} & \targ{} &\control{} & \qw & \measure[style={fill=lime!10}]{_{\langle {\texttt{X}}, {\texttt{Y}}, {\texttt{Z}}\rangle}} \\
            		\meterD[style={fill=cyan!10}]{_{\ket{\texttt{+}}}} & \qw & \ctrl{-1} & \qw & \qw\\
            		 \meterD[style={fill=cyan!10}]{_{\ket{\texttt{+}}}}& \ctrl{-2} & \qw & \ctrl{-2} & \qw\\
            		 \meterD[style={fill=cyan!10}]{_{\ket{\texttt{+}}}} & \ctrl{-1} & \gate[style={fill=red!10}]{_{\texttt{X}}} & \ctrl{-1} & \qw & \measure[style={fill=cyan!10}]{_{\langle {\texttt{X}}\rangle}}
                \end{quantikz}
                \caption{Circuit representation of the quantum experiment setting we used to witness the communication advantage.}
                \label{circ:ico_xz}
            \end{figure}
            
            At the end each run a post-selection occurs. Namely, coherently with our discussion of Section \ref{sec:soco}, we only keep those outputs where the control qubit results in the state $\ket{\texttt{-}}$. To the given set of outputs, we then apply a \textbf{classical bit-flip} in case the information qubit was subject to a measurement $\langle\texttt{X}\rangle$ or $\langle\texttt{Z}\rangle$. In fact, aware of the fact that the communication advantage comes from the post-processing of equation \eqref{eq:ppy}, which is a Pauli operation, this can be computed classically and has an effect only when measuring $\langle\texttt{X}\rangle$ and $\langle\texttt{Z}\rangle$.
            
            
            \subsubsection{Circuit decomposition and optimization}
            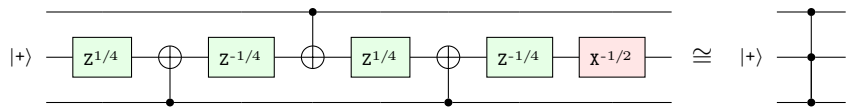
\begin{figure}[h]
                \centering
                \begin{quantikz}[thin lines, row sep={0.6cm,between origins}, column sep={0.35cm}]
    	            & \qw & \qw & \qw & \ctrl{+1}& \qw & \qw & \qw & \qw & \qw\\
                	\lstick{$_{\ket{\texttt{+}}}$}& \gate[style={fill=green!10}]{_{\texttt{Z}^{\sfrac{1}{4}}}} & \targ{} & \gate[style={fill=green!10}]{_{\texttt{Z}^{\texttt{-}\sfrac{1}{4}}}} & \targ{} & \gate[style={fill=green!10}]{_{\texttt{Z}^{\sfrac{1}{4}}}} & \targ{} & \gate[style={fill=green!10}]{_{\texttt{Z}^{\texttt{-}\sfrac{1}{4}}}} & \gate[style={fill=red!10}]{_{\texttt{X}^{\texttt{-}\sfrac{1}{2}}}} & \qw \rstick{$\ {\cong}$}\\
                	& \qw & \ctrl{-1} & \qw & \qw & \qw & \ctrl{-1} & \qw & \qw & \qw
                \end{quantikz}
                \begin{quantikz}[thin lines, row sep={0.6cm,between origins}, column sep={0.4cm}]
                \\
                &\ctrl{2}&\qw\\
                \lstick{$_{\ket{\texttt{+}}}$}&\ctrl{1}&\qw\\
                &\control{}&\qw\\
                \end{quantikz}
                \caption{Whenever the second wire takes $\ket{\texttt{+}}$ as input, the gate r.h.s. is equivalent, up to a global phase, to the decomposition l.h.s.}
                \label{fig:ccz_+}
            \end{figure}
            
            As already discussed in Section \ref{sec:soco_comp}, it is often the case that a real quantum technology doesn't supply natively a gate. For our setting of Figure \ref{circ:ico_xz}, this is the case of the 3-qubits gates, which have expensive decomposition \cite{shende2009cnot}.
            
            \begin{wrapfigure}{r}{6.2cm}
                \centering
                \includegraphics[scale=0.4]{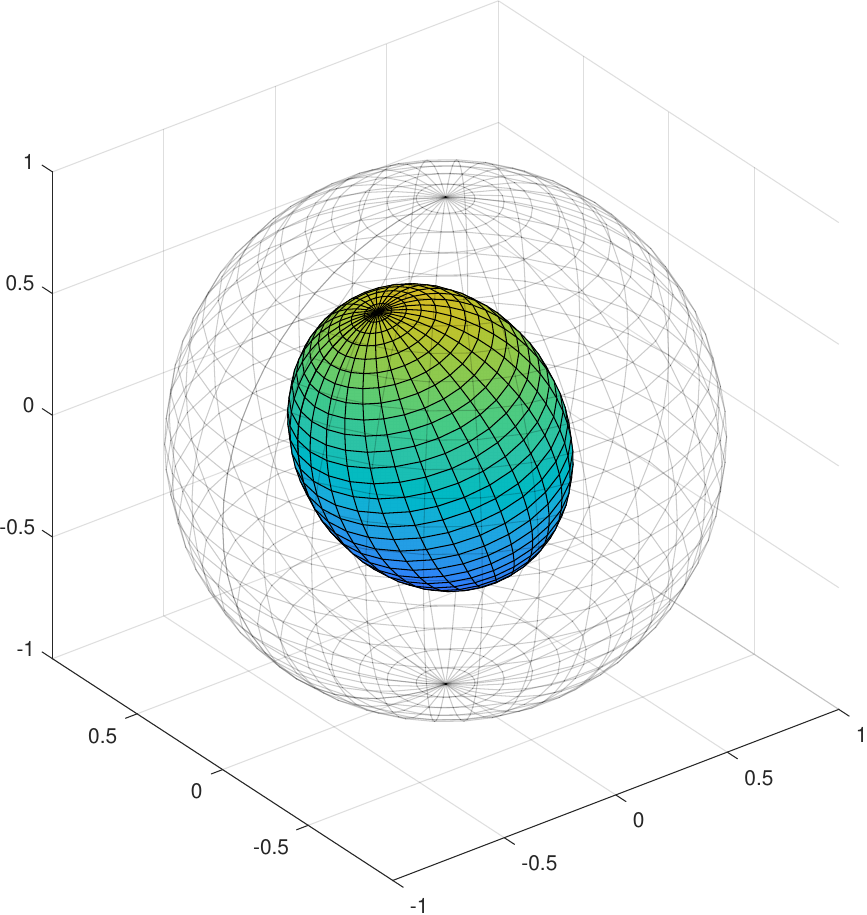}
                \caption{Bloch sphere representation of the experimental characterization for the physical setting of Figure \ref{circ:ico_xz}.}
                \label{fig:bloch_ico}
                \hrulefill
            \end{wrapfigure}
            To witness the communication advantage we managed to do some optimization on the first occurrence of such a gate.
            The rationale behind the optimization is that the gate decomposition could be simpler in case we have some knowledge of the input. We indeed know there is at least one qubit prepared in the state $\ket{\texttt{+}}$. This is enough to use the decomposition in Figure \ref{fig:ccz_+} instead of the standard one.

            The final result is shown by plotting the characterization of the channel as a bloch sphere -- See Figure \ref{fig:bloch_ico}. It represents the experimental characterization for the physical setting of Figure \ref{circ:ico_xz}. Gray sphere represents the ideal sphere, corresponding to a set of pure states. The inside coloured sphere is the deformation induced by the imperfections caused by the employed technology, which in this case is the \texttt{santiago} processor provided by \texttt{IBM}\footnote{The processor has been retired at the time of writing.}, which has a \textit{quantum volume} \cite{cross2019validating} of 16.

    \section{Modeling faulty gates}
     \begin{wrapfigure}{r}{6.2cm}
            \centering
            \begin{quantikz}[thin lines, row sep={0.6cm,between origins}, column sep={0.25cm}]
    	            &\qw &\gate[style={fill=violet!10},nwires={3},4]{_{\mathcal{N}}}&\qw&\gate[style={fill=lime!10},nwires={3},4]{_{\mathcal{U}}}&\qw&\qw\\
    	            &\qw&&\qw&&\qw&\qw\\
    	            &\overset{{\vdots}}{{\phantom{c}}}&&&&\overset{{\vdots}}{{\phantom{c}}}&\\
    	            &\qw&&\qw&&\qw&\qw
            \end{quantikz}
            \caption{Faulty operation $\mathcal{U}$ expressed as the composition $\mathcal{U}\circ\mathcal{N}$.}
            \label{fig:pauli-lind}
            \hrulefill
        \end{wrapfigure}
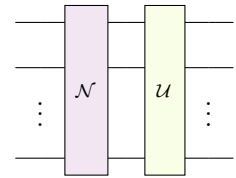
        According to the Pauli-Lindblad master equation \cite{breuer2002theory, berg2022probabilistic,schlosshauer2007decoherence,cacciapuoti2019toward}, a faulty operation can be modeled as $\mathcal{U}\circ\mathcal{N}$. In words, the model allows to think of a faulty operator as the composition of an ideal operator $\mathcal{U}$ preceeded by some Pauli-noise $\mathcal{N}$ -- see Fig. \ref{fig:pauli-lind}.
        
        For numerical evaluations, a common assumption is that a quantum evolution undergoes a depolarization \cite{nielsen2003quantum} $\mathcal{D}$, which can be expressed as a mixture of Pauli errors:
         \begin{equation}
            \label{eq:dep}
                \mathcal{D}(\sigma) = \sqrt{1 - \frac{3\lambda}{4}}\mathds{1}\sigma\mathds{1} + \sqrt{\frac{\lambda}{4}}(\texttt{X}\sigma\texttt{X} + \texttt{Z}\sigma\texttt{Z} + \texttt{Y}\sigma\texttt{Y}).
            \end{equation}

        Starting from equation \eqref{eq:dep}, it is possible to relate $\lambda$ to a triplet of probabilities for the Pauli errors \texttt{X}, \texttt{Y} and \texttt{Z}. A method to do that is outlined in Ref. \cite{dep-prob}.
        Furthermore, according to Ref. \cite{dennis2002topological}, the approximation
            \begin{equation}
            \label{eq:dep2}
                \mathcal{D}(\sigma) \approx (1-\texttt{p})^2\sigma + \texttt{p}(1-\texttt{p})(\texttt{X}\sigma\texttt{X} + \texttt{Z}\sigma\texttt{Z}) + \texttt{p}^2\texttt{XZ}\sigma\texttt{XZ}.
            \end{equation}

        \noindent{}is a good model for decoherence and it applies independently to each qubit of the system -- i.e., $\mathcal{D}^{\otimes d}$, for a $d$-dimensional system.
        
        The joint result of equation \eqref{eq:dep2} together with the Pauli-Lindblad assumption, allows to invastigate circuit as a composition of ideal operators affected by some single qubit Pauli error that affect the logic of computation. Such errors propagate among the circuit coherently with the logical operators. As basic example, consider the operator $\land(\texttt{X})$. According to equation \eqref{eq:dep2}, a Pauli error may precede its execution and propagates through the system as in Figure \ref{circ:prop}.
        
        Despite the simplicity of the example, the two propagation rules shown in Figure \ref{circ:prop} are complete, as $\land(\texttt{X})$ is the only non-single qubit necessary for universal computation\footnote{We discussed this in detail in Section \ref{sec:programming}}.
        
        \begin{figure}[h]
            \centering
             \begin{subfigure}[b]{0.48\textwidth}
                 \centering
                 \begin{quantikz}[thin lines, row sep={0.35cm,between origins}, column sep={0.35cm}]
    	            &\qw&\ctrl{2}&\qw & & & \ctrl{2}&\gate[style={fill=green!10}]{_{\texttt{Z}}}&\qw\\
    	            & & & &_{\rightarrow} & & &\\
    	            &\gate[style={fill=green!10}]{_{\texttt{Z}}}&\targ{}&\qw & & & \targ{}&\gate[style={fill=green!10}]{_{\texttt{Z}}}&\qw\\
                \end{quantikz}
             \end{subfigure}
             \hfill
             \begin{subfigure}[b]{0.48\textwidth}
                 \centering
                 \begin{quantikz}[thin lines, row sep={0.35cm,between origins}, column sep={0.35cm}]
    	            &\gate[style={fill=red!10}]{_{\texttt{X}}}&\ctrl{2}&\qw & & & \ctrl{2}&\gate[style={fill=red!10}]{_{\texttt{X}}}&\qw\\
    	            & & & &_{\rightarrow} & & &\\
    	            &\qw&\targ{}&\qw & & & \targ{}&\gate[style={fill=red!10}]{_{\texttt{X}}}&\qw\\
                \end{quantikz}
             \end{subfigure}
             \caption{Propagation over $\land(\texttt{X})$ operator.}
             \label{circ:prop}
        \end{figure}
        As regards single-qubit operators, consider operators $\texttt{Z}^{\sfrac{1}{2}}, \texttt{X}^{\sfrac{1}{2}}$, necessary to generate the Clifford group as section \ref{sec:cliff}. In such a case, orthogonal Pauli errors invert their logic, as in circuits of Figure \ref{circ:inv_logic}
        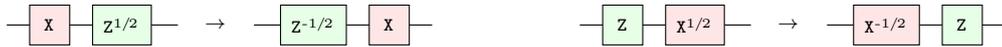
\begin{figure}[h]
            \centering
             \begin{subfigure}[b]{0.48\textwidth}
                \centering
                \begin{quantikz}[thin lines, row sep={0.35cm,between origins}, column sep={0.35cm}]
    	            &\gate[style={fill=red!10}]{_{\texttt{X}}}&\gate[style={fill=green!10}]{_{\texttt{Z}^{\sfrac{1}{2}}}}&\qw
                \end{quantikz}
                $\ \ _{\rightarrow}$
                \begin{quantikz}[thin lines, row sep={0.35cm,between origins}, column sep={0.35cm}]
    	            &\gate[style={fill=green!10}]{_{\texttt{Z}^{\texttt{-}\sfrac{1}{2}}}}&\gate[style={fill=red!10}]{_{\texttt{X}}}&\qw
                \end{quantikz}
             \end{subfigure}
             \hfill
             \begin{subfigure}[b]{0.48\textwidth}
                 \centering
                 \begin{quantikz}[thin lines, row sep={0.35cm,between origins}, column sep={0.35cm}]
    	            &\gate[style={fill=green!10}]{_{\texttt{Z}}}&\gate[style={fill=red!10}]{_{\texttt{X}^{\sfrac{1}{2}}}}&\qw
                \end{quantikz}
                $\ \ _{\rightarrow}$
                \begin{quantikz}[thin lines, row sep={0.35cm,between origins}, column sep={0.35cm}]
    	            &\gate[style={fill=red!10}]{_{\texttt{X}^{\texttt{-}\sfrac{1}{2}}}}&\gate[style={fill=green!10}]{_{\texttt{Z}}}&\qw
                \end{quantikz}
             \end{subfigure}
             \caption{Pauli noise affecting orthogonal single-qubit operators.}
             \label{circ:inv_logic}
        \end{figure}

        \section{Error correction and logical computing}
        \subsection{Code functions}
            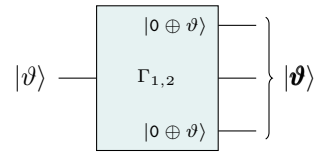
\begin{wrapfigure}{r}{5cm}
                \centering
                 \begin{quantikz}[thin lines,row sep={0.7cm,between origins}]
                    &\gate[3,nwires={1,3},style={fill=teal!10}][1.6cm]{_{\Gamma_{1,2}}}\gateoutput{${\ket{\texttt{0}\oplus\vartheta}}$}&\qw\rstick[wires=3]{$\ket{\pmb{\vartheta}}$}\\
                    \lstick[]{$\ket{\vartheta}$}& &\qw\\
                    &\gateoutput{${\ket{\texttt{0}\oplus\vartheta}}$}&\qw
                 \end{quantikz}
                 \caption{Example of code function $\Gamma_{1,2}: \mathbb{H} \rightarrow \mathbb{H}^{\otimes 3}$.}
                 \label{fig:encoding}
                 \hrulefill
            \end{wrapfigure}
        Similarly to what is done in classical information \cite{baylis2018error}, quantum information can be protected through the introduction of an \textit{error correction} scheme \cite{mancini2020quantum}.
        This starts with the definition of a \textit{code} function $\Gamma$, which introduce redundancy to the system and exploit it to restore the original information, in case an error occurs. Formally, a code function $\Gamma_{k,m}: \mathbb{H}^{\otimes k} \rightarrow \mathbb{H}^{\otimes(k+m)}$, able to encode $k$ qubits into $k+m$ with $m > 0$. The code is said to have ratio $\frac{k}{k+m}$. When defining a code, a \textit{desideratum} is to achieve a high ratio. Another likewise important metric is the \textit{distance} of a code. Specifically, $\Gamma_{k,m}$ creates a \textit{codeword} space s.t. $|\texttt{Im}(\Gamma_{k,m})| = 2^{k+m}$ where only $2^k$ elements are valid codewords. Generally speaking a highly sparse valid space allows for higher distance among valid codewords. On contrary a dense valid space makes fuzzier the identification of a couple of valid-invalid codewords. This issue is treated more formally in Section \ref{sec:dist}, while now it is just important to observe that distance and ratio are, by their nature, inversely proportional, creating a challenging trade-off to handle.

\subsubsection{Redundancy through entanglement}
As an example, consider a code function which takes in input a generic single qubit $\ket{\vartheta} = \alpha\ket{\texttt{0}} + \beta\ket{\texttt{1}}$ and generates a logical qubit
\begin{equation}
\label{eq:logical_q}
\ket{\pmb{\vartheta}} \equiv \alpha\ket{\texttt{{0}}} + \beta\ket{\texttt{{1}}} \equiv \alpha\ket{\texttt{000}} + \beta\ket{\texttt{111}}.
\end{equation}
Such code can be implemented as in Fig. \ref{fig:encoding}.
  
Coherently with the definition of code, the system has $3$ qubits with $2^k = 2$ possible outcomes: $\ket{\texttt{000}}$ and $\ket{\texttt{111}}$.
After the application of $\Gamma_{1,2}$, some errors become detectable. For example, assume that the bit-flip noise $\mathcal{X} \mapsto \{\mathds{1},\texttt{X}\}$ affects each qubit, with i.i.d. error probability $\texttt{p}$. One can perform a projection over the even space and the odd space for qubits $1$ and $2$, and then the same for qubits $2$ and $3$. This can be done by performing the \textit{non-destructive measurement}\footnote{A possible implementation of non-destructive measurement is is given in Sec. \ref{sec:classic-codes}.}
\begin{equation}
    \texttt{Z}^{\otimes 2} = (\ket{\texttt{00}}\bra{\texttt{00}} + \ket{\texttt{11}}\bra{\texttt{11}}) - (\ket{\texttt{01}}\bra{\texttt{01}} + \ket{\texttt{10}}\bra{\texttt{10}}).
\end{equation}
The eigenvalues are $\pm 1$ and the detectable errors are represented in the table below.
\begin{center}
\begin{tabular}{cc}
\hline
\rowcolor[HTML]{FFFFFF} 
Syndrome & Detection     \\ \hline
$+1,+1$ & $\texttt{I}\otimes\texttt{I}\otimes\texttt{I}$  \\
\rowcolor[HTML]{EFEFEF} 
$-1,+1$ & $\texttt{X}\otimes\texttt{I}\otimes\texttt{I}$ \\
\rowcolor[HTML]{FFFFFF} 
$-1,-1$ & $\texttt{I}\otimes\texttt{X}\otimes\texttt{I}$ \\
\rowcolor[HTML]{EFEFEF} 
$+1,-1$ & $\texttt{I}\otimes\texttt{I}\otimes\texttt{X}$
\end{tabular}%
\end{center}  
Notice that $\texttt{Z}^{\otimes 2}$ has no effect on $\alpha\ket{\texttt{000}} + \beta\ket{\texttt{111}}$, whatever the target qubits are. For this reason its measurement can be used to detect some error, without affecting the original state.

This error correction scheme works for $\texttt{p}<\frac{1}{2}$. Without the scheme the minimum fidelity is
\[\mathfrak{f} = \underset{\forall \ket{\vartheta}}{\texttt{min}}\sqrt{\bra{\vartheta}\mathcal{X}(\ket{\vartheta}\bra{\vartheta})\ket{\vartheta}} = \sqrt{1-\texttt{p}}.\] 
The probability of getting an error on the scheme is given by $\texttt{p}_e = 3\texttt{p}^2(1-\texttt{p}) + \texttt{p}^3$, so it is required that $\sqrt{1 - \texttt{p}_e} > \sqrt{1-\texttt{p}}$, which happens when $\texttt{p}<\frac{1}{2}$.

\section{Stabilizer codes}
Generally speaking, defining an efficient code is a hard task. Here, we review a fundamental family called \textit{stabilizer codes}, which is particularly helpful as it can relate to linear codes coming from classical error correction.

A code function $\Gamma_{k,m}$ is a stabilizer code if its image is characterized by an abelian subgroup\footnote{A group of which components commute one another. I.e. $[\texttt{S}_i,\texttt{S}_j] = 0$ holds for any $ \texttt{S}_i,\texttt{S}_j \in \mathbb{S}$.} $\mathbb{S}\subseteq\mathbb{E}$ as follows
\[\texttt{Im}(\Gamma_{k,m}) \equiv \{\ket{\pmb{\vartheta}}\ :\  \texttt{S}\ket{\pmb{\vartheta}} = \ket{\pmb{\vartheta}}, \forall \texttt{S} \in \mathbb{S}\}\]
$\mathbb{S}$ must be abelian in order to stabilize a non-trivial code. Consider the case $[\texttt{S}_1,\texttt{S}_2] = 1$, then
$\texttt{S}_1\texttt{S}_2\ket{\pmb{\vartheta}} = -\texttt{S}_2\texttt{S}_1\ket{\pmb{\vartheta}} = -\ket{\pmb{\vartheta}}$, but also $\texttt{S}_1\texttt{S}_2\ket{\pmb{\vartheta}} = \ket{\pmb{\vartheta}}$; hence $\ket{\pmb{\vartheta}} = - \ket{\pmb{\vartheta}} = 0$ and $\mathbb{S}$ stabilizes the trivial code $\texttt{Im}(\Gamma_{k,m}) = \{0\}$. 

For $2^k$ possible outcomes, a stabilizer group $\mathbb{S}$ has $2^{m}$ elements and, since it is abelian, it can be specified by $m$ generators $\{\texttt{S}_1,\texttt{S}_2,\dots,\texttt{S}_{m}\}$. The benefit of using generators is that to check whether a state vector is stabilized by $\mathbb{S}$ or not, one needs only to check it for the generators. 

To see how the correction strategy works, consider an error operator $\texttt{E}\in \mathbb{E}$. Let us analyse how $\texttt{E}$ relates with a generator $\texttt{S}_i$.
\begin{enumerate}
    \item $\exists \texttt{S}_i : \texttt{E}\texttt{S}_i =  - \texttt{S}_i\texttt{E}$, then $\texttt{S}_i\texttt{E}\ket{\pmb{\vartheta}} = -\texttt{E}\texttt{S}_i\ket{\pmb{\vartheta}} = -\texttt{E}\ket{\pmb{\vartheta}}$. Therefore, $\texttt{E}\ket{\pmb{\vartheta}}$ is a $-1$ eigenvector of $\texttt{S}_i$ and \texttt{E} can be detected by measuring $\texttt{S}_i$.
    \item Otherwise $\texttt{E}\texttt{S}_i = \texttt{S}_i\texttt{E}\ \forall \texttt{S}_i$\footnote{Because any couple of $\mathbb{E}$ either commutes or anti-commutes.}, then if $\texttt{E} \in \mathbb{S}$, it clearly doesn't corrupt the state. So the problem arises when $\texttt{E} \notin \mathbb{S}$, making $\texttt{E}$ \textit{undetectable}.
\end{enumerate}
The set of undetectable errors is given by $C_{\mathbb{E}}(\mathbb{S}) \smallsetminus \mathbb{S}$, where $C$ is the \textit{centralizer} function. 
Nevertheless, a noise $\mathcal{N}$ with some undetectable operators, may still be \textit{correctable}.
Formally, a generic noise $\mathcal{N}(\sigma) = \sum_{k}p_k\texttt{E}_k\sigma\texttt{E}_k$ is correctable if any two operators $\texttt{E}_i, \texttt{E}_j \in \{\texttt{E}_k\}_k$, differ in syndrome or have the same syndrome but differ by a
stabilizer, i.e. 
\[\texttt{E}_i\texttt{E}_j \in \mathbb{S} \cup \big(\mathbb{E} \smallsetminus C_{\mathbb{E}}(\mathbb{S})\big).\]

\section{Relation with classical binary codes}
\label{sec:classic-codes}
    Instead of expressing the error detection as the measurement operator $\texttt{Z}^{\otimes 2}$, we can use the math coming from classical linear codes to define a quantum error correction scheme, e.g.
    \[H = \begin{pmatrix}
        1 & 1 & 0\\
        0 & 1 & 1
    \end{pmatrix}.\]
    $H$ can be used as parity-check matrix to detect and correct a bit flip occurring in one of the three physical qubits of $\ket{\pmb{\vartheta}}$ -- see Fig. \ref{fig:h-circuit}. Clearly, this is true as long as ancillary qubits undergoes a negligible noise .
 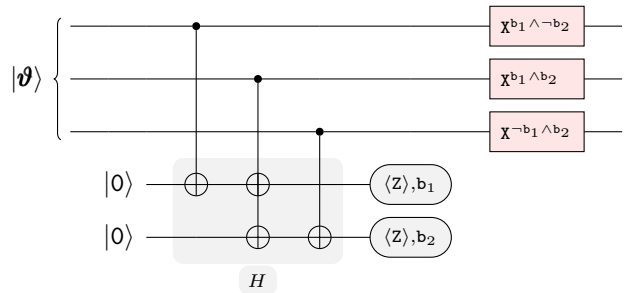
\begin{figure}[ht]
    \centering
     \begin{quantikz}[thin lines,row sep={0.7cm,between origins}]
            \lstick[wires=3]{$\ket{\pmb{\vartheta}}$}&\qw&\qw& \ctrl{3} &\qw & \qw & \qw & \gate[style={fill=red!10}]{_{\texttt{X}^{\texttt{b}_1 \land \lnot{\texttt{b}}_2}}}& \qw\\
            & \qw& \qw& \qw & \ctrl{3} & \qw & \qw & \gate[style={fill=red!10}]{_{\texttt{X}^{\texttt{b}_1 \land \texttt{b}_2\phantom{\lnot}}}} & \qw \\
            & \qw& \qw & \qw & \qw & \ctrl{2} & \qw & \gate[style={fill=red!10}]{_{\texttt{X}^{\lnot{\texttt{b}}_1\land \texttt{b}_2}}}& \qw\\
            &&\lstick[]{$\ket{\texttt{0}}$}& \targ{}\gategroup[wires=2,steps=3,style={draw =none,rounded corners,inner xsep=1pt,inner ysep=-1pt, fill=gray!10}, background,label style={rounded corners,label position=below, yshift=-0.37cm, fill=gray!10}, background]{$_{H}$} & \targ{} & \qw & \measure[style={fill=gray!10}]{_{\langle {\texttt{Z}\rangle},\texttt{b}_1}}\\
            &&\lstick[]{$\ket{\texttt{0}}$}& \qw & \targ{} & \targ{} & \measure[style={fill=gray!10}]{_{\langle {\texttt{Z}\rangle},\texttt{b}_2}}
     \end{quantikz}
     \caption{The gray coloring helps to visualize how the classical matrix representation $H$ of the detection scheme is implemented by means of auxilliary qubits.}
    \label{fig:h-circuit}
\end{figure}

Ancillary qubits allows to abstract from the hardware, while the real implementation of this scheme very depends on what kind of measurements the quantum processor supplies\footnote{See for example \cite{cramer2016repeated} for an experimental implementation, based on ancillary qubits to get non-destructive measurements.}.
Since $\texttt{E} \in \mathbb{E}$, one can represent a generic error as a $2(k+m)$-dimensional binary vector $(e_\texttt{x}|e_\texttt{z})$, such that if $e_{i,\texttt{x}} = 1$, \texttt{E} has an $\texttt{X}$ operator affecting the $i$-th qubit, or the identity otherwise. Symmetrically, for $e_{i,\texttt{z}}$ the subject operator is $\texttt{Z}$.

With the same criteria, let $(s_{i,\texttt{x}}|s_{i,\texttt{z}})$ be the binary vector representing a generator $\texttt{S}_i$. Thus one can construct an $m\times2(k+m)$ matrix $H$ such that
\begin{equation}
\label{eq:s}
H \equiv 
\begin{pmatrix}
s_{1,\texttt{x}}|s_{1,\texttt{z}}\\
\vdots\\
s_{{m},\texttt{x}}|s_{{m},\texttt{z}}
\end{pmatrix}
\end{equation}

To check if the noise $\mathcal{N}$ is fully correctable, for any couple $(e_{i,\texttt{x}}|e_{i,\texttt{z}}),(e_{j,\texttt{x}}|e_{j,\texttt{z}})$, related to the operators $\texttt{E}_i, \texttt{E}_j$, the following holds:
\[H(e_{i,\texttt{x}}+e_{j,\texttt{x}}|e_{i,\texttt{z}}+e_{j,\texttt{z}})^{\intercal} \neq 0.\]

To get a stabilizer code starting from a classical linear code function\footnote{$\mathbb{F}_2^n$ is a binary $n$-dimensional Hamming space, i.e. a vector of $n$ binary values.} $\Sigma_{k,m} : \mathbb{F}_2^k \rightarrow \mathbb{F}_2^{(k+m)}$, with $0<m<k$ and parity-check matrix $H_{\Sigma}$. The corresponding matrix for the quantum paradigm is given by
\begin{equation}
\label{eq:h}
    H_{\Gamma} \equiv 
    \begin{pmatrix}
      H_{\Sigma} & \textbf{0}\\
      \textbf{0} & H_{\Sigma}
    \end{pmatrix}
\end{equation}
If the code is self-orthogonal -- i.e. $\texttt{Im}(\Sigma_{k,m})^{\perp} \subseteq \texttt{Im}(\Sigma_{k,m})$ --, $H_{\Gamma}$ relates to a stabilizer group $\mathbb{S}$ -- in accordance with matrix \eqref{eq:s} -- for $k-m$ logical qubits and $k+m$ physical qubits\footnote{Notice the loss in the ratio -- i.e. $m/(k+m)$ -- caused by ``quantizing'' a classical code.}.
Formally, following this procedure leads to a quantum code $\Gamma_{k-m,m}$ such that
\[\texttt{Im}(\Gamma_{k-m,m}) \equiv \{\ket{\pmb{\vartheta}}\ :\  \texttt{S}\ket{\pmb{\vartheta}} = \ket{\pmb{\vartheta}}, \forall \texttt{S} \in \mathbb{S}\}.\]

A slightly more general definition considers a couple of classical linear codes $\Sigma_{k_1,m_1}, \Sigma_{k_2,m_2}$. As long as $\texttt{Im}(\Sigma_{k_2,m_2})^{\perp} \subseteq \texttt{Im}(\Sigma_{k_1,m_1})$ and $k_1 + m_1 = k_2 + m_2$ hold, we can perform the parity check, in accordance with matrix \eqref{eq:h}. The result is a quantum code $\Gamma_{k,m}$ with $k= k_1 + k_2 - m_1 - m_2$ and $m = m_1+m_2$. $\Gamma_{k,m}$ is called CSS code because of their creators Calderbank, Shor and Steane \cite{calderbank1996good, steane1996multiple}.

\section{Distance and bounds}
\label{sec:dist}
\subsection{Classical bounds}
Consider a codeword set s.t. $|\texttt{Im}(\Sigma_{k,m})| = 2^k$, defined in a $2^{k+m}$-dimensional Hamming space. An error relates to a codeword $u$, creating a new word $\tilde{u} = u + e$. Generally speaking $\tilde{u}$ can relate to more than one $u$ or $e$, making the definition of a good code a hard task.
A good code should be able to relate any error $e$ to one and one only codeword $u$. In this sense, such a code relates to each $u$ a lattice sphere, centered in $u$ with radius $r$. Each word $\tilde{u}$ in the sphere is such that $d(u,\tilde{u}) \leq r$ and $d(v,\tilde{u}) > r$ for any other codeword $v$.
In order to avoid any overlap, the radius is upper-bounded by $r\leq \floor{d/2}$, where $d$ is the minimum distance of the code. 
From this the \textit{Hamming bound} follows:
\[\sum_{i=0}^{\floor{d/2}}\binom{k+m}{i} \leq 2^{m}\]
Where $2^m$ is the maximum number of words orthogonal to the code and $\binom{k+m}{i}$ is the number of errors involving $i$ bits.
Whenever the spheres saturate this bound, without creating any overlap, the code is said to be \textit{perfect}\footnote{Despite its name, the code is still unable to apply correction whenever $u_1 + e_1 = u_2 + e_2$ occurs, with $u_1,u_2$ being codewords and $e_1,e_2$ being errors.}.

The minimum possible distance to not run into overlaps is $3$. To see that consider the basic case of two codewords $\texttt{01}$, and $\texttt{11}$. A parity-check with even result on the first codeword is indistinguishable from an odd result on the second codeword. 
It is possible to define a group of perfect codes\footnote{Known as Hamming codes.} of distance $d=3$. Namely, by setting $k = 2^h - h - 1$ and $m = h$ it results
\[\sum_{i=0}^1\binom{2^h - 1}{i} = 2^h\]
Notice how the ratio -- i.e. $1 - \frac{h}{2^h-1}$ -- rapidly grows to $1$ as $h$ grows.
More in general, for any high-dimensional code $\Sigma_{k,m}$ and minimum distance $d \propto k+m$, it is possible to prove \cite{mancini2020quantum} that the ratio $1 - \mathfrak{h}_{\texttt{2}}(d/(k+m))$ is achievable , where $\mathfrak{h}_{\texttt{2}}$ is the binary entropy.



\subsection{Quantum bounds}
For a given noise $\mathcal{N}(\sigma) = \sum_k p_k \texttt{E}_k\sigma\texttt{E}_k$, let $\{\texttt{E}_{k_i}\}_{i} \subseteq \{\texttt{E}_{k}\}_k$ be the set of undetectable errors.
Then, the code distance is given by
\[d = {\texttt{min}} \big\{\mathfrak{w}(\texttt{E})\ :\ \texttt{E} \in \{\texttt{E}_{k_i}\}_{i}\big\},\]
where $\mathfrak{w}$ is the weight function, counting the number of single-qubit components differing from the identity operator \texttt{I}.
If the code $\texttt{Im}(\Gamma_{k,m})$ is characterized by a stabilizer group $\mathbb{S}$, then 
\begin{equation}
\label{eq:d}
    d = {\texttt{min}} \big\{\mathfrak{w}(\texttt{E})\ :\ \texttt{E} \in C_{\mathbb{E}}(\mathbb{S}) \smallsetminus \mathbb{S}\big\}.
\end{equation}


Symmetrically to classical codes, the Hamming bound related to a code $\Gamma_{k,m}$ is given by
\[\sum_{i=0}^{\floor{d/2}} 3^i\binom{k+m}{i} \leq 2^{m}.\]

The new factor $3^i$ expresses the possible error combinations from $\mathbb{E}$ involving any $i$ qubits. Saturating the Hamming bound establishes a perfect code only if this is non-degenerate. 


Consider, as an example, the following stabilizer set \[\mathbb{S} = \langle\texttt{X}_{1,4}\texttt{Z}_{2,3}, \texttt{X}_{2,5}\texttt{Z}_{3,4}, \texttt{X}_{1,3}\texttt{Z}_{4,5}, \texttt{X}_{2,4}\texttt{Z}_{1,5}\rangle\]
$\mathbb{S}$ generates a code $\Gamma_{k,m}$ such that $m = \log|\mathbb{S}| = 4$ and, as the operators are defined over a $5$-dimensional system, $k=1$. Ultimately, according to \eqref{eq:d}, the distance is given by any \texttt{E} operating on at least 3 qubits, coming from the centralizer group $C_{\mathbb{E}}(\mathbb{S})$, e.g. $\texttt{E} = \texttt{Z}_{1,2,4}\texttt{X}_{4}$. To see that, consider any 1- or 2-qubits operator and check that it anti-commutes with some element from $\mathbb{S}$. Hence $d = \mathfrak{w}(\texttt{E}) = 3$. The sphere coverage is
\[\sum_{i=0}^{1} 3^i\binom{5}{i} = 16 = 2^m\]
and, therefore, the code is perfect.

\section{The role of stabilizers in computing}
There is a tight relation between communication and computing scenarios, as regards error correction. Namely, in communication, noise affects information conveyed through a physical channel. Similarly, in computing, noise affects information during the life-time of an algorithm.
Both scenarios run under the same noise model $\mathcal{N}(\sigma) = \sum_{k}p_k\texttt{E}_k\sigma\texttt{E}_k$. However, during computation, logical states demand for the definition of \textit{logical operators}.

A unitary operator $\texttt{\pmb{U}}$ is a logical operator for a code $\Gamma$ stabilized by $\mathbb{S}$ if, for any logical state $\ket{\pmb{\vartheta}}$ and any $\texttt{S} \in \mathbb{S}$, the following holds:
    \[\texttt{\pmb{U}}\ket{\pmb{\vartheta}} \in \texttt{Im}(\Gamma),\ \  \texttt{\pmb{U}S\pmb{U}}^\dagger\in \mathbb{S}.\]
To prove that, it is sufficient to show that
\[\texttt{\pmb{U}S\pmb{U}}^\dagger\texttt{\pmb{U}}\ket{\pmb{\vartheta}} = \texttt{\pmb{U}S}\ket{\pmb{\vartheta}} = \texttt{\pmb{U}}\ket{\pmb{\vartheta}}\]
holds for all $\texttt{S}$ in the generator of $\mathbb{S}$.
Notice that $\texttt{\pmb{U}S\pmb{U}}^\dagger$ stabilizes $\texttt{\pmb{U}}\ket{\pmb{\vartheta}}$. 

\subsubsection{Defining a stabilizer code from scratch}
From the above emerges a general way to build a stabilizer code by defining together a code function $\bar{\Gamma}_{k,m}$ and its stabilizer group $\bar{\mathbb{S}}$.
Formally consider the state $\ket{\pmb{\vartheta}} \equiv \ket{\vartheta} \otimes\ket{\texttt{0}}^{\otimes m}$, which is (trivially) stabilized by the group $\mathbb{S} \equiv \langle\texttt{Z}_{k+1}, \texttt{Z}_{k+2}, \dots, \texttt{Z}_{k+m}\rangle$.
Let $\texttt{U}$ be a unitary operator mapping $\mathbb{S}$ to itself\footnote{This may be any operator coming from the Clifford group, as it satisfies the closure over the Pauli group.}, then 
\[\texttt{Im}(\bar{\Gamma}_{k,m}) = \{\texttt{U}\ket{\pmb{\vartheta}}\ :\  \bar{\texttt{S}}\texttt{U}\ket{\pmb{\vartheta}} = \texttt{U}\ket{\pmb{\vartheta}}, \forall \bar{\texttt{{S}}} \in \bar{\mathbb{S}} \},\]
where $\bar{\mathbb{S}} \equiv \langle\bar{\texttt{S}}_{1}, \bar{\texttt{S}}_{2}, \dots, \bar{\texttt{S}}_{m}\rangle$ and $\bar{\texttt{S}}_{i} \equiv \texttt{U}\texttt{Z}_{k+i}\texttt{U}^{\dagger}$.

If the code is meant for computation, the only missing ingredient is the set of logical operators, which is 
\[\big\{\pmb{\texttt{{Z}}_i}, \pmb{\texttt{{X}}_i}\ :\ \pmb{\texttt{{Z}}_i} \equiv \texttt{U}\texttt{Z}_{i}\texttt{U}^{\dagger},\ \pmb{\texttt{{X}}_i} \equiv \texttt{U}\texttt{X}_{i}\texttt{U}^{\dagger}\big\}_{1\leq i\leq k}.\]

\subsubsection{Achieving fault-tolerant computing}
Let $\Gamma$ be any stabilizer code satisfying \textit{self-duality}\footnote{$\texttt{Im}(\Gamma) = \texttt{Im}(\Gamma)^{\perp}.$} and being \textit{doubly-even}\footnote{Any codeword has Hamming weight divisible by $4$.}.
Then the Clifford group generators $\land(\texttt{X}), \texttt{X}^{\sfrac{1}{2}}$ and $\texttt{Z}^{\sfrac{1}{2}}$ relate to \textit{fault-tolerant} logical operators $\pmb{\land(\texttt{{X}})}, \pmb{\texttt{X}^{\sfrac{1}{2}}}, \pmb{\texttt{Z}^{\sfrac{1}{2}}}$.

\begin{figure}[h]
    \centering
        \begin{quantikz}[thin lines, row sep={0.5cm,between origins}, column sep={0.6cm}]
                \lstick[3]{$\ket{\pmb{\vartheta}}$}\qw &\ctrl{3}&\qw&\qw&\qw\\
                \qw&\qw&\ctrl{3}&\qw&\qw\\
                \qw&\qw&\qw&\ctrl{3}&\qw\\
                \lstick[3]{$\ket{\pmb{\varphi}}$}\qw
                &\targ{}&\qw&\qw&\qw\\
                \qw &\qw&\targ{}&\qw&\qw\\
                \qw&\qw&\qw&\targ{}&\qw\\
        \end{quantikz}
        $\ \ \ \ \longrightarrow{}$
        \begin{quantikz}[row sep={0.5cm,between origins}, column sep={0.45cm}]
                \lstick[]{$\ket{\pmb{\vartheta}}$}\qw &\ctrl{1}&\qw \\
                \lstick[]{$\ket{\pmb{\varphi}}$}\qw &\targ{}&\qw 
        \end{quantikz}
    \caption{Transversal implementation of a logical $\pmb{\land(\texttt{{X}})}$.}
    \label{fig:transv_cx}
\end{figure}
These operators can be claimed to be fault-tolerant because they admit (in principle) a so-called \textit{transversal} implementation, which is very efficient in terms of circuit depth and error propagation. 
Figure \ref{fig:transv_cx} shows an example of transversal implementation of $\pmb{\land(\texttt{{X}})}$ between two logical qubits $\ket{\pmb{\vartheta}}$ and $\ket{\pmb{\varphi}}$. Its efficiency is given by the fact the each physical operator acts on independent pairs of physical qubits. For the same reason, also the noise does not mix up among the physical qubits.

\begin{wrapfigure}{r}{6.5cm}
    \centering
    \begin{quantikz}[thin lines,row sep={0.65cm,between origins}, column sep ={0.4cm}]
        \lstick[]{$_{\ket{\vartheta}}$}&\ctrl{1}&\qw&\gate[style={fill=green!10}]{_{\texttt{Z}^{\sfrac{\texttt{b}}{2}}}}&\qw \rstick[]{$_{\texttt{Z}^{\sfrac{1}{4}}\ket{\vartheta}}$}\\
        \meterD[style={fill=cyan!10}]{_{\ket{\omega}}}&\targ{}&\measure[style={fill=gray!10}]{_{\langle {\texttt{Z}\rangle},\texttt{b}}}&
    \end{quantikz}
    \caption{Example of $\texttt{Z}^{\sfrac{1}{4}}$ gate \textit{injection}.
    }
    \label{fig:inj}
    \hrulefill
\end{wrapfigure}
As regard fault-tolerance for the universal gate set $\mathbb{C}^{\texttt{+}}$ -- see Sec. \ref{sec:cliff} -- and especially for the non-Clifford operator $\texttt{Z}^{\sfrac{1}{4}}$; there are some proposal for transversal implementation for the logical operator $\pmb{\texttt{Z}^{\sfrac{1}{4}}}$, but these usually do not relate to any stabilizer group. Hence, in literature, two main branches of research emerged:
\begin{itemize}
    \item circuit manipulation with the goal of minimizing $\texttt{Z}^{\sfrac{1}{4}}$ occurrences \cite{Sel-13, AmyMasMos-14};
    \item design of $\pmb{\texttt{Z}^{\sfrac{1}{4}}}$ by means of \textit{injection} protocols \cite{yoganathan2019quantum, li2015magic,zhou2000methodology}.
\end{itemize}

A basic example of $\texttt{Z}^{\sfrac{1}{4}}$ \textit{injection} is shown in Figure \ref{fig:inj}; this performs the injection by introducing one auxiliary qubit to the processor, prepared in the state
\begin{equation}
    \label{eq:Tgate}
    \ket{\omega} = \frac{1}{\sqrt{2}}(\ket{\texttt{0}} + e^{\frac{\texttt{i}\pi}{4}}\ket{\texttt{1}}).
\end{equation}

Normal forms -- see Sec. \ref{sec:cliff} -- for universal circuits are also possible. An interesting result in this sense is available in Ref. \cite{van2021constructing}, where authors showed that non-Clifford operators can be pushed to the beginning of the circuit.


\section{Conclusion}
With this chapter we covered many important topics to know when dealing with quantum computation. Especially considering the current state-of-the-art of quantum technologies, which are commonly referred as Noisy Intermediate-Scale Quantum (NISQ) architectures \cite{Preskill2018}. With the incoming Ch. \ref{ch:compile}, we will focus on optimizing circuits by means of \textit{circuit compilation}, which preserves the circuit logic, while aiming to circuits more compliant to the hardware limitations. For practical reasons, we will consider circuit optimization without error correction schemes, as at the current stage of quantum technologies, these are to be considered at an early stage, where the gain promised by the theory struggle to be witnessed in real implementations. In fact, such schemes usually demands for more resources than the actual availability. 

In accordance with our full-stack development proposed in Ch. \ref{ch:techs}, we expect error correction scheme implementations to show up at the \textit{control system} level. Hence, a scheme will be in charge of providing a logical view of the physical resources. Because of such an organization, the compiler should not affect the logic on which the scheme relies on. This observation lead to think of new challenges specific to distributed architectures. 

Before proceeding with the investigation of distributed compilers, we conclude the chapter by making some observation on the complication arising when trying and embedding error correction schemes within a distributed system.

\subsection{Open challenges}
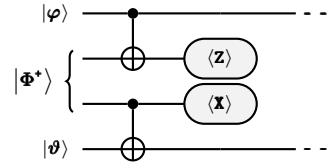
\begin{wrapfigure}{r}{5cm}
    \centering
    \begin{quantikz}[row sep={0.6cm,between origins}, column sep={0.5cm}]
        \lstick[]{$_{\ket{\pmb{\varphi}}}$}& \ctrl{1} &\qw & \qw \ _{\pmb{\text{-}\ \text{-}}} \\
        \lstick[2]{$_{\ket{\pmb{\Phi^{\texttt{+}}}}}$}& \targ{} & \measure[style={fill=gray!10}]{_{\ \langle {\pmb{\texttt{Z}}}\rangle\ }} \\
        & \ctrl{1} & \measure[style={fill=gray!10}]{_{\ \langle \pmb{\texttt{X}}\rangle\ }}\\
        \lstick[]{$_{\ket{\pmb{\vartheta}}}$}  & \targ{} & \qw & \qw \ _{\pmb{\text{-}\ \text{-}}}
    \end{quantikz}
    \caption{Transversal telegate, without post-processing.}
    \label{fig:lrcx}
    \hrulefill
\end{wrapfigure}
As already shown in Ch. \ref{ch:essentials}, telegates work by means of the generation and distribution of Bell states. Implementing a \textit{transversal logical telegate} would result into a high parallelism of each required task. However, it is not straightforward to \textit{import} the quantum error correction schemes into the distributed paradigm. To understand why, consider Fig. \ref{fig:lrcx}, it shows all the telegate steps but the last one -- i.e. the post-processing. The bold representation refer to a transversal implementation, according to the formalism we introduced -- see Fig. \ref{fig:transv_cx}. 
Unfortunately, such an implementation is not logical in general, as the Bell states may affect the system which can end up \textit{outside} the code scheme. For example, assuming:
\begin{equation}
    \label{eq:lbs}
    \ket{\pmb{\Phi^{\texttt{+}}}} \equiv \ket{\Phi^{\texttt{+}}}^{\otimes k + m}.
\end{equation}
In such a case the Bell states are independent one another and because of this, there is a time lapse during which the system lays outside the code scheme and, for this reason, it is vulnerable to undetectable errors. The time lapse starts when applying the transversal operators ${\land(\texttt{X})}$ and can only be restored by the post-processing, hoping that no error occurred in the meantime. In other words, the only detectable errors would be those caused by the post-processing, which, however, from a hardware perspective, results to be the most reliable step. Hence, this should not be considered as a practical way to proceed.

A possible way round is to generate and distribute a maximally entangled system, e.g., the generalized GHZ state:
\begin{equation}
    \label{eq:lghz}
    \ket{\pmb{\Phi^{\texttt{+}}}} \equiv \frac{1}{\sqrt{2}}\big(\ket{{\texttt{0}}}^{\otimes k+m} + \ket{{\texttt{1}}}^{\otimes k+m}\big).
\end{equation}
Since the distributed system is now fully entangled, such a system can be used not only to perform the non-local operator $\pmb{\land(\texttt{X})}$, but it results that the measurement outcome can be used to detect errors, combining so the syndrome with the post-processing.

The proposed distributed encoded system -- Eq. \eqref{eq:lghz} -- seems to be the solution to the problem. However, it has significant drawbacks from an hardware perspective. It should be already clear, this far, that generating a Bell state with high fidelity and within a reasonable time lapse is very challenging. This gets even harder, when thinking of higher degree systems, as the one of Eq. \eqref{eq:lghz}\footnote{E.g. the smallest transversal code has $k=1$ and $m=6$ \cite{eczoo_steane}, resulting in the generation and distribution of a maximally entangled system of $14$ qubits.}.

In conclusion, to model a fault-tolerant scheme for $\pmb{\land(\texttt{X})}$, Eq. \eqref{eq:lghz} may be an assumption too strong to be practical.
A more clever approach would be to define the generation and distribution protocol as to be \textbf{part of the encoding}. Very little has been done in this direction experimentally. An inspiring set-up can be found in Ref. \cite{zhang2022loss}, where authors managed to create a Shor code \cite{eczoo_shor_nine} by means of photons paired in Bell states. As being photon-based, this may bring to future experimental settings where \textit{stationary-flying hybrid systems} are considered.

    
        
\printbibliography[title=References,heading=subbibintoc]
\end{refsection}

%% file: compilation/compilation.tex
\begin{refsection}
\chapter{Circuit compilers on distributed architectures}
\label{ch:compile}
\thispagestyle{empty}
\newpage

As outlined in Ch. \ref{ch:techs}, a full-stack development of a distributed system for quantum computation requires to be carefully engineered. The proposed stack allows a circuit designer to focus on the logic of its algorithm, without necessary consider all the issues coming from the physical infrastructure that will take care of computing it.

In this chapter we consider the layer interfacing with the algorithm (written in circuit model). This takes care of \textit{optimizing the circuit}, adapting it to the constraints given by the underlying layers. Such a layer is commonly referred as \textit{compiler} and the corresponding optimization problem is called the \textit{compilation problem}. This is a generally tough task to solve, even on single processor, and for which an \texttt{NP}-hardness proof is available \cite{BotKisMar-18}. 
An ever growing literature arises with a variety of proposals for local computation \cite{madden2022best,hillmich2021exploiting,burgholzer2022limiting, MasFalMos-08,SirSanCol-18,WilBurZul-19,LiDinXie-19,ZulWil-19,ItoRayIma-19,ZhaZheZha-20,KarTezPet-20,MorParRes-21,MarMorRoc-21, BooDoBec-18,FerAmo-21} and for distributed computation \cite{BeaBriGra-13,ZomHouHou-18,DaeNavZom-20,FerCacAmo-21, NikMohSed-21, DadZomHos-21,DaeNavZom-21, SarZom-21, ZomDavGho-21, SunGupRam-21}.

Even if quantum processors are already available, distributed architectures are at an early stage and must be discussed from several perspectives. A key concept is that of \textit{telegates} as the fundamental inter-processor operations \cite{stephenson2020high, cuomo2020towards,VanDev-16}. We already discussed in Ch. \ref{ch:essentials} how telegates works, but we report below the main steps.

Each telegate can be decomposed into several tasks, that we group as follows: (i) the generation and distribution of entangled states among different processors, (ii) local operations and (iii) classical communications. These tasks makes the telegate an expensive resource, in terms of running time and/or fidelity. As a consequence, they have critical impact on the performance of the overall computation.
In contrast to such a limit, telegates offer remarkable opportunities of parallelization. 
In fact, much circuit manipulation is possible to keep computation independent from telegate's tasks. Therefore, we aim to model an optimization problem that embeds such opportunities.


Fig. \ref{fig:overview} provides the reader with a conceptual map concerning the main scientific contributions.

\vspace{1cm}
\begin{figure}[h]
    \centering
    \input{compilation/overview}
    \caption{Overview of the main contribution; blue blocks denote the main steps in the problem modeling, scanned by blue arrows. The violet blocks are the main ingredients to the entry blue blocks.}
    \label{fig:overview}
\end{figure}
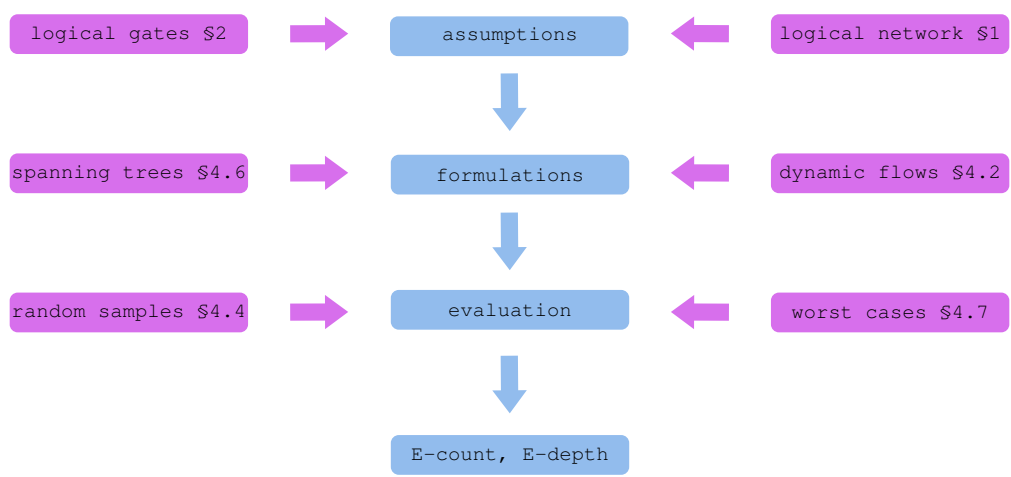


\newpage
\section{Mathematical modeling}
\label{sec:graph}
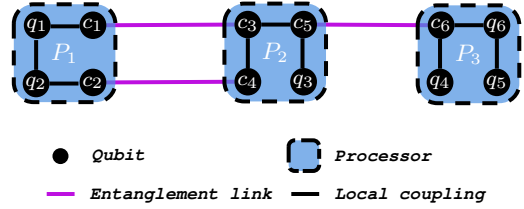
\begin{wrapfigure}{r}{7cm}
    \centering
    \input{compilation/Figures/3qpu.tex}
    \caption{Toy distributed quantum architecture with 3 processors.}
    \label{fig:3qpu}
	\hrulefill
\end{wrapfigure}
According to the current trend on quantum technologies -- reported in Ch. \ref{ch:techs} -- we can give a mathematical description. Formally, let $\mathcal{N} = (V,P,F)$ be a network triple representing the architecture.
$V = Q \cup C$ is a set of nodes describing qubits, therefore it is the disjoint union of computation qubits $Q = \{q_1, q_2,\dots, q_{|Q|}\}$ and communication qubits $C = \{c_1, c_2,\dots, c_{|C|}\}$. We can represent $n$ processors by partitioning $V$ into $P = \{P_1, P_2, \dots, P_n\}$. Therefore, a sub-set $P_i$ characterizes a processor as its set of qubits/nodes.

$F = L \cup R$ is as a set of undirected edges. $L$ represents the local couplings, therefore
\[L\subseteq \bigcup_i P_i\times P_i.\]
Notice that there is no particular assumption on connectivity nor cardinality within processors. This keeps the treating hardware-independent and it allows for heterogeneous architectures.

$R$ represents entanglement links. Since entanglement links connect only communication qubits, we introduce, for each processor, a set of those qubits only; i.e., $C_i = C\cap P_i$. Therefore, we have 
\[R \subseteq \bigcup_{i,j\ :\ i \neq j} C_i \times C_j.\]

Fig. \ref{fig:3qpu} shows an exemplary architecture, with three processors in $P$, six computation qubits in $Q$, six communication qubits in $C$, three entanglement links in $R$ and ten local couplings in $L$.

Concerning minimal assumptions, we only care about architectures actually able to perform any operation. This translated into a simple connection assumption.

\input{compilation/dqcc.tex}


\section{Clifford circuits compiler}
\label{sec:cliff-compiler}
In Sec. \ref{sec:cz-comp} we evaluated our model on $\mathcal{L}_{\land(\texttt{Z})}$ circuits. Even if they represent an important group of circuits, both theoretically \cite{maslov2018shorter} and practically \cite{bravyi2022constant,grzesiak2022efficient,bassler2022synthesis,maslov2018use,van2021constructing}  -- see Sec. \ref{sec:cz-comp}-- we now aim to extend our model to cover any Clifford circuit. Hence, since we managed to split the Clifford compilation in 3 independent problems -- see Sec. \ref{sec:norm-cliff} -- the final step is to define a compiler for $\mathcal{L}_{\land(\texttt{X})}$ circuits.

\subsection{Parity check circuits}
Any $\mathcal{L}_{\land(\texttt{X})}$ can be interpreted as a \textit{parity check} circuit, e.g. \cite{duncan2020graph}:
\[\ket{\texttt{b}_1,\texttt{b}_2,\texttt{b}_3,\texttt{b}_4}\ \mapsto\  \ket{\texttt{b}_1\oplus\texttt{b}_2,\texttt{b}_1\oplus\texttt{b}_3,\texttt{b}_4,\texttt{b}_3}.\]
We highlight such a relation to make more intuitive how we are now going to envision a generic $\mathcal{L}_{\land(\texttt{X})}$. It is common in parity check circuits to layering the circuit in operations of two kinds: \textit{fan-in} or \textit{fan-out} \cite{maslov2018use}. These two kinds are shown in circuit representation in Figs. \ref{fig:fan-in} and \ref{fig:fan-out}. Hence, one can see an $\mathcal{L}_{\land(\texttt{X})}$ circuit as a sequence of $\land(\texttt{X}^{\otimes m})$ operations, eventually interleaved by a layer of local gates $\mathcal{L}_{\texttt{Y}^{\sfrac{1}{2}}}$. $m$ denotes the number of target qubits and respects $m < n$, with $n$ being the total number of qubits.

We already showed in Sec. \ref{sec:rcxx} how to efficiently implement -- in terms of \texttt{E}-count -- an $\land(\texttt{X}\otimes\texttt{X})$ distributed over three processors. But we didn't detailed the general case $\land(\texttt{X}^{\otimes m})$; because it depends on the connectivity of the network.

Assuming to work with rectangular lattices $\mathcal{R}_{_{\blacktriangledown}}$ -- since we showed in Sec. \ref{sec:cz-comp} having the best performance -- it is convenient to generalize the definition of entanglement path -- given in Sec. \ref{sec:rcxx} -- to the concept of \textit{entanglement tree}.
\begin{figure}[h]
    \centering
    \begin{subfigure}{0.49\textwidth}
        \centering
        \begin{quantikz}[thin lines, row sep={0.65cm,between origins}, column sep={0.15cm}]
            &\ctrl{1}&\ctrl{2}&\qw& \cdots\phantom{c}  & \ctrl{2} & \qw\\
            &\targ{} & \qw &\qw & \cdots\phantom{c} & \qw & \qw \\
            & \qw & \targ{}&\qw& \cdots\phantom{c} &\qw\vqw{1} &\qw &\\
            & & & & &\overset{{\vdots}}{{\phantom{c}}} & \\
            &\qw &\qw &\qw& \cdots{\phantom{c}} & \targ{}\vqw{-1} & \qw
        \end{quantikz}
        $\ \equiv$
        \begin{quantikz}[thin lines, row sep={0.65cm,between origins}, column sep={0.15cm}]
            &\ctrl{2}&\qw\\
            &\targ{}&\qw\\
            &\targ{}\vqw{1}&\qw\\
            &\overset{{\vdots}}{{\phantom{c}}}&\\
            &\targ{}\vqw{-1}&\qw
        \end{quantikz}
        \caption{Fan-in}
        \label{fig:fan-in}
    \end{subfigure}
    \hfill
    \begin{subfigure}{0.49\textwidth}
        \centering
        \begin{quantikz}[thin lines, row sep={0.65cm,between origins}, column sep={0.15cm}]
            &\targ{}&\targ{}&\qw& \cdots\phantom{c}  & \targ{}\vqw{2} & \qw\\
            &\ctrl{-1} & \qw &\qw & \cdots\phantom{c} & \qw & \qw \\
            & \qw & \ctrl{-2}&\qw& \cdots\phantom{c} &\qw\vqw{1} &\qw &\\
            & & & & &\overset{{\vdots}}{{\phantom{c}}} & \\
            &\qw &\qw &\qw& \cdots{\phantom{c}} & \control{}\vqw{-1} & \qw
        \end{quantikz}
        $\ \equiv$
        \begin{quantikz}[thin lines, row sep={0.65cm,between origins}, column sep={0.15cm}]
            &\gate[style={fill = gray!10}]{_{\texttt{Y}^{\sfrac{1}{2}}}}&\ctrl{2}&\gate[style={fill = gray!10}]{_{\texttt{Y}^{\sfrac{1}{2}}}}&\qw\\
            &\gate[style={fill = gray!10}]{_{\texttt{Y}^{\sfrac{1}{2}}}}&\targ{}&\gate[style={fill = gray!10}]{_{\texttt{Y}^{\sfrac{1}{2}}}}&\qw\\
            &\gate[style={fill = gray!10}]{_{\texttt{Y}^{\sfrac{1}{2}}}}&\targ{}\vqw{1}&\gate[style={fill = gray!10}]{_{\texttt{Y}^{\sfrac{1}{2}}}}&\qw\\
            &&\overset{{\vdots}}{{\phantom{c}}}&\\
            &\gate[style={fill = gray!10}]{_{\texttt{Y}^{\sfrac{1}{2}}}}&\targ{}\vqw{-1}&\gate[style={fill = gray!10}]{_{\texttt{Y}^{\sfrac{1}{2}}}}&\qw
        \end{quantikz}
        \caption{Fan-out}
        \label{fig:fan-out}
    \end{subfigure}
    \caption{Generic fan-in and fan-out operations}
\end{figure}
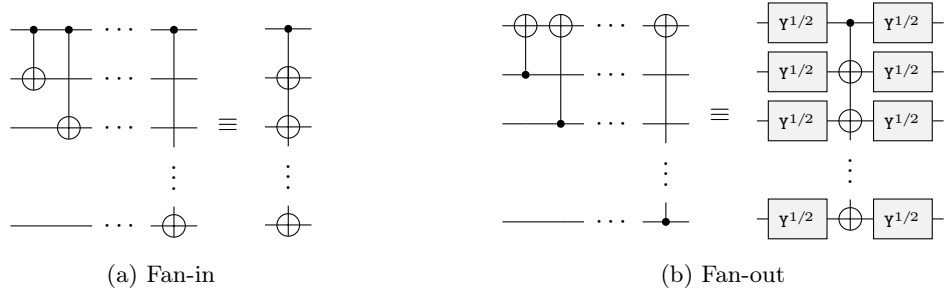

\subsection{Entanglement trees}
While it is always possible to find an entanglement path passing through all the qubits involved by a $\land(\texttt{X}^{\otimes m})$ operators, this may not be the best approach to run the subject operator. In fact, the shortest entanglement path does not necessarily finds the minimal \texttt{E}-count. In alternative, consider that a fan-in (fan-out) operator is rooted in the control (target) qubit and it doesn't actually constrains the targets (controls) to be covered by a single entanglement path. 
Hence, we can opt to compute an \textit{entanglement tree}, which may result in better solutions. In\ fact, the search for a tree covering $\land(\texttt{X}^{\otimes m})$ can be expressed as a generalization of the \textit{minimum spanning tree} problem \cite[Chapter~4.5]{kleinberg2006algorithm}. Such a generalization is known as \textit{minumum Steiner tree} problem \cite{cieslik2013steiner}, which is, in general, not tractable. Nevertheless efficient approximation ratio have been achieved \cite{kou1981fast,berman20091,chlebik2008steiner} and it can be used for any topology.
Since a lattice is a special case of the Euclidean plane, a further interesting result is that, for such a group of topologies, the problem admits a \textit{polynomial-time approximation scheme} \cite{byrka2010improved}. Finally, we report a result important to us within the following remark \cite{brazil1996minimal,brazil1997minimal}:
\begin{remark}
    The minimal Steiner tree on rectangular lattices $\mathcal{R}$ can be found in polynomial time.
\end{remark}

\begin{figure}[h]
    \centering
    \begin{subfigure}{0.48\textwidth}
        \centering
        \includegraphics[width=0.5\textwidth]{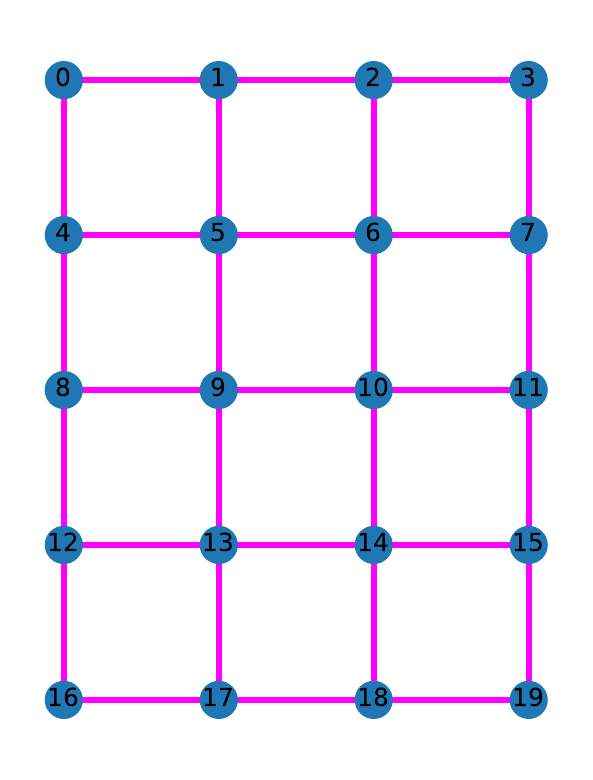}
        \caption{Rectangular lattice $\mathcal{R}_{_{\blacktriangledown}}$ with generator $g=6$.}
    \end{subfigure}
    \hfill
    \begin{subfigure}{0.48\textwidth}
        \centering
        \includegraphics[width=0.5\textwidth]{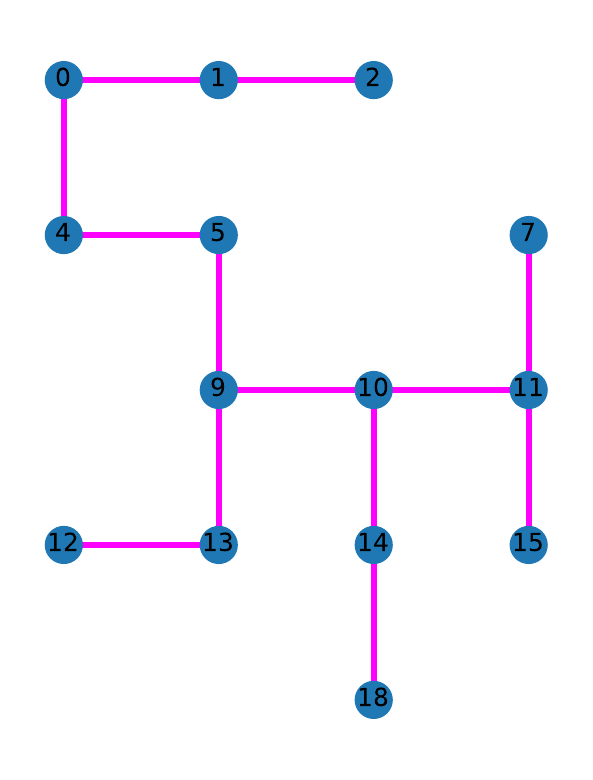}
        \caption{A Steiner tree with \texttt{E}-count of 13.}
    \end{subfigure}
    \caption{A Steiner tree computed from a network in the rectangular lattice topology. The tree is able to perform any non-local $\land(\texttt{X}^{\otimes m})$, spread over thirteen processors. }
    \label{fig:stainer}
\end{figure}
The above remark tells us that we can find, for any fan-in (fan-out) operation, the minimum \texttt{E}-count possible, by constructing an entanglement tree treated as a minimum Steiner tree problem. Fig. \ref{fig:stainer} shows an example computed with the method provided by the \texttt{networkx} library \cite{HagSchSwa-08}.

\subsection{Circuit construction and partitioning}
\begin{wrapfigure}{r}{6cm}
    \centering
    \begin{quantikz}[thin lines, row sep={0.65cm,between origins}, column sep={0.3cm}]
            \lstick{$_{q_1}$}&\ctrl{2}&\qw&\qw\ \cdots\ & \qw&\qw&\qw\\
            \lstick{$_{q_2}$}&\targ{}&\ctrl{1}&\qw\ \cdots\ & \qw&\qw&\qw\\
            \lstick{$_{q_3}$}&\targ{}\vqw{1}&\targ{}\vqw{1}&\qw\ \cdots\ & \qw&\qw&\qw\\
            &\overset{{\vdots}}{{\phantom{c}}}&\overset{{\vdots}}{{\phantom{c}}} &\\
            \lstick{$_{q_{n-2}}$}&\targ{}\vqw{-1}&\targ{}\vqw{-1}&\qw\ \cdots\ &\ctrl{2}&\qw&\qw\\
            \lstick{$_{q_{n-1}}$}&\targ{}\vqw{-1}&\targ{}\vqw{-1}&\qw\ \cdots\ &\targ{}&\ctrl{1}&\qw\\
            \lstick{$_{q_{n}}$}&\targ{}\vqw{-1}&\targ{}\vqw{-1}&\qw\ \cdots\ &\targ{}&\targ{}&\qw
        \end{quantikz}
    \caption{Hardest fan-in circuit composed by $\binom{n}{2}$ telegates, with $n$ being the number of qubits/processors.}
    \label{fig:hard-fan-in}
    \hrulefill
\end{wrapfigure}
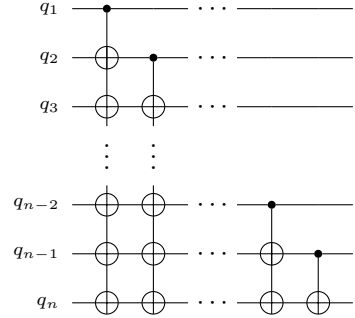
To prove the quality of our proposal, we start by considering \textit{dense fan-in (fan-out) circuits}, where with dense we mean that the layers of a circuit $\mathcal{L}_{\land(\texttt{X})} = \{\ell_1,\ell_2,\dots,\ell_d\}$ are such that each $\ell_{\dot{d}} \in \mathcal{L}_{\land(\texttt{X})}$ is a fan-in (fan-out) involving as many qubits as possible\footnote{Notice that including further operators would turn up to just create repeating pairs of $\land(\texttt{X})$, which cancel out each other.}$^,$\footnote{One may then generalize the definition to let more than one fan-in (or fan-out) within the same layer, as long as they operate on different qubits.}. The rationale behind this choice is that, when writing a parity circuit, it is convenient to aggregate $\land(\texttt{X})$ operators to form fan-in operators, i.e. $\land(\texttt{X}^{\otimes m})$. \noindent{}This does not restrict the kind of circuits the compiler will be able to process, since also the two-qubit operator $\land(\texttt{X})$ is a basic case of fan-in. By proceeding this way, we can push our compiler to solve the most complex parity check circuits, in terms both of number of operations and commutativity. Specifically, the hardest fan-in circuit has $\binom{n}{2}$ operations and $n-1$ non-commuting layers. For this reason, the \texttt{E}-depth corresponds to the number of layers, while we can optimize the \texttt{E}-count by computing the minimal Steiner tree for each layer. Circuit in Fig. \ref{fig:hard-fan-in} shows the criterion.

In terms of our formulation \eqref{dqcc} of \texttt{DQCC}, this means that we are fixing the order relation for any pair $\land(\texttt{X})_{s_i,t_i} \in \ell_{\dot{d}}$ and $\land(\texttt{X})_{s_j,t_j}  \in \ell_{\ddot{d}}$
to satisfy the following statement:
\begin{equation}
    i\prec j \iff \dot{d} < \ddot{d}
\end{equation}

Since we are facing a circuit composed by $\binom{n}{2}$ non-local operations, this is also a lower-bound on the optimal \texttt{E}-count. One last step is required to make the compiler reach this lower-bound. Specifically, for a given set of qubits $Q = \{q_1, q_2, \dots, q_n\}$ and processors $P = \{P_1, P_2, \dots, P_n\}$\footnote{Please refer to Sec. \ref{sec:graph} for the employed formalism.}, then $q_i \in P_i,\ \forall i$.
A remark follows:
\begin{remark}
    The compiler hits the lower-bound of $\binom{n}{2}$ for the hardest fan-in (fan-out) circuit, which means the \texttt{E}-count is optimal.
\end{remark}
From this result we also notice that we can apply the same idea to solve $\mathcal{L}_{\land(\texttt{Z})}$ circuits, achieving even better results. Specifically, first notice that fan-in and fan-out circuits do not limit to $\land(\texttt{X}^{\otimes m})$ operator, but also to $\land(\texttt{Z}^{\otimes m})$. Furthermore, since $\mathcal{L}_{\land(\texttt{Z})}$ are commuting circuit, one can always reduce it to a dense fan-in circuit, by proceeding in a greedy fashion outlined below.
\begin{enumerate}
    \item For each qubits, compute the number of $\land(\texttt{Z})$ in which they are involved;
    \item Enumerate the qubits by decreasing order and partition them such that $q_i\in P_i$;
    \item iterate over the enumeration and associate the maximal fan-in coming from $\mathcal{L}_{\land(\texttt{Z})}$;
    \item solve each layer with the spanning tree based compiler.
\end{enumerate}
In this way, any $\mathcal{L}_{\land(\texttt{Z})}$ circuit is expressed as a dense fan-in circuit, for which we can compute the optimal \texttt{E}-count. In other words:
\begin{remark}
    The minimal \texttt{E}-count for any $\mathcal{L}_{\land(\texttt{Z})}$ circuit can be computed efficiently with at most $n-1$ fan-in layers.
\end{remark}
\noindent{}Such a result relates to the same upper-bound of $n-1$ fan-in (fan-out) given in Ref. \cite{van2021constructing}.

Notice that, with this last remark, we now have all the ingredients to construct an efficient compiler for Clifford circuits\footnote{Which may also be expanded to universal circuits with methods as the one outlined in \cite{van2021constructing}.}. In fact, even if so far we proceeded under the assumption that each processor has only one computational qubit and link capacity one. This helped us to let out the potentiality of distributed architectures. We indeed managed to show very promising results. However, this approach does not limit the compiler to only distributed operations and can be extended to differ between local and non-local fan-in (fan-out) and even mixtures. Also the capacity can be generalized as the minimum spanning tree problem already works to minimize the total capacity.

It is also important to remember that the normal form we showed in Sec. \ref{sec:norm-cliff} is only one of several \cite{duncan2020graph,aaronson2004improved,dehaene2003clifford,maslov2018shorter,bataille2020reducing,bravyi2021hadamard}; each with implications on the structure of $\mathcal{L}_{\land(\texttt{Z})}$ and $\mathcal{L}_{\land(\texttt{X})}$ circuits. Hence, despite the appealing properties of the spanning tree approach, whenever the circuit is naturally composed by a few gates, a Steiner tree over fan-in layers may still be beaten by a quickest multi-commodity flow solver.

As an example, consider that for any $\mathcal{L}_{\land(\texttt{X})}$ circuit, it has already been discovered an algorithm to compute the minimum number of $\land(\texttt{X})$ operators \cite{patel2008optimal}. Even if the minimum number of $\land(\texttt{X})$ is an important result that, in perspective, closes positively the problem, this does not mean that optimal \texttt{E}-depth and \texttt{E}-count reside here. Hence, other techniques may be more practical -- e.g. see \cite{van2021constructing}.

The conclusion of our analysis is that a good compiler should be \textit{tunable}; meaning that it is able to use one or the other approach depending on the circuits in input\footnote{It may be also convenient to further refine the Steiner tree approach, e.g. whenever a single layer has more than one fan-in (fan-out) acting on independent qubits, then the layer can be solved by computing the Steiner three of the induced sub-graphs.}. For this reason we think a good conclusion of our investigation is a comparison between the two approaches. But before that, let us sum up the steps of a \textit{tunable compiler}:
\begin{enumerate}
    \item Transform the circuit in input into some normal form.
    \item Split the problem into $\mathcal{L}_{\land(\texttt{X})}$ and $\mathcal{L}_{\land(\texttt{Z})}$ sub-circuits.
    \item Analyse the structure of the sub-circuits (fan-in density, number of operators, etc.).
    \item For each sub-circuit, eventually reduce to a fan-in structure.
    \item For each sub-circuit, solve with the multi-commodity flow or the spanning tree approach.
\end{enumerate}

\section{Multi-commodity flow vs. Steiner trees}
\label{sec:flow_vs_tree}
Thanks to this new perspective for circuits -- i.e. as dense fan-in (fan-out) circuits -- we can now evaluate the performance of our previous compiler -- see Algorithm \ref{algo:approx} -- with the new one. In fact, while with the quickest multi-commodity flow we focused on minimizing the \texttt{E}-depth, we are now minimizing the \texttt{E}-count. By means of a careful construction of dense fan-in circuits, we can't do much optimization on the \texttt{E}-depth, while we can still work on the \texttt{E}-count.

To this aim, we focus on $\mathcal{L}_{\land(\texttt{Z})}$ circuits, where both compilers match with the constraints. However, we consider again the worst case scenario  --  i.e., fan-in circuits of $\binom{n}{2}$ operators -- where we expect the multi-commodity flow approach suffering the most. We already showed the optimality in the \texttt{E}-count for the spanning tree approach, but let us see now how the multi-commodity flow approach work in such a hard case, compared to Steiner trees.

\begin{figure}[h]
    \centering
    \begin{subfigure}{0.48\textwidth}
        \centering
        \includegraphics[width=0.85\textwidth]{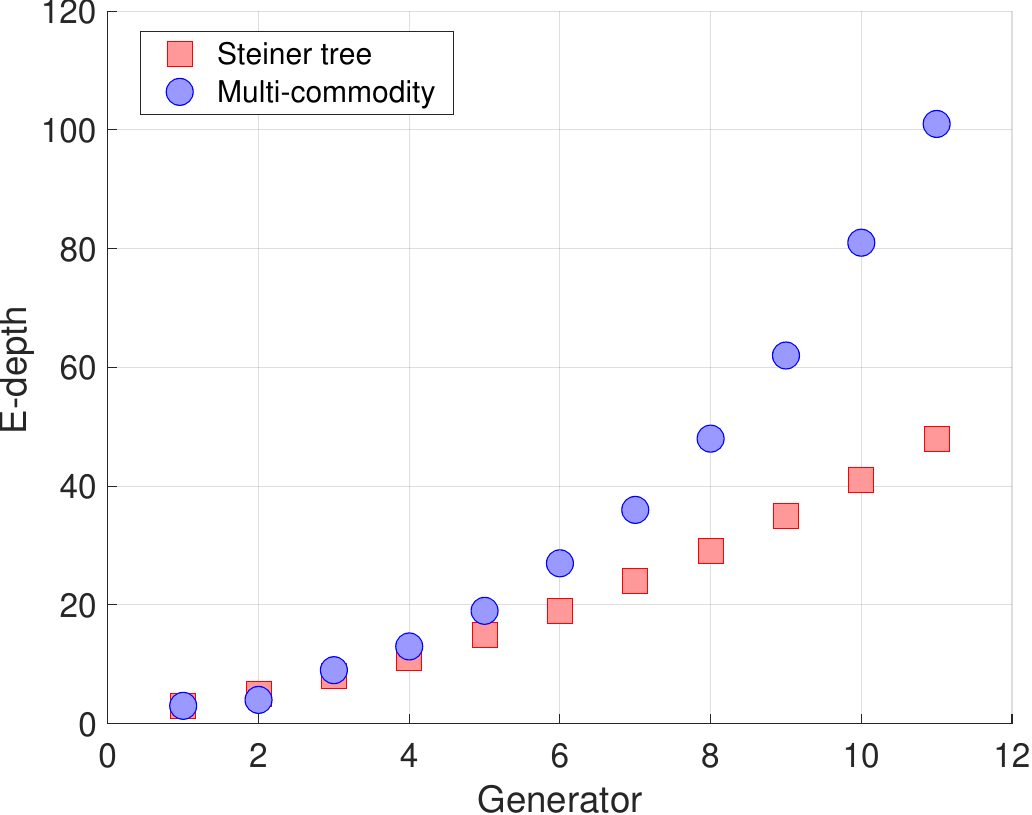}
        \caption{\texttt{E}-depth comparison.}
        \label{fig:e-depths}
    \end{subfigure}
    \hfill
    \begin{subfigure}{0.48\textwidth}
        \centering
        \includegraphics[width=0.85\textwidth]{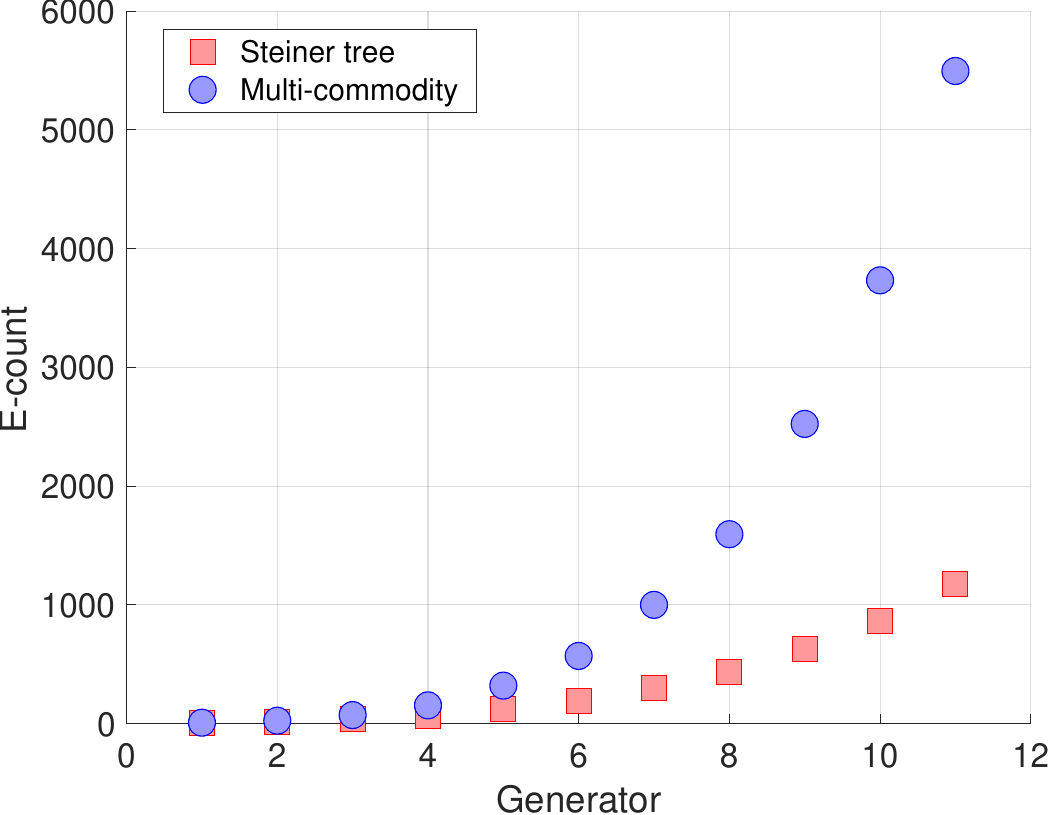}
        \caption{\texttt{E}-count comparison.}
        \label{fig:e-counts}
    \end{subfigure}
    \caption{This is a comparison between a compiler based on multi-commodity flow with one based on Steiner trees.}
    \label{fig:mcf-mst}
\end{figure}
The plot in Fig. \ref{fig:e-depths} shows that the \texttt{E}-depths achieved are still comparable, which preserve a side interest into employing the multi-commodity flow approach. However, from the plot in  Fig. \ref{fig:e-depths} emerges a neat advantage in terms of \texttt{E}-count. This is probably related to the formulation of the multi-commodity flow, which has primal interest into minimizing the \texttt{E}-depth, while the \texttt{E}-count is minimized only in the case this is convenient in terms of \texttt{E}-depth. Such an analysis brings the attention to different formulations, called \textit{path-based}, while the one we gave -- i.e. formulation \eqref{dqcc} -- is \textit{edge-based}. An interesting branch of research in this direction can be found in Ref. \cite{rabbie2022designing}.

\section{Conclusion: the importance of a compiler}
\label{sec:importance}
According to our envision of the full-stack development -- see Sec. \ref{sec:stack} -- a compiler will take care of creating a logical circuit which is compliant with the hardware. To understand how helpful a compiler can be, it is important to keep in mind that the word ``logical" should be interpreted in two different ways.
\begin{enumerate}
    \item The lower layers guarantee a reliable abstraction on which the compiler can operate logically.
    \item The compiler lighten the lower layers from all the tasks which are merely logical.
\end{enumerate}

Specifically to the second interpretation, consider that each layer of the stack has important and complex tasks to care about. Lower layers are more prone to treat real-time problems. In fact, the scheduler cares about synchronization and connectivity maximization. The hardware layer cares of being efficient in terms of, e.g., pushing the technologies at their maximal performance to get high fidelities and high success rates.
The above tasks are for sure hard. For this reason, identify what can be \textbf{delegated} to a logical reasoner -- i.e. the compiler -- may bring critical advantages to the overall architecture.

To get a practical intuition of what we are stating, consider for example our assumption for the \texttt{E} operator. We used this assumption to provide the reader with a model relatively easy to understand. But with a careful knowledge of the underlying architecture from an information theory perspective, something better can be achieved.
As an example, consider the stationary-flying system -- see Ch. \ref{ch:techs} -- generating and distributing entangled states. When it succeeds, it generally means that a heralded Bell state has been produced -- i.e. $\{\ket{\Psi^{\texttt{+}}},\ket{\Psi^{\texttt{-}}},\ket{\Phi^{\texttt{+}}},\ket{\Phi^{\texttt{-}}} \}$ --.

Now; if the reader believe, as we do, that delegating to the compiler to analyse post-processing will bring advantages to the general performance of the quantum computation, then it is clear that this approach does not stop to the model we created throughout this chapter. As a matter of fact, Bell states differ one another by Pauli corrections, which are usually treated at the hardware layer to provide the upper layers with $\ket{\Phi^{\texttt{+}}}$. But, even if local operations are very efficient, when it comes to multiple repeated steps -- as for the case of quantum computation -- every single gate avoided has a positive impact on the final fidelity. For the case we are considering now, the presence of a compiler enable the hardware to delegate the corrections, which now includes them to the big set of logical instructions to optimize. As basic example consider the circuit in Fig. \ref{fig:bell-rcx}, where a $\land(\texttt{X})$ runs with a different Bell state.
\begin{figure}[h]
    \centering
    \begin{quantikz}[thin lines, row sep={0.7cm,between origins}, column sep={0.4cm}]
        \lstick[]{$_{\ket{\varphi}}$}& \ctrl{1}\gategroup[wires=2,steps=3,style={draw=none,rounded corners,inner xsep=8pt,inner ysep=-1pt, fill=lime!10}, background, label style={rounded corners,label position=above,  yshift=0.08cm, fill=lime!10}, background]{$_{P_i}$} &\qw & \gate[style={fill=lime!10}]{_{\texttt{Z}^{\lnot\texttt{b}_2}}} &  \qw \\
        \lstick[2]{$_{\ket{\Psi^{\texttt{-}}}}$}& \targ{} & \measure[style={fill=lime!10}]{_{\langle {\texttt{Z}}\rangle,\ \texttt{b}_1}} \\
        & \ctrl{1}\gategroup[wires=2,steps=3,style={draw=none,rounded corners,inner xsep=8pt,inner ysep=-1pt, fill=cyan!10}, background,label style={rounded corners,label position=below, yshift=-0.37cm, fill=cyan!10}, background]{$_{P_j}$} & \measure[style={fill=cyan!10}]{_{\langle {\texttt{X}}\rangle,\ \texttt{b}_2}}\\
        \lstick[]{$_{\ket{\vartheta}}$}  & \targ{} & \qw & \gate[style={fill=cyan!10}]{_{\texttt{X}^{\lnot\texttt{b}_1}}} & \qw
    \end{quantikz}
    $\ \ \equiv$
    \begin{quantikz}[thin lines, row sep={0.7cm,between origins}, column sep={0.4cm}]
        \lstick[]{$_{\ket{\varphi}}$}&\ctrl{1}\gategroup[wires=1,steps=1,style={draw=none,rounded corners,inner xsep=8pt,inner ysep=5pt, fill=lime!10}, background, label style={rounded corners,label position=above,  yshift=0.08cm, fill=lime!10}, background]{$_{P_i}$} & \qw\\
        \lstick[]{$_{\ket{\vartheta}}$}& \targ{}\gategroup[wires=1,steps=1,style={draw=none,rounded corners,inner xsep=8pt,inner ysep=2.5pt, fill=cyan!10}, background,label style={rounded corners,label position=below, yshift=-0.37cm, fill=cyan!10}, background]{$_{P_j}$} & \qw
    \end{quantikz}
    \caption{Non-local $\land(\texttt{X})$ performed by means of $\ket{\Psi^{\texttt{-}}}$. This example shows how to avoid Pauli corrections.}
    \label{fig:bell-rcx}
\end{figure}

Let us see numerically what happens without delegating the correction to the compiler. Consider a single-qubit gate error of probability $\texttt{p} = 10^{-n}$. Then, for $10\cdot m$ non-local operations, the probability get worse of $m$ orders of magnitude, i.e. $10^{-n+m}$. On contrary, the compiler eliminates all the corrections, preserving the error probability to $10^{-n}$. The same reasoning applies to the running-time.

With our models we covered several interesting group of circuits\footnote{In Sec. \ref{sec:norm-cliff} we also gave a perspective on how our approach would eventually bring us to cover universal groups as well.}, for which the compiler eliminates every single correction from the quantum computation. The Bell state correction is no different. Specifically, instead of performing Pauli-corrections, the compiler keeps track of the logical error propagation caused by avoiding it. At the end of the computation the compiler provides the classical computer the necessary bit-flip corrections to get the right final state, just like we have done so far, by means of the rules introduced in Sec. \ref{sec:zx}.



\printbibliography[title=References,heading=subbibintoc]
\end{refsection}

%% file: compilation/overview.tex
\tikzset{every picture/.style={line width=0.75pt}} 

\begin{tikzpicture}[x=0.75pt,y=0.75pt,yscale=-1,xscale=1]

\draw  [color={rgb, 255:red, 0; green, 0; blue, 0 }  ,draw opacity=0 ][fill={rgb, 255:red, 189; green, 16; blue, 224 }  ,fill opacity=0.6 ] (439.77,99.84) -- (422.38,99.77) -- (422.36,104.08) -- (410.81,95.41) -- (422.44,86.84) -- (422.42,91.15) -- (439.81,91.22) -- cycle ;
\draw  [color={rgb, 255:red, 0; green, 0; blue, 0 }  ,draw opacity=0 ][fill={rgb, 255:red, 189; green, 16; blue, 224 }  ,fill opacity=0.6 ] (460.75,89.67) .. controls (460.75,87.46) and (462.54,85.67) .. (464.75,85.67) -- (575.75,85.67) .. controls (577.96,85.67) and (579.75,87.46) .. (579.75,89.67) -- (579.75,101.67) .. controls (579.75,103.88) and (577.96,105.67) .. (575.75,105.67) -- (464.75,105.67) .. controls (462.54,105.67) and (460.75,103.88) .. (460.75,101.67) -- cycle ;
\draw  [color={rgb, 255:red, 0; green, 0; blue, 0 }  ,draw opacity=0 ][fill={rgb, 255:red, 189; green, 16; blue, 224 }  ,fill opacity=0.6 ] (460.75,19.67) .. controls (460.75,17.46) and (462.54,15.67) .. (464.75,15.67) -- (575.75,15.67) .. controls (577.96,15.67) and (579.75,17.46) .. (579.75,19.67) -- (579.75,31.67) .. controls (579.75,33.88) and (577.96,35.67) .. (575.75,35.67) -- (464.75,35.67) .. controls (462.54,35.67) and (460.75,33.88) .. (460.75,31.67) -- cycle ;
\draw  [color={rgb, 255:red, 0; green, 0; blue, 0 }  ,draw opacity=0 ][fill={rgb, 255:red, 189; green, 16; blue, 224 }  ,fill opacity=0.6 ] (439.75,29.77) -- (422.36,29.69) -- (422.34,34) -- (410.79,25.33) -- (422.42,16.76) -- (422.4,21.07) -- (439.79,21.14) -- cycle ;
\draw  [color={rgb, 255:red, 0; green, 0; blue, 0 }  ,draw opacity=0 ][fill={rgb, 255:red, 74; green, 144; blue, 226 }  ,fill opacity=0.6 ] (334.45,115.39) -- (334.59,132.78) -- (338.9,132.74) -- (330.38,144.41) -- (321.66,132.88) -- (325.97,132.85) -- (325.83,115.46) -- cycle ;
\draw  [color={rgb, 255:red, 0; green, 0; blue, 0 }  ,draw opacity=0 ][fill={rgb, 255:red, 74; green, 144; blue, 226 }  ,fill opacity=0.6 ] (334.45,45.39) -- (334.59,62.78) -- (338.9,62.74) -- (330.38,74.41) -- (321.66,62.88) -- (325.97,62.85) -- (325.83,45.46) -- cycle ;
\draw  [color={rgb, 255:red, 0; green, 0; blue, 0 }  ,draw opacity=0 ][fill={rgb, 255:red, 74; green, 144; blue, 226 }  ,fill opacity=0.6 ] (271,90.33) .. controls (271,88.12) and (272.79,86.33) .. (275,86.33) -- (386,86.33) .. controls (388.21,86.33) and (390,88.12) .. (390,90.33) -- (390,102.33) .. controls (390,104.54) and (388.21,106.33) .. (386,106.33) -- (275,106.33) .. controls (272.79,106.33) and (271,104.54) .. (271,102.33) -- cycle ;
\draw  [color={rgb, 255:red, 0; green, 0; blue, 0 }  ,draw opacity=0 ][fill={rgb, 255:red, 74; green, 144; blue, 226 }  ,fill opacity=0.6 ] (270.6,20.73) .. controls (270.6,18.52) and (272.39,16.73) .. (274.6,16.73) -- (385.6,16.73) .. controls (387.81,16.73) and (389.6,18.52) .. (389.6,20.73) -- (389.6,32.73) .. controls (389.6,34.94) and (387.81,36.73) .. (385.6,36.73) -- (274.6,36.73) .. controls (272.39,36.73) and (270.6,34.94) .. (270.6,32.73) -- cycle ;
\draw  [color={rgb, 255:red, 0; green, 0; blue, 0 }  ,draw opacity=0 ][fill={rgb, 255:red, 189; green, 16; blue, 224 }  ,fill opacity=0.6 ] (220.75,21.19) -- (238.14,21.14) -- (238.13,16.83) -- (249.75,25.41) -- (238.18,34.07) -- (238.17,29.76) -- (220.78,29.81) -- cycle ;
\draw  [color={rgb, 255:red, 0; green, 0; blue, 0 }  ,draw opacity=0 ][fill={rgb, 255:red, 189; green, 16; blue, 224 }  ,fill opacity=0.6 ] (80.75,19.67) .. controls (80.75,17.46) and (82.54,15.67) .. (84.75,15.67) -- (195.75,15.67) .. controls (197.96,15.67) and (199.75,17.46) .. (199.75,19.67) -- (199.75,31.67) .. controls (199.75,33.88) and (197.96,35.67) .. (195.75,35.67) -- (84.75,35.67) .. controls (82.54,35.67) and (80.75,33.88) .. (80.75,31.67) -- cycle ;
\draw  [color={rgb, 255:red, 0; green, 0; blue, 0 }  ,draw opacity=0 ][fill={rgb, 255:red, 74; green, 144; blue, 226 }  ,fill opacity=0.6 ] (271,158.33) .. controls (271,156.12) and (272.79,154.33) .. (275,154.33) -- (386,154.33) .. controls (388.21,154.33) and (390,156.12) .. (390,158.33) -- (390,170.33) .. controls (390,172.54) and (388.21,174.33) .. (386,174.33) -- (275,174.33) .. controls (272.79,174.33) and (271,172.54) .. (271,170.33) -- cycle ;
\draw  [color={rgb, 255:red, 0; green, 0; blue, 0 }  ,draw opacity=0 ][fill={rgb, 255:red, 189; green, 16; blue, 224 }  ,fill opacity=0.6 ] (220.75,161.19) -- (238.14,161.14) -- (238.13,156.83) -- (249.75,165.41) -- (238.18,174.07) -- (238.17,169.76) -- (220.78,169.81) -- cycle ;
\draw  [color={rgb, 255:red, 0; green, 0; blue, 0 }  ,draw opacity=0 ][fill={rgb, 255:red, 189; green, 16; blue, 224 }  ,fill opacity=0.6 ] (80.75,159.67) .. controls (80.75,157.46) and (82.54,155.67) .. (84.75,155.67) -- (195.75,155.67) .. controls (197.96,155.67) and (199.75,157.46) .. (199.75,159.67) -- (199.75,171.67) .. controls (199.75,173.88) and (197.96,175.67) .. (195.75,175.67) -- (84.75,175.67) .. controls (82.54,175.67) and (80.75,173.88) .. (80.75,171.67) -- cycle ;
\draw  [color={rgb, 255:red, 0; green, 0; blue, 0 }  ,draw opacity=0 ][fill={rgb, 255:red, 74; green, 144; blue, 226 }  ,fill opacity=0.6 ] (334.45,187.39) -- (334.59,204.78) -- (338.9,204.74) -- (330.38,216.41) -- (321.66,204.88) -- (325.97,204.85) -- (325.83,187.46) -- cycle ;
\draw  [color={rgb, 255:red, 0; green, 0; blue, 0 }  ,draw opacity=0 ][fill={rgb, 255:red, 74; green, 144; blue, 226 }  ,fill opacity=0.6 ] (271,231.33) .. controls (271,229.12) and (272.79,227.33) .. (275,227.33) -- (386,227.33) .. controls (388.21,227.33) and (390,229.12) .. (390,231.33) -- (390,243.33) .. controls (390,245.54) and (388.21,247.33) .. (386,247.33) -- (275,247.33) .. controls (272.79,247.33) and (271,245.54) .. (271,243.33) -- cycle ;
\draw  [color={rgb, 255:red, 0; green, 0; blue, 0 }  ,draw opacity=0 ][fill={rgb, 255:red, 189; green, 16; blue, 224 }  ,fill opacity=0.6 ] (220.75,91.19) -- (238.14,91.14) -- (238.13,86.83) -- (249.75,95.41) -- (238.18,104.07) -- (238.17,99.76) -- (220.78,99.81) -- cycle ;
\draw  [color={rgb, 255:red, 0; green, 0; blue, 0 }  ,draw opacity=0 ][fill={rgb, 255:red, 189; green, 16; blue, 224 }  ,fill opacity=0.6 ] (80.75,89.67) .. controls (80.75,87.46) and (82.54,85.67) .. (84.75,85.67) -- (195.75,85.67) .. controls (197.96,85.67) and (199.75,87.46) .. (199.75,89.67) -- (199.75,101.67) .. controls (199.75,103.88) and (197.96,105.67) .. (195.75,105.67) -- (84.75,105.67) .. controls (82.54,105.67) and (80.75,103.88) .. (80.75,101.67) -- cycle ;
\draw  [color={rgb, 255:red, 0; green, 0; blue, 0 }  ,draw opacity=0 ][fill={rgb, 255:red, 189; green, 16; blue, 224 }  ,fill opacity=0.6 ] (439.77,169.84) -- (422.38,169.77) -- (422.36,174.08) -- (410.81,165.41) -- (422.44,156.84) -- (422.42,161.15) -- (439.81,161.22) -- cycle ;
\draw  [color={rgb, 255:red, 0; green, 0; blue, 0 }  ,draw opacity=0 ][fill={rgb, 255:red, 189; green, 16; blue, 224 }  ,fill opacity=0.6 ] (460.75,159.67) .. controls (460.75,157.46) and (462.54,155.67) .. (464.75,155.67) -- (575.75,155.67) .. controls (577.96,155.67) and (579.75,157.46) .. (579.75,159.67) -- (579.75,171.67) .. controls (579.75,173.88) and (577.96,175.67) .. (575.75,175.67) -- (464.75,175.67) .. controls (462.54,175.67) and (460.75,173.88) .. (460.75,171.67) -- cycle ;

\draw (330.5,96.33) node  [font=\fontsize{0.77em}{0.92em}\selectfont] [align=left] {{\fontfamily{pcr}\selectfont formulations}};
\draw (520.25,95.67) node  [font=\fontsize{0.77em}{0.92em}\selectfont] [align=left] {{\fontfamily{pcr}\selectfont dynamic flows \S\ref{sec:problem}}};
\draw (330.5,26.33) node  [font=\fontsize{0.77em}{0.92em}\selectfont] [align=left] {{\fontfamily{pcr}\selectfont assumptions}};
\draw (520.25,25.67) node  [font=\fontsize{0.77em}{0.92em}\selectfont] [align=left] {{\fontfamily{pcr}\selectfont logical network \S\ref{ch:techs}}};
\draw (140.25,25.67) node  [font=\fontsize{0.77em}{0.92em}\selectfont] [align=left] {{\fontfamily{pcr}\selectfont logical gates \S\ref{ch:essentials}}};
\draw (330.5,164.33) node  [font=\fontsize{0.77em}{0.92em}\selectfont] [align=left] {{\fontfamily{pcr}\selectfont evaluation}};
\draw (140.25,165.67) node  [font=\fontsize{0.77em}{0.92em}\selectfont] [align=left] {{\fontfamily{pcr}\selectfont random samples \S\ref{sec:norm-cliff}}};
\draw (330.5,237.33) node  [font=\fontsize{0.77em}{0.92em}\selectfont] [align=left] {{\fontfamily{pcr}\selectfont E-count, E-depth}};
\draw (140.25,95.67) node  [font=\fontsize{0.77em}{0.92em}\selectfont] [align=left] {{\fontfamily{pcr}\selectfont spanning trees \S\ref{sec:cliff-compiler}}};
\draw (520.25,165.67) node  [font=\fontsize{0.77em}{0.92em}\selectfont] [align=left] {{\fontfamily{pcr}\selectfont worst cases \S\ref{sec:flow_vs_tree}}};

\end{tikzpicture}

%% file: compilation/Figures/3qpu.tex
\tikzset{every picture/.style={line width=0.75pt}} 

\begin{tikzpicture}[x=0.65pt,y=0.65pt,yscale=-0.45,xscale=0.45]

\draw [color={rgb, 255:red, 189; green, 16; blue, 224 }  ,draw opacity=1 ][line width=1.5]    (48.12,252.22) -- (83.8,252.22) ;
\draw  [fill={rgb, 255:red, 74; green, 144; blue, 226 }  ,fill opacity=0.7 ][dash pattern={on 3.75pt off 2.25pt on 7.5pt off 1.5pt}][line width=1.5]  (357.78,190.78) .. controls (357.78,186.47) and (361.28,182.97) .. (365.6,182.97) -- (389.05,182.97) .. controls (393.37,182.97) and (396.87,186.47) .. (396.87,190.78) -- (396.87,214.69) .. controls (396.87,219.01) and (393.37,222.51) .. (389.05,222.51) -- (365.6,222.51) .. controls (361.28,222.51) and (357.78,219.01) .. (357.78,214.69) -- cycle ;
\draw  [draw opacity=0][fill={rgb, 255:red, 0; green, 0; blue, 0 }  ,fill opacity=1 ][line width=1.5]  (55.51,203.35) .. controls (55.51,197.69) and (60.09,193.1) .. (65.75,193.1) .. controls (71.41,193.1) and (76,197.69) .. (76,203.35) .. controls (76,209.01) and (71.41,213.6) .. (65.75,213.6) .. controls (60.09,213.6) and (55.51,209.01) .. (55.51,203.35) -- cycle ;
\draw [color={rgb, 255:red, 189; green, 16; blue, 224 }  ,draw opacity=1 ][line width=1.5]    (124.48,34.57) -- (288.41,33.69) ;
\draw [color={rgb, 255:red, 189; green, 16; blue, 224 }  ,draw opacity=1 ][line width=1.5]    (124.48,108.32) -- (288.41,107.44) ;
\draw  [fill={rgb, 255:red, 74; green, 144; blue, 226 }  ,fill opacity=0.7 ][dash pattern={on 5.63pt off 4.5pt}][line width=1.5]  (6,32.76) .. controls (6,18.82) and (17.3,7.52) .. (31.24,7.52) -- (110.79,7.52) .. controls (124.73,7.52) and (136.03,18.82) .. (136.03,32.76) -- (136.03,108.49) .. controls (136.03,122.43) and (124.73,133.73) .. (110.79,133.73) -- (31.24,133.73) .. controls (17.3,133.73) and (6,122.43) .. (6,108.49) -- cycle ;
\draw  [draw opacity=0][fill={rgb, 255:red, 0; green, 0; blue, 0 }  ,fill opacity=1 ][line width=1.5]  (17.43,109.76) .. controls (17.43,99.97) and (25.11,92.04) .. (34.58,92.04) .. controls (44.05,92.04) and (51.73,99.97) .. (51.73,109.76) .. controls (51.73,119.55) and (44.05,127.48) .. (34.58,127.48) .. controls (25.11,127.48) and (17.43,119.55) .. (17.43,109.76) -- cycle ;
\draw [line width=1.5]    (34.58,53.73) -- (34.58,92.04) ;
\draw  [draw opacity=0][fill={rgb, 255:red, 0; green, 0; blue, 0 }  ,fill opacity=1 ][line width=1.5]  (18.83,34.57) .. controls (18.83,24.78) and (26.5,16.84) .. (35.97,16.84) .. controls (45.44,16.84) and (53.12,24.78) .. (53.12,34.57) .. controls (53.12,44.35) and (45.44,52.29) .. (35.97,52.29) .. controls (26.5,52.29) and (18.83,44.35) .. (18.83,34.57) -- cycle ;
\draw  [draw opacity=0][fill={rgb, 255:red, 0; green, 0; blue, 0 }  ,fill opacity=1 ][line width=1.5]  (90.19,108.32) .. controls (90.19,98.54) and (97.86,90.6) .. (107.33,90.6) .. controls (116.8,90.6) and (124.48,98.54) .. (124.48,108.32) .. controls (124.48,118.11) and (116.8,126.05) .. (107.33,126.05) .. controls (97.86,126.05) and (90.19,118.11) .. (90.19,108.32) -- cycle ;
\draw  [draw opacity=0][fill={rgb, 255:red, 0; green, 0; blue, 0 }  ,fill opacity=1 ][line width=1.5]  (90.19,34.57) .. controls (90.19,24.78) and (97.86,16.84) .. (107.33,16.84) .. controls (116.8,16.84) and (124.48,24.78) .. (124.48,34.57) .. controls (124.48,44.35) and (116.8,52.29) .. (107.33,52.29) .. controls (97.86,52.29) and (90.19,44.35) .. (90.19,34.57) -- cycle ;
\draw [line width=1.5]    (90.19,36) -- (53.12,36) ;
\draw [color={rgb, 255:red, 189; green, 16; blue, 224 }  ,draw opacity=1 ][line width=1.5]    (394.06,33.69) -- (535.79,34.51) ;
\draw  [fill={rgb, 255:red, 74; green, 144; blue, 226 }  ,fill opacity=0.7 ][dash pattern={on 5.63pt off 4.5pt}][line width=1.5]  (275.58,31.88) .. controls (275.58,17.94) and (286.88,6.64) .. (300.82,6.64) -- (380.38,6.64) .. controls (394.32,6.64) and (405.62,17.94) .. (405.62,31.88) -- (405.62,107.61) .. controls (405.62,121.55) and (394.32,132.85) .. (380.38,132.85) -- (300.82,132.85) .. controls (286.88,132.85) and (275.58,121.55) .. (275.58,107.61) -- cycle ;
\draw  [draw opacity=0][fill={rgb, 255:red, 0; green, 0; blue, 0 }  ,fill opacity=1 ][line width=1.5]  (288.41,107.44) .. controls (288.41,97.66) and (296.09,89.72) .. (305.55,89.72) .. controls (315.02,89.72) and (322.7,97.66) .. (322.7,107.44) .. controls (322.7,117.23) and (315.02,125.17) .. (305.55,125.17) .. controls (296.09,125.17) and (288.41,117.23) .. (288.41,107.44) -- cycle ;
\draw  [draw opacity=0][fill={rgb, 255:red, 0; green, 0; blue, 0 }  ,fill opacity=1 ][line width=1.5]  (288.41,33.69) .. controls (288.41,23.9) and (296.09,15.96) .. (305.55,15.96) .. controls (315.02,15.96) and (322.7,23.9) .. (322.7,33.69) .. controls (322.7,43.47) and (315.02,51.41) .. (305.55,51.41) .. controls (296.09,51.41) and (288.41,43.47) .. (288.41,33.69) -- cycle ;
\draw  [draw opacity=0][fill={rgb, 255:red, 0; green, 0; blue, 0 }  ,fill opacity=1 ][line width=1.5]  (359.77,107.44) .. controls (359.77,97.66) and (367.45,89.72) .. (376.91,89.72) .. controls (386.38,89.72) and (394.06,97.66) .. (394.06,107.44) .. controls (394.06,117.23) and (386.38,125.17) .. (376.91,125.17) .. controls (367.45,125.17) and (359.77,117.23) .. (359.77,107.44) -- cycle ;
\draw  [draw opacity=0][fill={rgb, 255:red, 0; green, 0; blue, 0 }  ,fill opacity=1 ][line width=1.5]  (359.77,33.69) .. controls (359.77,23.9) and (367.45,15.96) .. (376.91,15.96) .. controls (386.38,15.96) and (394.06,23.9) .. (394.06,33.69) .. controls (394.06,43.47) and (386.38,51.41) .. (376.91,51.41) .. controls (367.45,51.41) and (359.77,43.47) .. (359.77,33.69) -- cycle ;
\draw [line width=1.5]    (375.52,51.41) -- (375.52,89.72) ;
\draw  [fill={rgb, 255:red, 74; green, 144; blue, 226 }  ,fill opacity=0.7 ][dash pattern={on 5.63pt off 4.5pt}][line width=1.5]  (522.97,34.14) .. controls (522.97,20.2) and (534.27,8.9) .. (548.21,8.9) -- (627.76,8.9) .. controls (641.7,8.9) and (653,20.2) .. (653,34.14) -- (653,109.86) .. controls (653,123.8) and (641.7,135.11) .. (627.76,135.11) -- (548.21,135.11) .. controls (534.27,135.11) and (522.97,123.8) .. (522.97,109.86) -- cycle ;
\draw  [draw opacity=0][fill={rgb, 255:red, 0; green, 0; blue, 0 }  ,fill opacity=1 ][line width=1.5]  (534.4,108.26) .. controls (534.4,98.48) and (542.08,90.54) .. (551.55,90.54) .. controls (561.02,90.54) and (568.69,98.48) .. (568.69,108.26) .. controls (568.69,118.05) and (561.02,125.99) .. (551.55,125.99) .. controls (542.08,125.99) and (534.4,118.05) .. (534.4,108.26) -- cycle ;
\draw [line width=1.5]    (551.55,52.23) -- (551.55,90.54) ;
\draw  [draw opacity=0][fill={rgb, 255:red, 0; green, 0; blue, 0 }  ,fill opacity=1 ][line width=1.5]  (535.79,34.51) .. controls (535.79,24.72) and (543.47,16.78) .. (552.94,16.78) .. controls (562.41,16.78) and (570.08,24.72) .. (570.08,34.51) .. controls (570.08,44.29) and (562.41,52.23) .. (552.94,52.23) .. controls (543.47,52.23) and (535.79,44.29) .. (535.79,34.51) -- cycle ;
\draw  [draw opacity=0][fill={rgb, 255:red, 0; green, 0; blue, 0 }  ,fill opacity=1 ][line width=1.5]  (607.15,108.26) .. controls (607.15,98.48) and (614.83,90.54) .. (624.3,90.54) .. controls (633.77,90.54) and (641.44,98.48) .. (641.44,108.26) .. controls (641.44,118.05) and (633.77,125.99) .. (624.3,125.99) .. controls (614.83,125.99) and (607.15,118.05) .. (607.15,108.26) -- cycle ;
\draw [color={rgb, 255:red, 0; green, 0; blue, 0 }  ,draw opacity=1 ][line width=1.5]    (361.19,252.22) -- (396.87,252.22) ;
\draw [line width=1.5]    (607.15,34.51) -- (570.08,34.51) ;
\draw [line width=1.5]    (88.79,109.76) -- (51.73,109.76) ;
\draw [line width=1.5]    (359.77,33.69) -- (322.7,33.69) ;
\draw [line width=1.5]    (305.55,51.41) -- (305.55,89.72) ;
\draw  [draw opacity=0][fill={rgb, 255:red, 0; green, 0; blue, 0 }  ,fill opacity=1 ][line width=1.5]  (607.15,34.51) .. controls (607.15,24.72) and (614.83,16.78) .. (624.3,16.78) .. controls (633.77,16.78) and (641.44,24.72) .. (641.44,34.51) .. controls (641.44,44.29) and (633.77,52.23) .. (624.3,52.23) .. controls (614.83,52.23) and (607.15,44.29) .. (607.15,34.51) -- cycle ;
\draw [line width=1.5]    (624.3,52.23) -- (624.3,90.54) ;

\draw (85.8,252.22) node [anchor=west] [inner sep=0.75pt]  [font=\fontsize{0.67em}{0.8em}\selectfont] [align=left] {\textit{{\fontfamily{pcr}\selectfont \textbf{ Entanglement link}}}};
\draw (398.87,203.22) node [anchor=west] [inner sep=0.75pt]  [font=\fontsize{0.67em}{0.8em}\selectfont] [align=left] {\textit{{\fontfamily{pcr}\selectfont \textbf{ Processor}}}};
\draw (88,202.23) node [anchor=west] [inner sep=0.75pt]  [font=\fontsize{0.67em}{0.8em}\selectfont] [align=left] {\textit{{\fontfamily{pcr}\selectfont \textbf{ Qubit}}}};
\draw (71.02,70.63) node  [font=\small,color={rgb, 255:red, 255; green, 255; blue, 255 }  ,opacity=1 ]  {$P_{1}$};
\draw (35.97,34.57) node  [font=\footnotesize,color={rgb, 255:red, 255; green, 255; blue, 255 }  ,opacity=1 ]  {$q_{1}$};
\draw (107.33,34.57) node  [font=\footnotesize,color={rgb, 255:red, 255; green, 255; blue, 255 }  ,opacity=1 ]  {$c_{1}$};
\draw (34.58,109.76) node  [font=\footnotesize,color={rgb, 255:red, 255; green, 255; blue, 255 }  ,opacity=1 ]  {$q_{2}$};
\draw (107.33,108.32) node  [font=\footnotesize,color={rgb, 255:red, 255; green, 255; blue, 255 }  ,opacity=1 ]  {$c_{2}$};
\draw (340.6,69.75) node  [font=\small,color={rgb, 255:red, 255; green, 255; blue, 255 }  ,opacity=1 ]  {$P_{2}$};
\draw (305.55,33.69) node  [font=\footnotesize,color={rgb, 255:red, 255; green, 255; blue, 255 }  ,opacity=1 ]  {$c_{3}$};
\draw (376.91,33.69) node  [font=\footnotesize,color={rgb, 255:red, 255; green, 255; blue, 255 }  ,opacity=1 ]  {$c_{5}$};
\draw (305.55,107.44) node  [font=\footnotesize,color={rgb, 255:red, 255; green, 255; blue, 255 }  ,opacity=1 ]  {$c_{4}$};
\draw (376.91,107.44) node  [font=\footnotesize,color={rgb, 255:red, 255; green, 255; blue, 255 }  ,opacity=1 ]  {$q_{3}$};
\draw (587.98,72) node  [font=\small,color={rgb, 255:red, 255; green, 255; blue, 255 }  ,opacity=1 ]  {$P_{3}$};
\draw (552.94,34.51) node  [font=\footnotesize,color={rgb, 255:red, 255; green, 255; blue, 255 }  ,opacity=1 ]  {$c_{6}$};
\draw (551.55,108.26) node  [font=\footnotesize,color={rgb, 255:red, 255; green, 255; blue, 255 }  ,opacity=1 ]  {$q_{4}$};
\draw (624.3,108.26) node  [font=\footnotesize,color={rgb, 255:red, 255; green, 255; blue, 255 }  ,opacity=1 ]  {$q_{5}$};
\draw (398.87,252.22) node [anchor=west] [inner sep=0.75pt]  [font=\fontsize{0.67em}{0.8em}\selectfont] [align=left] {\textit{{\fontfamily{pcr}\selectfont \textbf{ Local coupling}}}};
\draw (624.3,34.51) node  [font=\footnotesize,color={rgb, 255:red, 255; green, 255; blue, 255 }  ,opacity=1 ]  {$q_{6}$};

\end{tikzpicture}

%% file: compilation/dqcc.tex
\section{Distributed quantum circuit compilation problem}
\label{sec:problem}

Usually, in the literature dealing with compiler design \cite{ZulWil-19, ItoRayIma-19, WilBurZul-19, FerCacAmo-21}, a circuit is encoded as a set of \textit{layers}. Formally, a layer is a set $\ell$ of independent operators, meaning that each operator in $\ell$ acts on a different collection of qubits. A circuit is an enumeration of layers $\mathcal{L} = \{\ell_1, \ell_2, \dots, \ell_{|\mathcal{L}|}\}$, where the cardinality is also commonly referred as circuit \textit{depth}. 


A quantum programmer writes a logical circuit, abstracting from the real architecture and assuming that qubits are fully connected, i.e., any couple of qubits can perform a $\land(\texttt{X})$ operation -- defined as in Eq. \eqref{eq:cu} -- directly.
Such an abstraction holds also when stepping to distributed architectures. 

However, NISQ architectures do not always provide full coupling. This is especially true in the distributed scenarios. As a consequence, there must be an interface -- namely, a compiler -- able to map an abstract circuit to an equivalent one, but meeting the available coupling. In general, such a mapping implies overhead in terms of circuit depth. Therefore, finding a mapping with minimum depth overhead is an optimization problem. We refer to it as the \textit{quantum circuit compilation} problem (\texttt{QCC}), which is proved to be \texttt{NP}-hard \cite{BotKisMar-18}. 
Its version on distributed architectures, 
which we refer to as the \textit{distributed quantum circuit compilation} problem (\texttt{DQCC}), is likely 
to be at least as hard as \texttt{QCC}. In fact, while in \texttt{QCC} we deal with local connectivity restrictions, in \texttt{DQCC} local connectivity stands alongside with remote connectivity -- i.e., the entanglement links --, which is less dense than the local one\footnote{Because the more communication qubits there are, the less computing resources are available.}. Furthermore, performing a remote operation is much more time consuming than a local operation. Just consider that a remote operation relies on communication of both quantum and classical information. 

The above reasons make telegates the bottleneck in distributed computing. Therefore, they are worth of dedicated analysis to minimize their impact. 



\subsection{Objective function}
\label{sec:of}
To optimize a circuit, the first thing we need to do is choosing an objective function to rate the expected performance of a circuit. A common approach is to evaluate only those operators which are somehow a bottleneck to computation. Considering the universal gate set $\mathbb{C}^{\texttt{+}}$ -- defined in Sec. \ref{sec:cliff}, in the context of fault-tolerant quantum computing \cite{Got-98}, the bottleneck is the $\texttt{Z}^{\sfrac{1}{4}}$ operator \cite{Sel-13,AmyMasMos-14}, since error correction protocols are designed for the Clifford group $\mathbb{C}$. Conversely, on current NISQ technologies, the bottleneck lies in the interaction between qubits -- as for the case of $\land(\texttt{X})$. The relevant metric can either be the number of occurrences of some operator $\texttt{O}$, namely the $\texttt{O}$-count, or the number of layers containing $\texttt{O}$ at least once, namely the $\texttt{O}$-depth.
To rate a compiled circuit on distributed architectures, we do something along the lines of this approach. Specifically, the bottleneck are the non-local $\land(\texttt{X})$ (and $\land(\texttt{Z})$) operators, each of which implies one occurrence of entanglement generation and distribution stage. We refer to such a stage as the \texttt{E} operator. Therefore, we will rate a circuit by means of its \texttt{E}-depth and \texttt{E}-count.

\subsection{Modeling the time domain}
\label{sec:time}
It should be clear that \texttt{E} has central interest in our treating. In fact, we are also going to model the time by scanning it as \texttt{E} occurs.
Specifically, notice that link generations among different couples of qubits are independent. For this reason we assume that all the possible links generate simultaneously and, as soon as all the states are measured, a new round of simultaneous generations begins.

Clearly, after that a measurement generates a boolean \texttt{b}, there is at least one post-processing operator that need to wait for that boolean to arrive. Generally speaking, the longer the path the more time \texttt{b} takes to reach its destination.  
We need to account for that by a proper model. To this aim, we do some observations.
\begin{remark}Consider a generic single-qubit unitary operator $\texttt{U}$. The time required to perform $\texttt{U}^{\texttt{b}}$ is given by the sum of the travel time of \texttt{b} plus the time to perform $\texttt{U}$. However, the traveling of \texttt{b} is independent from computation and any operation preceding $\texttt{U}^{\texttt{b}}$ can run. Hence, we compactly refer to the post-processing waiting-time as $\Delta_{\texttt{U}^{\texttt{b}}}$. A second observation is that the travel of \texttt{b} is also independent by entanglement link creations, which we assume to take time $\Delta_{\texttt{E}}$. It is, therefore, also reasonable to assume $\Delta_{\texttt{U}^{\texttt{b}}} \lesssim \Delta_{\texttt{E}}$ because of the following observation: even if \texttt{b} need to cover a longer distance than the one covered by \texttt{E}, \texttt{b} relies on classical technologies, which are way more efficient\footnote{The design of a distributed quantum architecture can easily adapt to satisfy requirements coming from assumptions on classical technologies, since these are very advanced.} than entanglement generation and distribution protocols. 
For this reason, in our treating we neglect $\Delta_{\texttt{U}^{\texttt{b}}}$, since it happens in parallel with $\Delta_{\texttt{E}}$. Furthermore, in Secs. \ref{sec:norm-cliff} and \ref{sec:cliff-compiler} we will focus on groups of circuits where all the post-processing operations are fully separated from the quantum computation.
\end{remark}


Stemming from this, we can model the time domain as a discrete set of steps $\tau \in \{1,2, \dots, d\}$, where $d$ is an unknown time horizon, which is also the \texttt{E}-depth. At the beginning of each time step $\tau$, the whole set of entanglement links is available for telegates.
Notice that most of the local operators are expected to run during the creation of the links. Because we relate them to the following inequality
\begin{equation}
\label{eq:time}
    \textcolor{black}{\Delta_{\texttt{E}} \gg \Delta_{\land(\texttt{U})}, \Delta_{\texttt{U}}},
\end{equation}
where \texttt{U} is a single-qubit unitary operator. Therefore, since \texttt{E} is independent from local operators, we can always attempt to run these while \texttt{E} is running -- and also while classical bits \texttt{b} are traveling, as explained in Sec. \ref{sec:e-path}.



\subsection{Modeling the distributed architecture}
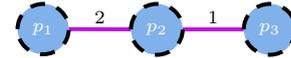
\begin{wrapfigure}{r}{6cm}
    \centering
    \input{compilation/Figures/quotient.tex}
    \caption{Quotient graph derived from Fig. \ref{fig:3qpu}. The processors become the nodes, the entanglement links between a couple of processors gather into one edge, with capacity equal to the number of original links \cite{cuomo2021optimized}. 
    }
    \label{fig:quotient}
    \hrulefill
\end{wrapfigure}
In light of the above observations, it is reasonable and convenient to consider the whole processor as a network node, and define a function $c$ that provides the number of available links between two processors. Specifically, we first formalized a distributed architecture as the network graph $\mathcal{N} = (V,P,F)$ introduced in subsection \ref{sec:graph}; this step was important to understand the interior behavior of remote operations from a qubit perspective. However, now it is useful to re-state it to a more compact encoding, which highlights the main bottleneck of a distributed quantum architecture, the entanglement links. Formally speaking, we will consider a \textit{quotient graph} of $\mathcal{N}$. 

To not further weigh down the formalism, we re-model the instance, by considering as main nodes, the processors, corresponding to an enumeration for the partition $P$, i.e., $P = \{p_1,p_2,\dots,p_n\}$.
All the entanglement links, connecting the same couple of processors, now collapse to an only edge with integer capacity $c$, describing how many parallel entanglement links the two processors supplies\footnote{Notice that this has no impact when the quantum processor is based has a full-connected topology, as in the case of Ion-traps.}. We refer to this sets of edges as 
\[E \subseteq \bigcup_{i,j\ :\ i\neq j} p_i\times p_j.\]
Hence, the new undirected graph is $\mathcal{Q} = (P,E,c)$. With this reformulation a remote operation will refer to a \textit{control processor} and a \textit{target processors} -- e.g., $\land(\texttt{X})_{u,v}$ with $p_u,p_v \in P$.

In Fig. \ref{fig:quotient} we show the quotient graph related to the toy architecture of Fig. \ref{fig:3qpu}. 



\subsection{Single layer formulation}
\label{sec:layer}
Consider a basic circuit expressed as the singleton $\mathcal{L} = \{\ell\}$.
Assume that in $\ell$ there occur $k$ $\land(\texttt{X})$ -- or $\land(\texttt{Z})$ -- operators. From a logical perspective, all the $k$ operators can run in parallel, by definition of layer. 
In other words, if the architecture connectivity had infinite capacity -- i.e., $c(e) = \infty,\ \forall e \in E$ -- we could run $\mathcal{L}$ with \texttt{E}-depth $1$, that is optimal.
As the capacity values decrease, the optimal \texttt{E}-depth value grows, up to \texttt{E}-depth $k$ in the worst-case.

Let us formulate an optimization problem for the single-layer case -- we will introduce a generalization to any circuit in subsection \ref{sec:layers}.
Specifically, the \textit{quickest multi-commodity flow} \cite{FleSku-02} wraps this basic scenario. 

In brief, the goal is to find a flow over time which satisfy the constraints imposed by a set of so-called
commodities, which are going to represent the $\land(\texttt{X})$ of a quantum circuit. The less time the flow takes, the better. 
To formalize this problem, one can directly model an objective function that evaluates a flow by the time it takes. This is an approach employed in Ref. \cite{LinJai-14}, but for single commodity. Alternatively, authors in Ref. \cite{FleSku-02} propose to start from a formulation of the \textit{multi-commodity flow} problem over time \texttt{MCF}$_d$, where $d$ is a given \textit{time horizon}\footnote{The choice of using letter $d$ should highlight that the time horizon is going to be the \texttt{E}-depth.}, namely a maximal number of time steps in which the flow is constrained. We prefer this latter way because dynamic flows like \texttt{MCF}$_d$ has been deeply studied since long time ago \cite{ForFul-58, FulFor-58}. Furthermore, even if this approach has an important drawback, explained at the end of this sub-section, it does not apply to our scenario.

\subsubsection{Commodities}
To formulate \texttt{MCF}$_d$, first, we enumerate the occurrences of two-qubit remote operators in $\mathcal{L}$ as a set of commodities $[k] = \{1,2,\dots, k\}$. A set of couples source-sink nodes associates to the commodities. To do that, let $\mathbf{s} = ({s_1}, {s_2}, \dots {s_k})$ and $\mathbf{t} = ({t_1}, {t_2}, \dots {t_k})$ be two vectors induced by the operators $\land(\texttt{X})$ -- or $\land(\texttt{Z})$ -- in $\mathcal{L}$\footnote{We need to use vector notation to admit repetitions.} such that,
\begin{equation*}
    \land(\texttt{X})_{s_i,t_i} \in \ell \iff \exists{i}\in [k]\ :\ p_{s_i}, p_{t_i} \in P.
\end{equation*}
Namely, $p_{s_i}$ ($p_{t_i}$) is the processor where the control (target) qubit of operation $i$ occurs.

\subsubsection{Decision variables}
The decision variables of the optimization problem are the time-dependent functions  $f_{e,i}(\tau) \in \{0,1\}$, indicating the flow on edge $e \in E$ dedicated to operation $i \in [k]$ at time $\tau$. The function has a binary co-domain because an operation $i$ uses at most one entanglement link.


\subsubsection{Constraints}
As usual, the first constraint we introduce is the \textit{flow conservation} constraint. Formally, $\forall i\in [k]$, $\forall \tau \in [d]$ and $\forall p_j \in P\smallsetminus\{p_{s_i},p_{t_i}\}$ the following holds:
\begin{equation}
\label{c1}
    \sum_{e \in \delta^{-}(p_j)}f_{e,i}(\tau) - \sum_{e \in \delta^{+}(p_j)} f_{e,i}(\tau) = 0
\end{equation}
where $\delta^-, \delta^+ : P \rightarrow E$ are the standard functions outputting the set of entering and exiting edges of the input node, respectively.

Since a flow $f_{e,i}(\tau) = 1$ identifies the usage of an entanglement link in $e$ to perform $i$, we need to guarantee that the flow going through intermediate links of a path does not stop there. Conversely, whenever an end point of the path occurs in the control or target processor -- i.e., $p_{s_i}$ or $p_{t_i}$ --, the \textit{operation demand} -- or \textit{commodity demand} -- constraint holds instead of the conservation constraint.
Namely, $\forall i \in [k]$, this can be written as:

\begin{equation}
\label{c2.1}
\sum_{e \in \delta^{-}(p_{s_i})}\sum_{\tau \in [d]} f_{e,i}(\tau) - \sum_{e \in \delta^{+}(p_{s_i})}\sum_{\tau \in [d]} f_{e,i}(\tau) = -1
\end{equation}

\begin{equation}
\label{c2.2}
\sum_{e \in \delta^{-}(p_{t_i})}\sum_{\tau \in [d]} f_{e,i}(\tau) - \sum_{e \in \delta^{+}(p_{t_i})}\sum_{\tau \in [d]} f_{e,i}(\tau) = +1
\end{equation}
The above constraint explicitly requests that a flow dedicate to $i$ reaches its target $p_{t_i}$, without exiting. Symmetrically, it leaves its control processor $p_{s_i}$ without returning.

Notice that constraint \eqref{c1} forces the operation demand to be satisfied within a single time-step.

The last constraint ensures that, at any time step, the number of operations does not exceed the entanglement resources. Hence, $\forall e \in E$ and $\forall \tau \in [d]$, we introduce a \textit{capacity bound}:
\begin{equation}
    \label{c3}
    \sum_{i \in [k]}{f_{e,i}(\tau)} \leq c(e)
\end{equation}

Ultimately, the objective function is the total flow $f = \sum_{e \in E}\sum_{i \in [k]}\sum_{\tau}{f_{e,i}(\tau)}$.

By gathering the above equations, we obtain the Integer Linear Programming formulation \eqref{mcf}, which models \texttt{MCF}$_d$. 
A flow $f$ perfectly matches a set of entanglement paths used by the telegates. 

\begin{figure}[ht]
\centering
   \begin{mini}
    {}{f = \sum_{e \in E}\sum_{i \in [k]}\sum_{\tau \in [d]}{f_{e,i}(\tau)}}{}{}
    \label{mcf}
    \addConstraint{\sum_{e \in \delta^{-}(p_j)}f_{e,i}(\tau) - \sum_{e \in \delta^{+}(p_j)} f_{e,i}(\tau) = 0}{}{\forall i \in [k], \forall{\tau}\in [d], \forall p_j \in P\smallsetminus\{p_{s_i},p_{t_i}\}}
    \addConstraint{\sum_{e \in \delta^{-}(p_{s_i})}\sum_{\tau \in [d]} f_{e,i}(\tau) - \sum_{e \in \delta^{+}(p_{s_i})}\sum_{\tau\in [d]} f_{e,i}(\tau) = -1\ \ \ }{}{\forall i \in [k]} \addConstraint{\sum_{e \in \delta^{-}(p_{t_i})}\sum_{\tau \in [d]} f_{e,i}(\tau) - \sum_{e \in \delta^{+}(p_{t_i})}\sum_{\tau\in [d]} f_{e,i}(\tau) = +1\ \ \ }{}{\forall i \in [k]}
    \addConstraint{\sum_{i \in [k]}{f_{e,i}(\tau)} \leq c(e)}{}{\forall e \in E,  \forall \tau \in [d]}  
    \end{mini}
\end{figure}



\begin{wrapfigure}{R}{8cm}
\begin{algorithm}[H]
\DontPrintSemicolon
\KwIn{$\mathcal{Q}, [k]$}
\KwOut{$d$}
  $L \leftarrow 1, R \leftarrow k$\;
  \While{$L \leq R$}{
        {
            $\bar{d} \leftarrow \floor{\frac{L + R}{2}}$\;
            $S \leftarrow$ \texttt{MCF}$_{\bar{d}}(\mathcal{Q}, [k])$\;
            \uIf{$S$ \textbf{is feasible}}
            {
                $d \leftarrow \bar{d}$\;
                $R \leftarrow \bar{d}-1$\;
            }
            \Else
            {
                $L \leftarrow \bar{d} + 1$\;
            }
        }
    }
\caption{\small{Quickest multi-commodity flow}}
\label{algo:binary}
\end{algorithm}
\end{wrapfigure}
Notice that solutions with cycles are in general feasible, but are senseless in our scenario.
By expressing the problem as a minimization of $f$, a solver will avoid any cycle and will try to use as few entanglement links as possible.
Once defined a solver for $\texttt{MCF}_d$, we just need to use it as proposed in Ref. \cite{FleSku-02}, namely the solver occurs as sub-routine within a binary research on the minimum time where a feasible solution exists. 
Since the research space is over time, the algorithm is, in general, pseudo-logarithmic. 

\noindent{}Specifically to our case, we already know that the worst solution is where all the operations run in sequence -- i.e., \texttt{E}-depth equal to the amount $k$ of telegates. Therefore, the time horizon is upper-bounded by $k$ and the binary search has $\log{k}$ calls to the sub-routine.
Algorithm \ref{algo:binary} shows the steps. Notice that it makes use of an undetermined solver for $\texttt{MCF}_d$. Since we are facing an \texttt{NP}-hard problem, this means that a real implementation would generally look for sub-optimal solutions.



Unfortunately, standard $\texttt{MCF}_d$ cannot catch the whole features of \texttt{DQCC} when $\mathcal{L} = \{\ell_1, \ell_2, \dots, \ell_{|\mathcal{L}|}\}$; we need to consider that operations in $[k]$ are somehow related each other by a logic determined by $\mathcal{L}$. Hence, in Sec. \ref{sec:layers} we will model such relations by introducing extra constraints.

\subsubsection{Transformation to direct graph}
\begin{wrapfigure}{r}{6.5cm}
    \centering
    \input{compilation/Figures/digraph.tex}
    \caption{Mapping from an undirected graph to a directed one working for any multi-commodity flow problem. The transformation undergoes with a constant overhead in the number of nodes and edges.}
    \label{fig:digraph}
    \hrulefill
\end{wrapfigure}
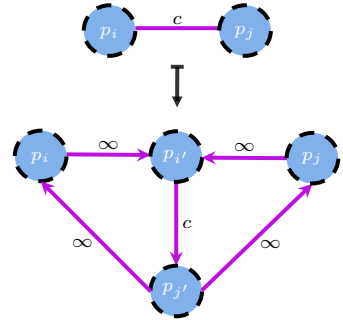
Since the literature dealing with multi-commodity flows usually assume a direct graph, we here report a mapping method from an undirected graph to an equivalent one with direct edges. This would bring just a constant overhead in the space, while it would not affect any approximation factor which a solver would rely on. Fig. \ref{fig:digraph} comes from \cite{ahuja1988network}. It is a fast approach to map an undirected multi-commodity flow problem to a directed one. Specifically, for each couple of nodes $p_i,p_j$ connected by an edge with capacity $c$, one have to introduce two extra nodes, say $p_{i'},p_{j'}$ and connect them with the direct edge $(p_{i'},p_{j'})$ of capacity $c$. The last step is creating directed cycles of infinite capacity, where the only bottleneck is $c$.

\subsection{Any layer formulation}
\label{sec:layers}
As mentioned, the formulation we just gave is not enough to model the \texttt{DQCC} problem to any $\mathcal{L} = \{\ell_1,\ell_2,\dots, \ell_{|\mathcal{L}|}\}$, 
because a circuit generally follows a logic which is related on the order of occurrence given by $\mathcal{L}$. Therefore, even if it might happen that two operations could run in any order, in general this is not true. One needs to define an order relation which is consistent with the logic of the circuit. From an optimization point of view, a critical matter is to choose an order relation \textcolor{black}{that} either wraps most of the good solutions or is prone to optimization algorithms. For this reason and for the sake of clarity, we here refer to a generic, irreflexive, order relation $\prec$ defined over $[k]$, without giving it a unique definition.
Formally, for any $i,j \in [k]$, $j\prec i$ means that to run $i$ we need to ensure that $j$ already ran. 
Starting from $\prec$, we can define a constraint to add to formulation \eqref{mcf}.
Namely, $\forall i \in [k], \forall e \in \delta^-(p_{t_i})$ the following holds:
\begin{equation}
\label{c4}
    f_{e,i}(\tau) \leq \underset{j \prec i}{\text{min}} \sum_{\bar{\tau}< \tau}{f_{e,j}(\bar{\tau})}
\end{equation}
The right part of the inequality is a value in $\{0,1\}$ and takes value $1$ only if all the operations logically preceding $i$ already ran. Notice that constraint \eqref{c4} is linear, as it takes the minimum value among linear functions, and it can be easily mapped to a set of independent constraints $f_{e,i}(\bar{\tau}) \leq \sum_{\bar{\tau}< \tau}{f_{e,j}(\bar{\tau})}$,  $\forall j : j \prec i$.

The formulation now models \texttt{DQCC}. But we will refine inequality \eqref{c4} to get a better solution space -- see Sec. \ref{sec:quasi}.

\section{Increasing the parallelism}
\label{sec:quasi}
\begin{wrapfigure}{r}{4cm}
    \centering
    \begin{quantikz}[thin lines,row sep={0.6cm,between origins}]
        &\ctrl{1}\gategroup[wires=2,steps=1,style={draw=none,rounded corners,inner xsep=1pt,inner ysep=-0.3pt, fill=teal!10}, background, label style={rounded corners,label position=above,  yshift=0.15cm, fill=teal!10}, background]{$i$} & \qw & \qw\\
        &\targ{} & \ctrl{1}\gategroup[wires=2,steps=1,style={draw=none,rounded corners,inner xsep=1pt,inner ysep=-0.3pt, fill=teal!10}, background, label style={rounded corners,label position=above,  yshift=-1.7cm, fill=teal!10}, background]{$j$} & \qw\\
        & \qw & \control{} & \qw
    \end{quantikz}
    \caption{Telegates in logical conflict.
    }
    \label{fig:conflict}
	\hrulefill
\end{wrapfigure}
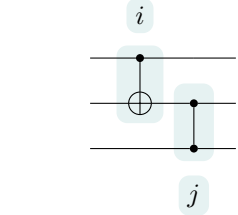
As before, from an optimization point of view, we are interested in considering as many good solutions as possible. To this aim, we propose an interesting approach which should enlarge the space of good solutions. Specifically, we notice that even if two operations $i,j \in [k]$ are such that $i\prec j$, this does not necessarily mean that they must run at different time steps. They, indeed, may run at the same time step and still respecting the logic imposed by $\prec$.

Consider the example from Fig. \ref{fig:conflict}. Since operations $i$ and $j$ operates over a common qubit, they are in logical conflict. Hence, it is reasonable to think that $i\prec j$ should hold.
However, when considering $i$ and $j$ in their extended form -- i.e., where communication qubits are explicit -- we notice that their logical conflict does not map over all the operations involved.
As Fig. \ref{fig:extended} shows, the left part of the equivalence is a naive implementation of $i$ followed by $j$, where the extended form completely inherits the logical conflict. Instead, the right part of the equivalence is way more efficient and it is still an implementation of circuit of Fig. \ref{fig:conflict}.
As consequence, even if $i$ and $j$ are in logical conflict, they can run at the same time step.
We refer to this property as \textit{quasi-parallelism}. For this reason we introduce a new binary relation between operations in $[k]$, which we refer to with the intuitive symbol $\shortparallel$. 
\begin{figure}[h]
    \centering
    \begin{quantikz}[thin lines,row sep={0.6cm,between origins},column sep=0.26cm]
		\slice[style=gray]{$_{\tau}$}&\qw & \ctrl{1} & \qw \slice[style=gray]{$_{\tau+1}$}& \gate[style={fill = green!10}]{_{\texttt{Z}^{\texttt{b}_2}}} &  \qw &\qw\slice[style=gray]{$_{\tau+2}$}&\qw&\qw\\
		& \gate[2, nwires = {1,2}, style={fill = teal!10}]{_{\texttt{E}}} & \targ{} & \measure[style={fill = gray!10}]{_{\langle\texttt{Z}\rangle,\texttt{b}_1}} & \\
		& & \ctrl{1} &  \measure[style={fill = gray!10}]{_{\langle\texttt{X}\rangle,\texttt{b}_2}} &\\
		& \qw & \targ{} & \qw & \gate[style={fill = red!10}]{_{\texttt{X}^{\texttt{b}_1}}} & \ctrl{1}& \qw & \gate[style={fill = green!10}]{_{\texttt{Z}^{\texttt{b}_4}}}&\qw\\
		& & & &\gate[2, nwires = {1,2}, style={fill = teal!10}]{_{\texttt{E}}} & \targ{} & \measure[style={fill = gray!10}]{_{\langle\texttt{Z}\rangle,\texttt{b}_3}} & \\
		& & & & & \ctrl{1} & \measure[style={fill = gray!10}]{_{\langle\texttt{X}\rangle,\texttt{b}_4}} &\\
		& \qw& \qw& \qw & \qw& \control{} & \qw  & \gate[style={fill = green!10}]{_{\texttt{Z}^{\texttt{b}_3}}} & \qw
	\end{quantikz}
	$\ \equiv$
    \begin{quantikz}[thin lines,row sep={0.6cm,between origins},column sep=0.26cm]
		\slice[style=gray]{$_{\tau}$}&\qw & \ctrl{1} & \qw & \qw\slice[style=gray]{$_{\tau+1}$} & \gate[style={fill = green!10}]{_{\texttt{Z}^{\texttt{b}_2}}} &  \qw\\
		& \gate[2, nwires = {1,2}, style={fill = teal!10}]{_{\texttt{E}}} & \targ{} & \qw & \measure[style={fill = gray!10}]{_{\langle\texttt{Z}\rangle,\texttt{b}_1}} & \\
		& & \ctrl{1} & \qw & \measure[style={fill = gray!10}]{_{\langle\texttt{X}\rangle,\texttt{b}_2}} &\\
		& \qw & \targ{} & \ctrl{1}& \qw & \gate[style={fill=violet!10}]{_{\texttt{Z}^{\texttt{b}_4}\texttt{X}^{\texttt{b}_1}}}&\qw\\
		&\gate[2, nwires = {1,2}, style={fill = teal!10}]{_{\texttt{E}}} & \qw &\targ{} & \measure[style={fill = gray!10}]{_{\langle\texttt{Z}\rangle,\texttt{b}_3}} & \\
		& & \ctrl{1} &\qw & \measure[style={fill = gray!10}]{_{\langle\texttt{X}\rangle,\texttt{b}_4}} &\\
		& \qw& \control{} & \qw  & \qw& \gate[style={fill = green!10}]{_{\texttt{Z}^{\texttt{b}_1 \oplus\texttt{b}_3}}} & \qw
	\end{quantikz}
    \caption{Example of how to achieve quasi-parallelism for two telegates in logical conflict.}
    \label{fig:extended}
\end{figure}
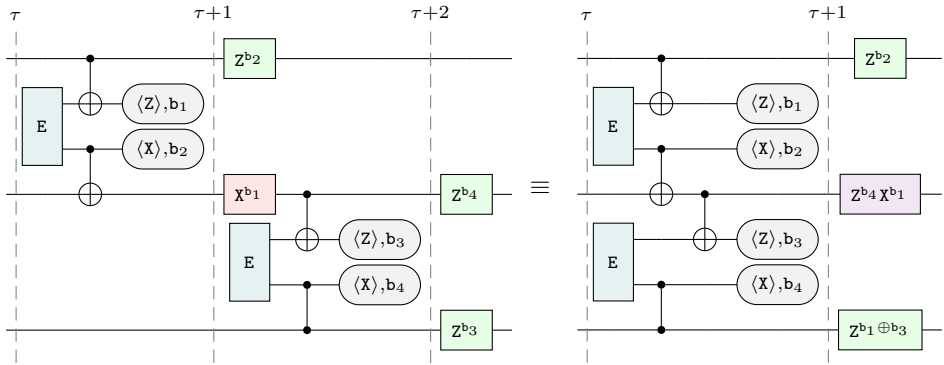

\noindent{}As before, we do not give here a unique definition of $\shortparallel$. Specifically, for any $i,j \in [k]$, we write $i \shortparallel j$ to mean that operations $i$ and $j$ can run at the same time step, but we did not fix a criterion to establish when $\shortparallel$ holds. Clearly, operations $i,j \in [k]$ which can run in full parallelism, are a special case of quasi-parallelism and $i \shortparallel j$ holds.

We can now split the constraint \eqref{c4}, by discriminating between operations which can run in quasi-parallelism and the ones which cannot. Formally, $\forall i \in [k], \forall e \in \delta^-(p_{t_i})$ we introduce two new constraints
\begin{equation}
\label{c5}
    f_{e,i}(\tau) \leq \underset{j \prec i \land j \nshortparallel i}{\text{min}} \sum_{\bar{\tau}< \tau}{f_{e,j}(\bar{\tau})}
\end{equation}
\begin{equation}
\label{c6}
    f_{e,i}(\tau) \leq \underset{j \prec i \land j \shortparallel i}{\text{min}} \sum_{\bar{\tau}\leq \tau}{f_{e,j}(\bar{\tau})}
\end{equation}

To sum up, we propose \eqref{dqcc} as Integer Linear Programming formulation of the \texttt{DQCC} problem. $\mathcal{C}$ is the set of constraints coming from the standard $\texttt{MCF}$ formulation given in \eqref{mcf}.
In what follows we propose a characterization for relation $\shortparallel$.

\begin{figure}[h]
\centering
   \begin{mini}
    {}{f = \sum_{e \in E}\sum_{i \in [k]}\sum_{\tau\in [d]}{f_{e,i}(\tau)}}{}{}
    \label{dqcc}
    \addConstraint{\mathcal{C}}
    \addConstraint{f_{e,i}(\tau) \leq \underset{j \prec i \land j \nshortparallel i}{\text{min}} \sum_{\bar{\tau}< \tau}{f_{e,j}(\bar{\tau})}}{}{\ \ \ \ \forall i \in [k], \forall e \in \delta^-(p_{t_i}), \forall \tau \in [d]}
    \addConstraint{f_{e,i}(\tau) \leq \underset{j \prec i \land j \shortparallel i}{\text{min}} \sum_{\bar{\tau}\leq \tau}{f_{e,j}(\bar{\tau})}}{}{\ \ \ \ \forall i \in [k], \forall e \in \delta^-(p_{t_i}), \forall \tau \in [d]}
    \end{mini}
\end{figure}

\textcolor{black}{For an exemplary characterization of $\shortparallel$, please refer to \cite{cuomo2021optimized} where we introduced a computationally efficient predicate, which states whenever two remote operations $i,j$ can run in quasi-parallelism. We opted to not report this part and rather focus on ways of getting rid of the extra constraints -- see Sec. \ref{sec:norm-cliff}. To this aim we investigate different groups of algorithms and what normal forms they offer. This allows to give a much wider perspective on what kind of shapes the compiler can get in input.}

\color{black}
\section{{The role of Clifford group in distributed architectures}}
\label{sec:norm-cliff}
In our model, we showed that by postponing the Pauli-corrections, we get the combined advantage of (i) parallelizing remote operations and (ii) delaying the correction, which amortizes the impact of the traveling time that a boolean value takes to reach its destination(s). An ideal result would be to push all the corrections to the end of the circuit. In fact, as already discussed in Sec. \ref{sec:e-comp}, if the corrections reach the end of the circuit, they could be replaced by classical computation. Driven by this goal, we now investigate the properties of quantum circuits to find when such a condition is satisfied, starting from the Clifford group $\mathbb{C}$.

The interest in the Clifford group derives from the fact that it covers a wide spectrum of circuits and, to be universal, it needs only one extra operator. In Sec. \ref{sec:cliff}, we referred to such an extension as the group $\mathbb{C}^{\texttt{+}} \equiv \langle \land(\texttt{X}), \texttt{X}^{\sfrac{1}{2}}, \texttt{Z}^{\sfrac{1}{2}}, \texttt{Z}^{\sfrac{1}{4}}\rangle$.
For this reason, it makes sense to represent an arbitrary circuit in terms of a Clifford circuit plus as little $\texttt{Z}^{\sfrac{1}{4}}$ as possible. This is, indeed, an active branch of research \cite{Sel-13, AmyMasMos-14}.


\subsection{\textcolor{black}{Circuit normal forms and implications on the post-processing}}
\label{sec:zx}
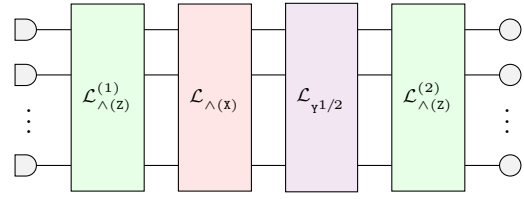
\begin{wrapfigure}{r}{7cm}
    \centering
    \begin{quantikz}[thin lines,row sep={0.6cm,between origins},column sep=0.45cm]
        \meterD[style={fill=gray!10}]{}&\gate[4,nwires={3},style={fill=green!10}]{_{\mathcal{L}^{(1)}_{\land(\texttt{Z})}}}&\gate[4,nwires={3},style={fill=red!10}]{_{\mathcal{L}^{\phantom{(0)}}_{\land(\texttt{X})}}}&\gate[4,nwires={3},style={fill=violet!10}]{_{\mathcal{L}^{\phantom{(0)}}_{\texttt{Y}^{\sfrac{1}{2}}}}}&\gate[4,nwires={3},style={fill=green!10}]{_{\mathcal{L}^{(2)}_{\land(\texttt{Z})}}}&\measure[style={fill=gray!10}]{}\\
        \meterD[style={fill=gray!10}]{}&&&&&\measure[style={fill=gray!10}]{}\\\
        &\lstick[]{$\overset{{\vdots}}{{\phantom{c}}}\ \ \ $}&&&\rstick[]{$\ \ \ \overset{{\vdots}}{\phantom{c}}$}&\\
        \meterD[style={fill=gray!10}]{}&&&&&\measure[style={fill=gray!10}]{}
    \end{quantikz}
    \caption{Normal form coming from the \texttt{ZX}-rules applied in Ref. \cite{duncan2020graph}.}
    \label{fig:norm}
    \hrulefill
\end{wrapfigure}
As said at the beginning of Sec.~\ref{sec:cliff}, important benefits could be achieved by postponing the post-processing to the end of the circuit, where they can be computed classically. An attempt in this direction is available in Ref. \cite{litinski2019game}, where authors delay Pauli operations together with non-Pauli ones. Instead, our approach is to show that the result can always be achieved on the Clifford group, by relying on the \textbf{normal forms} \cite{duncan2020graph,aaronson2004improved,dehaene2003clifford,maslov2018shorter,bataille2020reducing,bravyi2021hadamard}. Such a form results particularly useful for distributed computing and, more in general, for \textit{measurement-based} computation. It has been shown in \cite{duncan2020graph} that any Clifford gate acting on a Pauli state can be represented in the normal form depicted in Fig. \ref{fig:norm}. This normal form is of practical interest as it can be obtained starting from any Clifford circuit, which is in general not in normal form. Such a result comes from the employment of a \textit{\texttt{ZX}-calculus} reasoner, e.g. \cite{kissinger2020Pyzx}.
\texttt{ZX}-calculus \cite{duncan2020graph, van2020zx} is a graphical language, arisen as an optimizer for quantum circuits, that translates a quantum circuit into a \texttt{ZX}\textit{-diagram}. 
The main difference between the diagram and the original circuit is that the former works with \texttt{ZX}\textit{-rules}, which serve as a reasoning tool to smartly generate a new circuit, equivalent to the original one.
\texttt{ZX}-calculus was introduced in the literature in 2007 \cite{coeckegraphical}, with the main objective of minimizing a circuit gate-depth, and its potentiality is still being explored, raising increasing interest for its versatility. In fact, we use it here to perform architecture-compliant optimization. 

Coming back to Fig. \ref{fig:norm}, we use the circuit symbol
{\begin{quantikz}[thin lines,column sep=0.25cm]
    \meterD[style={fill=gray!10}]{}&\qw
\end{quantikz}}\ to express a generic Pauli state preparation. Similarly, the symbol
\begin{quantikz}[thin lines,column sep=0.25cm]
    &\measure[style={fill=gray!10}]{}
\end{quantikz}\ expresses a generic Pauli measurement.
$\mathcal{L}_{\texttt{O}}$ is a set of layers where only the \texttt{O} operator occurs. For example $\mathcal{L}_{\land(\texttt{Z})}$ encodes a circuit composed by $\land(\texttt{Z})$ operators.

For the subject normal form we need to define only a few \textit{pushing rules}. As regard the circuit $\mathcal{L}_{\land(\texttt{X})}$, the following rules always apply:
\begin{itemize}
    \item $\land(\texttt{X})\cdot\texttt{X}^{\texttt{b}} \otimes \mathds{1} \equiv \texttt{X}^{\texttt{b}} \otimes \texttt{X}^{\texttt{b}} \cdot\land(\texttt{X})$
    \item $\land(\texttt{X})\cdot\mathds{1} \otimes \texttt{Z}^{\texttt{b}} \equiv \texttt{Z}^{\texttt{b}} \otimes \texttt{Z}^{\texttt{b}}\cdot\land(\texttt{X})$
    \item $\land(\texttt{X}) \cdot \mathds{1} \otimes \texttt{X}^{\texttt{b}} \equiv \mathds{1} \otimes \texttt{X}^{\texttt{b}} \cdot\land(\texttt{X})$
    \item $\land(\texttt{X})\cdot \texttt{Z}^{\texttt{b}} \otimes \mathds{1} \equiv \texttt{Z}^{\texttt{b}} \otimes \mathds{1}\cdot\land(\texttt{X})$
\end{itemize}
Similarly, for $\mathcal{L}_{\land(\texttt{Z})}$ circuits, we can use the following rules:
\begin{itemize}
    \item $\land(\texttt{Z})\cdot\texttt{X}^{\texttt{b}} \otimes \mathds{1} \equiv \texttt{X}^{\texttt{b}} \otimes \texttt{Z}^{\texttt{b}}\cdot\land(\texttt{Z})$
    \item $\land(\texttt{Z})\cdot\texttt{Z}^{\texttt{b}} \otimes \mathds{1} \equiv \texttt{Z}^{\texttt{b}} \otimes \mathds{1}\cdot\land(\texttt{Z})$
\end{itemize}
Finally, the last single layer circuit $\mathcal{L}_{\texttt{Y}^{\sfrac{1}{2}}}$ can be handled as follows:
\begin{itemize}
    \item $\texttt{Y}^{\sfrac{1}{2}}\cdot\texttt{X}^{\texttt{b}} \cong \texttt{Z}^{\texttt{b}}\cdot\texttt{Y}^{\sfrac{1}{2}}$
    \item $\texttt{Y}^{\sfrac{1}{2}}\cdot\texttt{Z}^{\texttt{b}} \equiv \texttt{X}^{\texttt{b}}\cdot\texttt{Y}^{\sfrac{1}{2}}$
\end{itemize}

\begin{remark}
By means of the above rules, all the post-processing operations can be pushed forward, up to end of the circuit and they can be computed efficiently by a classical computer. Furthermore, since no post-processing occurs during quantum computation, the entanglement path length has negligible impact to the running-time (thanks to the non-locality of the operations).
\end{remark}



The normal form suggests that the problem can be separated into three parts, corresponding to $\mathcal{L}^{(1)}_{\land(\texttt{Z})}$, $\mathcal{L}_{\land(\texttt{X})}$ and $\mathcal{L}^{(2)}_{\land(\texttt{Z})}$. For two of them -- i.e., $\mathcal{L}^{(1)}_{\land(\texttt{Z})}$ and $\mathcal{L}^{(2)}_{\land(\texttt{Z})}$ -- the order relation is trivial (as all $\land(\texttt{Z})$ commute), and therefore we can use any quickest multi-commodity flow solver to get a feasible compilation. On the contrary, the optimal characterization of the order relation for $\mathcal{L}^{\phantom{(0)}}_{\land(\texttt{X})}$ is more complex. Indeed, a set of relations with minimal size may not be the best characterization from a practical point of view, if many of the relations involve remote qubits. The topic of optimal $\land(\texttt{X})$ order relations deserves a dedicated analysis. Hence we first evaluate what we achieved so far, by evaluating our model on $\mathcal{L}_{\land(\texttt{Z})}$, while we investigate later $\mathcal{L}_{\land(\texttt{X})}$ circuits. -- see Sec \ref{sec:cliff-compiler}.

Let us emphasize the importance of $\mathcal{L}_{\land(\texttt{Z})}$ circuits, by pointing out some facts from Ref. \cite{maslov2018shorter}. The authors therein introduce the \textit{Boolean degrees of freedom} as a way to count how many different algorithms can be implemented with a class of gates, and show that a generic $\mathcal{L}_{\land(\texttt{Z})}$ ``has roughly half the number of the degrees of freedom'' compared to a generic $\mathcal{L}_{\land(\texttt{X})}$, and roughly a quarter compared to the Clifford group. 

$\mathcal{L}_{\land(\texttt{Z})}$ circuits represent also an important group for efficient syntheses \cite{maslov2018use,bravyi2022constant,grzesiak2022efficient,bassler2022synthesis}, where these are used to maximize the efficiencies by means of parity check sequences.

We validate our compiler performance by solving $\mathcal{L}_{\land(\texttt{Z})}$ circuits on different architectures in Sec. \ref{sec:cliff-compiler}. So, being able to exploit normal forms to isolate two highly expressive blocks $\mathcal{L}^{(1)}_{\land(\texttt{Z})}$ and $\mathcal{L}^{(2)}_{\land(\texttt{Z})}$ that can be compiled without recurring to order relations, is a very relevant result.

\subsection{Analysis on the upper-bounds and future perspective}
There is a fair doubt arising from the employment of normal forms for compilation: do we know the overhead cause by mapping any Clifford circuit to some normal form? If yes, is it reasonable?

The answer is positive to both questions. By working with normal forms, we are not only able to work with a circuit with known shape, but we can also upper-bound the overhead for the number of introduced operations. Depending on whether or not ancillae are considered, the system get more complex in terms of space or run-time. In Ref. \cite{bravyi2022constant}, authors treat both cases, and prove linear upper-bounds.

Normal forms unlock also better opportunities from an hardware perspective. Specifically, dealing with well defined circuit allows to extend the gate set with more practical operators, as the $\land(\texttt{X}^{\otimes m})$ introduced in Ch. \ref{ch:essentials}. In Sec. \ref{sec:cliff-compiler} we refer to these gates as \textit{fan-in} (and \textit{fan-out}) gates. From a hardware perspective, these are also commonly referred as \textit{global gates}, as they may act over a large $m$ simultaneously \cite{lu2019global,casanova2012quantum,ivanov2015efficient,martinez2016compiling}. Citing \cite{lu2019global}: ``It has been suggested that polynomial or exponential speedups can be obtained with global [gates]''.

Other results in terms of overhead can be found in Ref. \cite{van2021constructing}, where authors proved that any $n$-qubit Clifford circuit can be synthesised to $4n-6$ global gates and any $n$-qubit circuit with $\dot{n}$ non-Clifford gates can be synthesised with no more than $2\dot{n} + \mathcal{O}(\sfrac{n}{\log{n}})$ global gates.

Ultimately, our choice to employ the normal form of Fig. \ref{fig:norm} has several benefits, besides the ones we already discussed:
\begin{itemize}
    \item It is practical, as the open-source \texttt{pyzx} \cite{kissinger2020Pyzx} provides the tools to perform the mapping.
    \item It is efficient, as the \texttt{pyzx} engine works to minimize the number of two-qubit gates.
    \item It has a good shape, as $\mathcal{L}_{\land(\texttt{Z})}$ circuits are generally easier than $\mathcal{L}_{\land(\texttt{X})}$ ones.
\end{itemize}

Let us make a final remark on \texttt{ZX}-calculus. We introduced it in the context of the Clifford group, but it is designed to work more broadly with any circuit \cite{jeandel2018complete, carette2021completeness, backens2014zx, kissinger2019reducing}. Therefore, we aim to expand our analysis in future works, by investigating normal forms for universal circuits. An interesting result in this direction is available in Ref. \cite{heyfron2018efficient}, where authors split a universal circuit into the following three steps:
\begin{enumerate}
    \item the system is prepared in a \textit{non-Clifford state} \cite{van2021constructing}, this involves auxiliary qubits which will do the work of injecting non-Clifford phases;
    \item an $\mathcal{L}_{\land(\texttt{X})}$ circuit;
    \item a measurement-based sequence of Clifford operations -- which can still be treated with \texttt{ZX}-calculus \cite{duncan2012graphical}.
\end{enumerate}

\section{Commuting circuits compiler}
\label{sec:cz-comp}
As distributed quantum architectures are still at an early stage, it is hard to predict with confidence what kind of connectivity and resources they will supply. It is therefore of interest to investigate on what kind of topology a distributed architecture should have. For this reason we now report a compiler for $\mathcal{L}_{\land(\texttt{Z})}$ and use it to evaluate the performance given by different topologies. Thanks to this evaluation, we could choose the topology most performing with our experiments\footnote{The same topology is also employed in \cite{pant2019routing}, where authors deal with unreliable optical links to create entanglement and dynamically choose a multi-path solution in order to maximise the entanglement success-rate.}.

Here we evaluate the \textit{rectangle lattice} topology -- see Figs. \ref{fig:lowgrid}, \ref{fig:biggrid} -- and comparing it with a \textit{hexagon lattice} topology -- see Fig \ref{fig:hexa}. We therefore verify the compiler performance for both the lattices in terms of:
\begin{itemize}
    \item solution quality;
    \item robustness to scale-up.
\end{itemize}
We conclude the comparison with the possible implications of the results.

\subsection{An approximation-based implementation}
\label{sec:approx}
\begin{wrapfigure}{R}{7cm}
\begin{algorithm}[H]
\DontPrintSemicolon
\KwIn{$\mathcal{Q}, [k]$}
\KwOut{$d$}
    $S \leftarrow [k]$\;
    $d \leftarrow 0$\;
    \While{$S \neq \varnothing$}{
        {
            $\bar{S} \leftarrow  \texttt{MCF}(\mathcal{Q}, S)$\;
            $S \leftarrow S \smallsetminus \bar{S}$\;
            $d \leftarrow d + 1$\;
        }
    }
\caption{\small{Iterative compiler}}
\label{algo:approx}
\end{algorithm}
\end{wrapfigure}
We already discussed in Sec. \ref{sec:layer} how to tackle \texttt{DQCC} as a particular case of quickest multi-commodity flow. In this way we managed to reduce the problem on the resolution of one or more static instances of \texttt{MCF}. In Refs.  \cite{kleitman1970matching,kleitman1971algorithm} it has been shown that whenever each commodity is a source (or a target) for any other node, then solving it through \texttt{LP}-relaxation outputs an optimal solution to \texttt{MCF}. This result can be of interest when treating \textit{fully entangling circuits}.

To keep the compiler more general, we opted to investigate algorithms with approximation boundary guaranteed \cite{mcf-edp, martens2009simple, kolman2002improved}. Specifically, we implemented the pseudo-code outlined in Ref. \cite{approx}. This is followed by a proof on the approximation quality for the case of capacity $c = 1$ and $c > 1$. We focus on the case $c=1$, but it can be extended to $c > 1$.

By using our formalism, the approximation algorithm aims to run as many non-local operators -- i.e. satisfying commodities demand -- as possible. A computed solution is a sub-set $S \subseteq [k]$. The optimal solution is $S^* \subseteq [k]$ and $|S| \leq |S^*|$. It follows the (optimal) approximation boundary \cite{martens2009simple, approx}:
\begin{equation}
    \label{eq:approx}
    |S| \geq \frac{|S^*|}{\mathcal{O}(\sqrt{m})},\ m = |E|
\end{equation}

Notice that the solution quality is inversely proportional to the number of entanglement links. It means that we cannot estimate an optimal solution to the \texttt{DQCC}, as for a given time horizon, this affects the quality of the solution space. Furthermore, the time-expansion increases the number of edges and so does the distance $|S^*| - |S|$. Ultimately, even if the allocated space by the time-expansion grows at most linearly with the number of non-local operations -- see Sec. \ref{sec:layer} --, this can seriously affect the performance when such an amount is very big\footnote{Better upper-bounds for the worst-case solution should be investigated.}. On contrary, it is possible to keep the time-expansion \textit{abstract} and compiling iteratively as many operations as possible at each time-step. This method is detailed in Algorithm \ref{algo:approx}.
Notice that each iteration guarantees the boundary of equation \eqref{eq:approx} and, above all, since the instance decreases in size, the distance $|S^*| - |S|$ tends to decrease as well.

\subsection{\textcolor{black}{Set-up}}
To compare the compiler performance on different topologies, we make use of a \textit{generator factor} $g \in \mathbb{N}\smallsetminus\{0\}$. The number of nodes and edges of each lattice will be expressed as a function of $g$. Because the two lattices differ by definition, it is not trivial to settle a fair comparison. 
To do that, we first generate\footnote{With the help of the \texttt{networkx} library \cite{HagSchSwa-08}.} a sample of hexagon lattices $\mathcal{H}$ such that
\begin{equation}
    \label{eq:hexa_size}
    |P| = \sfrac{1}{2}\cdot g^2 + 3g + \mathcal{O}(1),\ |E| = \sfrac{3}{4}\cdot g^2 + \sfrac{7}{2}\cdot g \pm \mathcal{O}(1).
\end{equation}
We compare $\mathcal{H}$ with two rectangle lattices, say $\mathcal{R}_{_{\blacktriangledown}}$ and $\mathcal{R}_{_{\blacktriangle}}$, that have sizes respectively lower and higher than $\mathcal{H}$ for each $g$ -- see Fig. \ref{fig:lattices}. Hence, $\mathcal{R}_{_{\blacktriangledown}}$ is such that
\begin{equation}
    \label{eq:lowgrid_size}
    |P| = \sfrac{1}{4}\cdot g^2 + \sfrac{3}{2}\cdot g + \mathcal{O}(1),\ |E| = \sfrac{1}{2}\cdot g^2 + 2g \pm \mathcal{O}(1).
\end{equation}
while $\mathcal{R}_{_{\blacktriangle}}$ is such that
\begin{equation}
    \label{eq:biggrid_size}
    |P| = 2g^2 + 2g,\ |E| = g^2 + 2g \pm \mathcal{O}(1).
\end{equation}
We show in the next subsection that $\mathcal{R}_{_{\blacktriangle}}$ and $\mathcal{R}_{_{\blacktriangledown}}$ perform better than $\mathcal{H}$ in terms of resulting \texttt{E}-depth. This implies that the rectangle lattice is a better design for distributed quantum computers, assuming that our compiler performs equally well on different topologies.
\begin{figure}[hb]
     \centering
     \begin{subfigure}[b]{0.32\textwidth}
         \centering
         \includegraphics[width=0.4\textwidth]{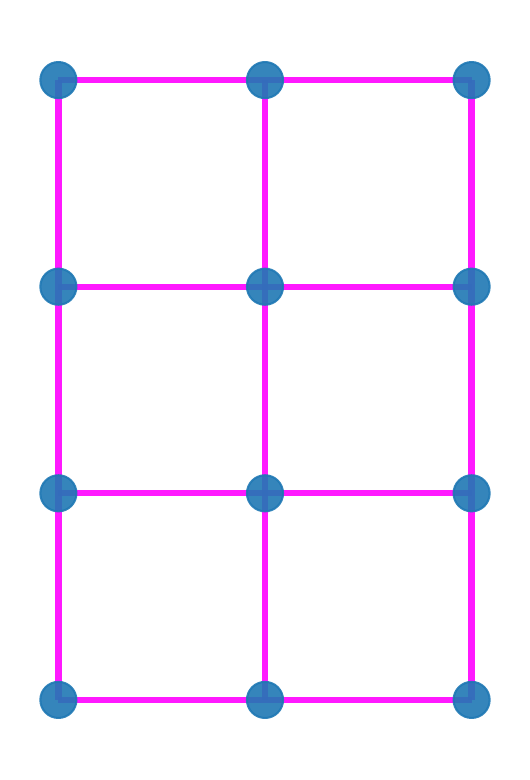}
         \caption{rectangle lattice $\mathcal{R}_{_{\blacktriangledown}}$.}
         \label{fig:lowgrid}
     \end{subfigure}
     \hfill
     \begin{subfigure}[b]{0.32\textwidth}
         \centering
         \includegraphics[width=0.68\textwidth]{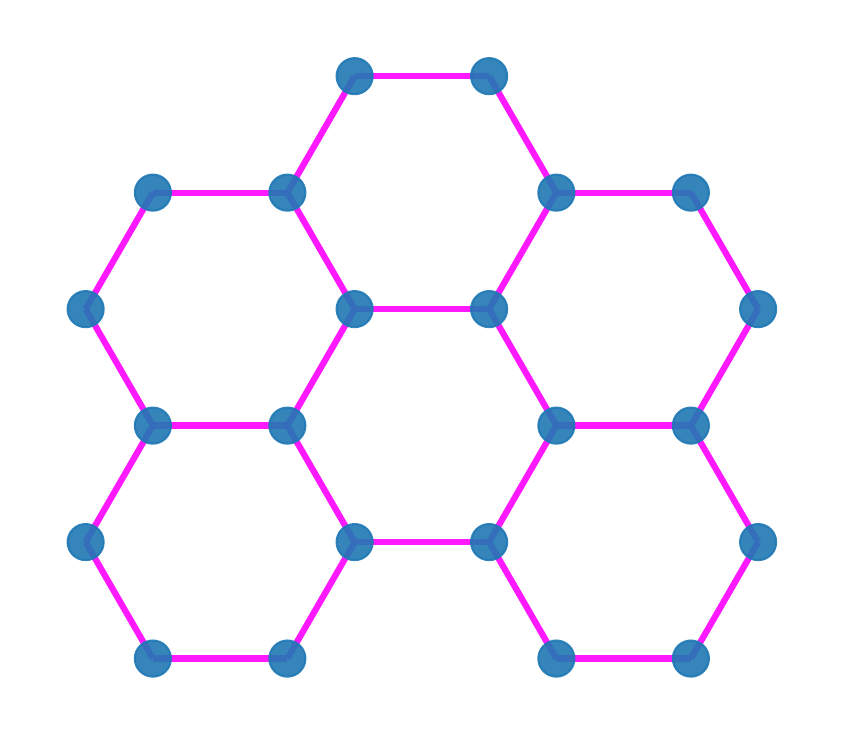}
         \caption{Hexagon lattice $\mathcal{H}$.}
         \label{fig:hexa}
     \end{subfigure}
     \hfill
     \begin{subfigure}[b]{0.32\textwidth}
         \centering
         \includegraphics[width=0.59\textwidth]{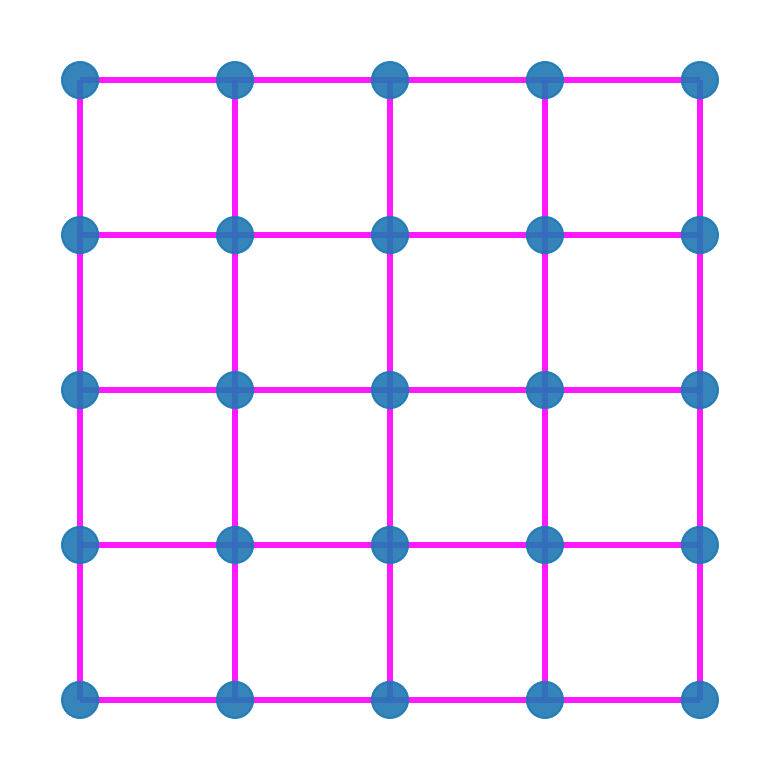}
         \caption{rectangle lattice $\mathcal{R}_{_{\blacktriangle}}$.}
         \label{fig:biggrid}
     \end{subfigure}
     \caption{\textcolor{black}{Example of lattices used for the experimental evaluation; they all come from generator $g = 4$.}}
     \label{fig:lattices}
\end{figure}

Since we use Algorithm \ref{algo:approx}, capacities are assumed to be 1. We already pointed out that such an algorithm can be extended to the case $c>1$.

Notice that different node degrees imply different assumptions on the processor units $P_i$. The hexagon lattice has node degree upper-bounded by $3$ and lower-bounded by $2$, which means that $P_i$ has $2$ to $3$ communication qubits. Similarly, the rectangle lattice has degree upper-bounded by $4$. Hence, the communication qubits per unit are $2$ to $4$. Since our focus here is on distributed compilation, we will assume that $P_i$ has $1$ computation qubit. This is especially reasonable when considering that real implementation of distributed architecture may use most of their local resources as \textit{auxiliary qubits}, meant to keep the computation fault-tolerant.


For the numerical evaluation we use a generating vector $\mathbf{g} = (1,2,\dots,11)$. Hence, when the generator is fixed to $11$, the size of $\mathcal{H}$ reaches $|P| = 96$ and $|E| = 131$, $\mathcal{R}_{_{\blacktriangledown}}$ reaches $|P| = 49$ and $|E| = 84$, while $\mathcal{R}_{_{\blacktriangledown}}$ reaches  $|P| = 144$ and $|E| = 264$.


We generate three samples classified by their size (or number of occurring operators). Each sample is composed by $10$ random circuits in order to average the results. The size of the samples are $256$, $512$ and $1024$.


\subsection{Architecture evaluation}
\begin{figure}[ht]
     \centering
     \begin{subfigure}[b]{0.3\textwidth}
         \centering
         \includegraphics[width=\textwidth]{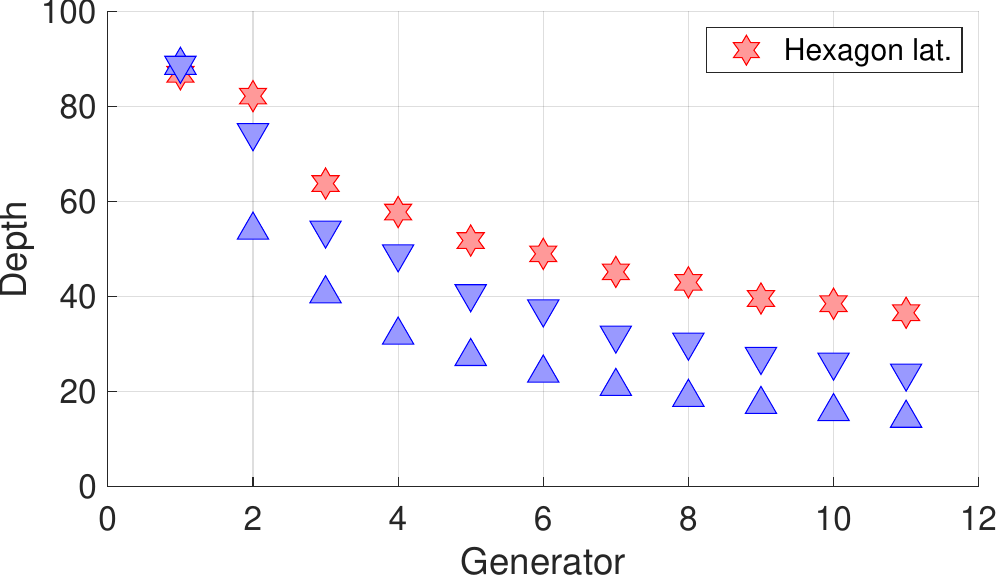}
         \caption{256 $\land(\texttt{Z})$.}
        \label{fig:depth256}
     \end{subfigure}
     \hfill
     \begin{subfigure}[b]{0.3\textwidth}
         \centering
         \includegraphics[width=\textwidth]{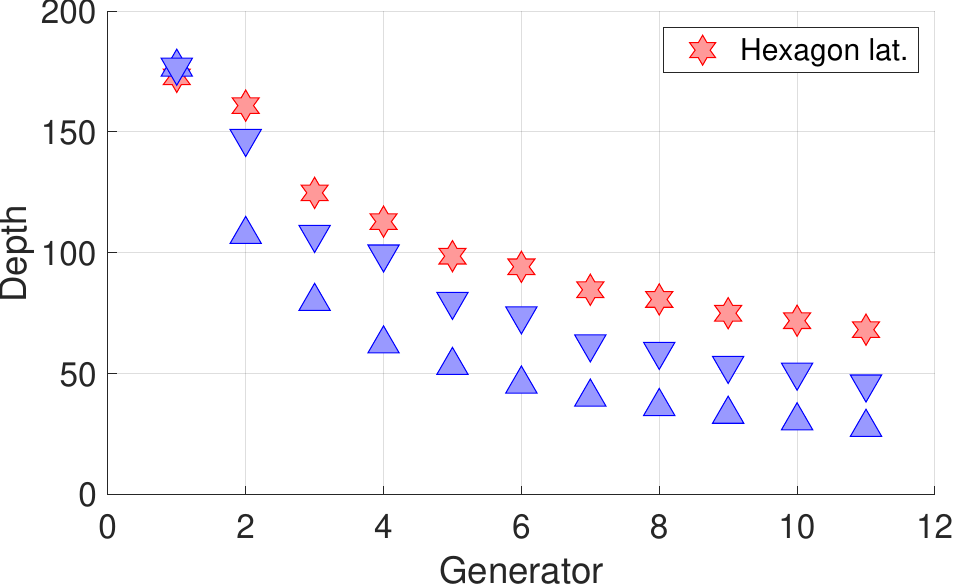}
         \caption{512 $\land(\texttt{Z})$.}
        \label{fig:depth512}
     \end{subfigure}
     \hfill
     \begin{subfigure}[b]{0.3\textwidth}
         \centering
         \includegraphics[width=\textwidth]{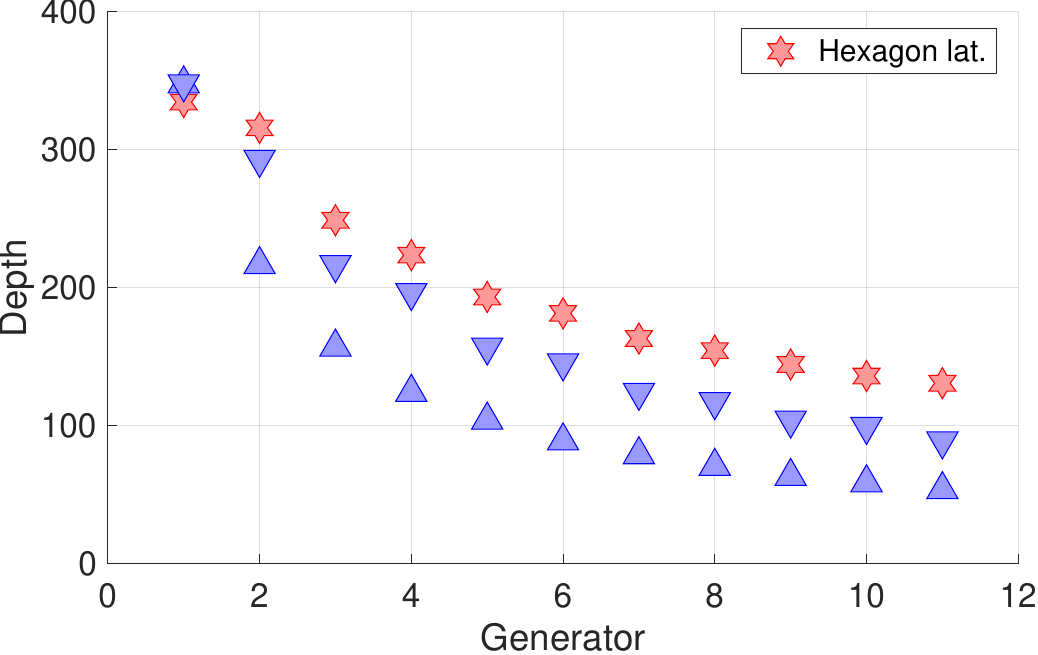}
         \caption{1024 $\land(\texttt{Z})$.}
        \label{fig:depth1024}
     \end{subfigure}\\
     \begin{subfigure}[b]{0.3\textwidth}
         \centering
         \includegraphics[width=1.028\textwidth]{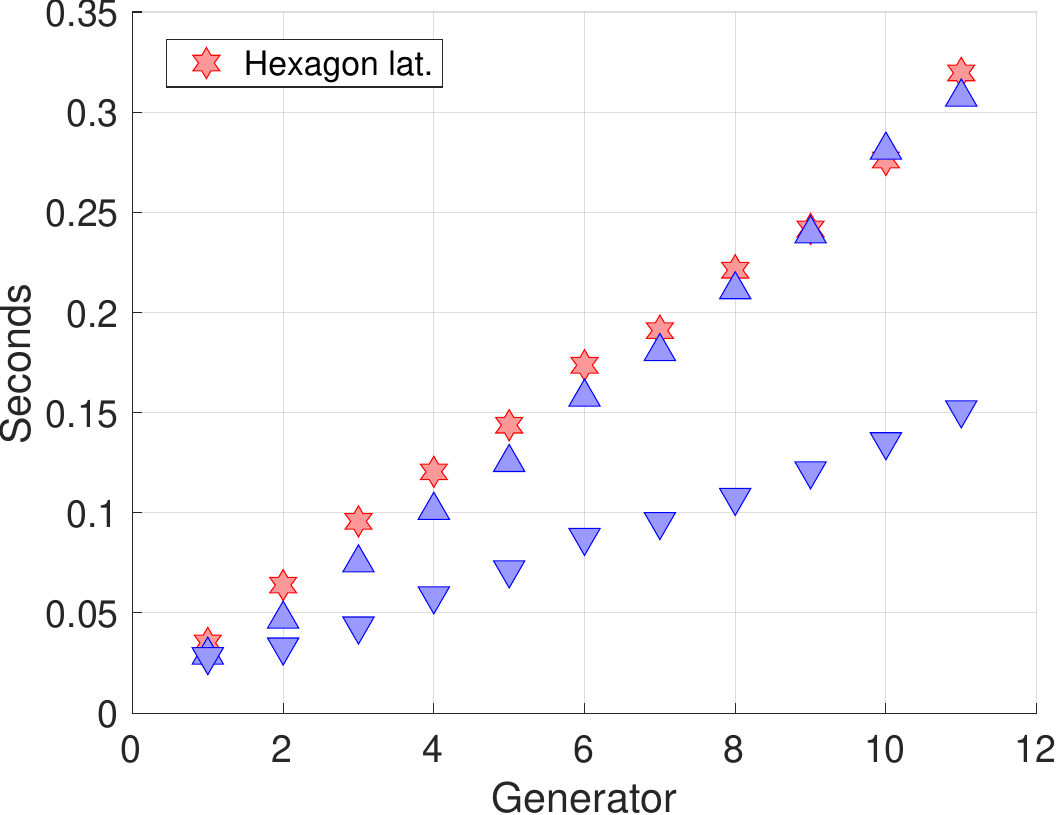}
         \caption{256 $\land(\texttt{Z})$.}
        \label{fig:time256}
     \end{subfigure}
     \hfill
     \begin{subfigure}[b]{0.3\textwidth}
         \centering
         \includegraphics[width=1.01\textwidth]{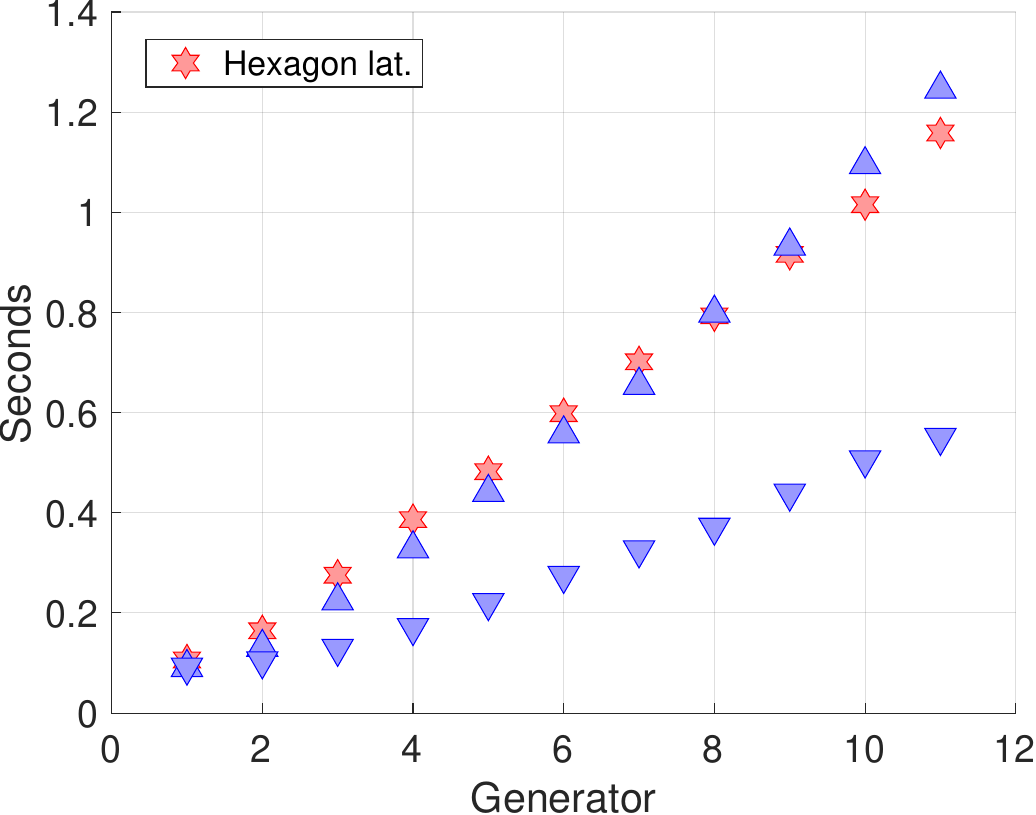}
         \caption{512 $\land(\texttt{Z})$.}
        \label{fig:time512}
     \end{subfigure}
     \hfill
     \begin{subfigure}[b]{0.3\textwidth}
         \centering
         \includegraphics[width=0.98\textwidth]{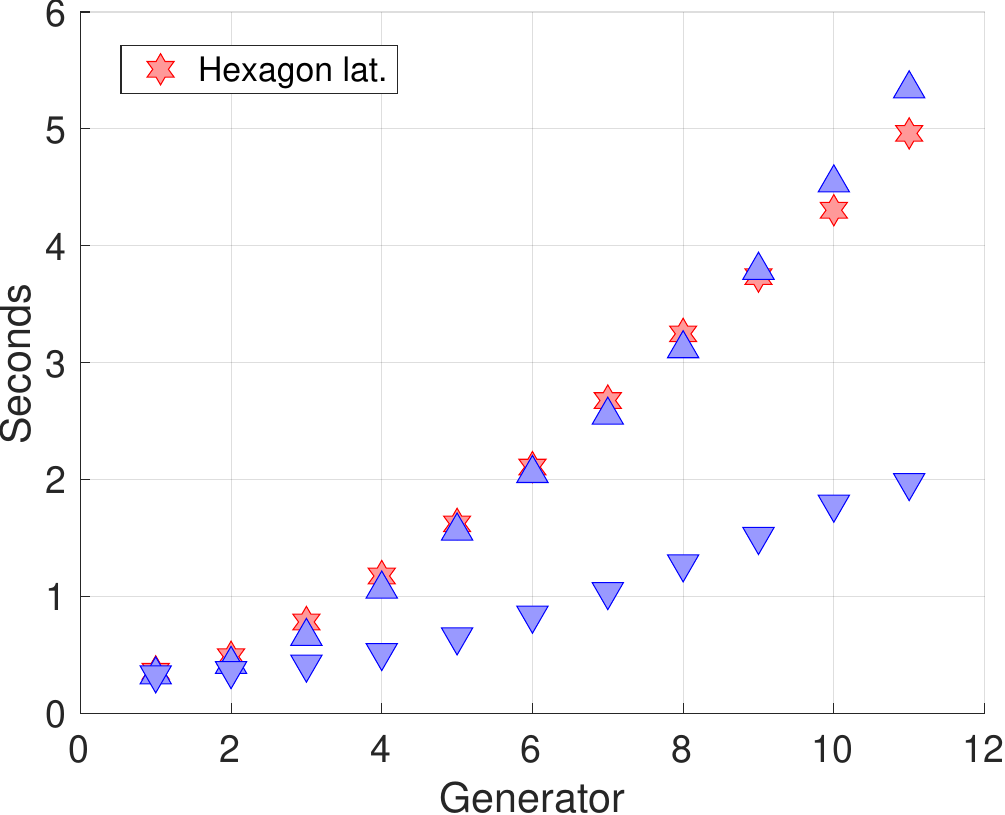}
         \caption{1024 $\land(\texttt{Z})$.}
        \label{fig:time1024}
     \end{subfigure}
        \caption{\textcolor{black}{Quality and time scale comparison.}}
\end{figure}
To evaluate the results we used the \texttt{matlab} environment \cite{MATLAB:2021}. The employed architecture is a MacBook Air - M1, 2020, 8GB RAM.

The first result -- shown in Figs. \ref{fig:depth256}, \ref{fig:depth512} and \ref{fig:depth1024} -- is a comparison on the solution quality, a.k.a. the \texttt{E}-depth. 

As anticipated, the plots show that a rectangle lattice gives better solutions, for any problem size. We can relate this behavior to the \textit{ratio edges-to-nodes}. Formally, let $\mathfrak{r}_{\mathcal{Q}} = \frac{|E|}{|P|}$ be such a ratio for a graph $\mathcal{Q}$. Then it results that rectangle lattices have ratio:
\begin{equation}
    \label{eq:sqr_ratio}
    \lim_{g\to\infty}\mathfrak{r}_{\mathcal{R}} = 2.
\end{equation}
Instead hexagon lattices have a lower ratio:
\begin{equation}
    \label{eq:hex_ratio}
    \lim_{g\to\infty}\mathfrak{r}_{\mathcal{H}} = \sfrac{3}{2}.
\end{equation}
This suggests that the bigger the ratio, the better the solutions. The plots also show that the depth achieved by the different lattices may be ruled by the same polynomial function (up to some constant factor). This is in line with the intuition that a more connected topology allows for shorter depth. Furthermore, we already mentioned in Sec. \ref{sec:approx} that, even if the approximation algorithm depends on the edges size, this is called as a subroutine that performs better and better at each iteration. All this may mean that the compiler has a convergence to an optimal depth. On contrary, if the compiler was affected by the number of edges, the functions should swap at some point, but we never observed such phenomenon.


To conclude our evaluation, we took the average times for each sample. The results are shown in Figs. \ref{fig:time256}, \ref{fig:time512} and \ref{fig:time1024}. Differently from what we got in the solution quality evaluation -- where we noticed a similar behaviour for each architecture -- the time-scale gives new perspectives in the lattices comparison. In fact, $\mathcal{H}$ and $\mathcal{R}_{_{\blacktriangle}}$ seems to need approximately the same time to compile any circuit, with $\mathcal{R}_{_{\blacktriangle}}$ performing slightly worse, which is coherent with the size difference between the twos. 
Instead, $\mathcal{R}_{_{\blacktriangledown}}$ outperforms the others lattices. Furthermore, it seems that it is more resistant to scale-up as the scaling seems to follow a lower degree function.

Thanks to the given experimental evaluation, we could give a first evaluation of our model, by implementing it for an important group of commuting circuits -- i.e. the $\mathcal{L}_{\land(\texttt{Z})}$. At the same time we could find a good topology for our next step, which is compiling Clifford circuits.

%% file: compilation/Figures/quotient.tex
\tikzset{every picture/.style={line width=0.75pt}} 

\begin{tikzpicture}[x=0.6pt,y=0.6pt,yscale=-1,xscale=1]

\draw [color={rgb, 255:red, 189; green, 16; blue, 224 }  ,draw opacity=1 ][line width=1.5]    (158.57,44) -- (196.69,44) ;
\draw [color={rgb, 255:red, 189; green, 16; blue, 224 }  ,draw opacity=1 ][line width=1.5]    (87.49,43.84) -- (126.26,44) ;
\draw  [color={rgb, 255:red, 0; green, 0; blue, 0 }  ,draw opacity=1 ][fill={rgb, 255:red, 74; green, 144; blue, 226 }  ,fill opacity=0.7 ][dash pattern={on 5.63pt off 4.5pt}][line width=1.5]  (55.5,43.84) .. controls (55.5,35.09) and (62.66,28) .. (71.49,28) .. controls (80.33,28) and (87.49,35.09) .. (87.49,43.84) .. controls (87.49,52.59) and (80.33,59.68) .. (71.49,59.68) .. controls (62.66,59.68) and (55.5,52.59) .. (55.5,43.84) -- cycle ;
\draw  [color={rgb, 255:red, 0; green, 0; blue, 0 }  ,draw opacity=1 ][fill={rgb, 255:red, 74; green, 144; blue, 226 }  ,fill opacity=0.7 ][dash pattern={on 5.63pt off 4.5pt}][line width=1.5]  (126.26,44) .. controls (126.26,35.16) and (133.49,28) .. (142.41,28) .. controls (151.33,28) and (158.57,35.16) .. (158.57,44) .. controls (158.57,52.84) and (151.33,60) .. (142.41,60) .. controls (133.49,60) and (126.26,52.84) .. (126.26,44) -- cycle ;
\draw  [color={rgb, 255:red, 0; green, 0; blue, 0 }  ,draw opacity=1 ][fill={rgb, 255:red, 74; green, 144; blue, 226 }  ,fill opacity=0.7 ][dash pattern={on 5.63pt off 4.5pt}][line width=1.5]  (196.69,44) .. controls (196.69,35.16) and (203.92,28) .. (212.85,28) .. controls (221.77,28) and (229,35.16) .. (229,44) .. controls (229,52.84) and (221.77,60) .. (212.85,60) .. controls (203.92,60) and (196.69,52.84) .. (196.69,44) -- cycle ;

\draw (71.49,43.84) node  [font=\scriptsize,color={rgb, 255:red, 255; green, 255; blue, 255 }  ,opacity=1 ]  {$p_{1}$};
\draw (142.41,44) node  [font=\scriptsize,color={rgb, 255:red, 255; green, 255; blue, 255 }  ,opacity=1 ]  {$p_{2}$};
\draw (106.87,43.92) node [anchor=south] [inner sep=2pt]  [font=\scriptsize]  {$2$};
\draw (212.85,44) node  [font=\scriptsize,color={rgb, 255:red, 255; green, 255; blue, 255 }  ,opacity=1 ]  {$p_{3}$};
\draw (177.63,44) node [anchor=south] [inner sep=2pt]  [font=\scriptsize]  {$1$};

\end{tikzpicture}

%% file: compilation/Figures/digraph.tex
\tikzset{every picture/.style={line width=0.5pt}} 

\begin{tikzpicture}[x=0.6pt,y=0.6pt,yscale=-1,xscale=1]

\draw [color={rgb, 255:red, 189; green, 16; blue, 224 }  ,draw opacity=1 ][line width=1.5]    (178.77,40.1) -- (126.47,40.1) ;
\draw  [color={rgb, 255:red, 0; green, 0; blue, 0 }  ,draw opacity=1 ][fill={rgb, 255:red, 74; green, 144; blue, 226 }  ,fill opacity=0.7 ][dash pattern={on 5.63pt off 4.5pt}][line width=1.5]  (94.24,41.6) .. controls (94.24,32.81) and (101.46,25.69) .. (110.36,25.69) .. controls (119.25,25.69) and (126.47,32.81) .. (126.47,41.6) .. controls (126.47,50.38) and (119.25,57.5) .. (110.36,57.5) .. controls (101.46,57.5) and (94.24,50.38) .. (94.24,41.6) -- cycle ;
\draw  [color={rgb, 255:red, 0; green, 0; blue, 0 }  ,draw opacity=1 ][fill={rgb, 255:red, 74; green, 144; blue, 226 }  ,fill opacity=0.7 ][dash pattern={on 5.63pt off 4.5pt}][line width=1.5]  (178.77,41.6) .. controls (178.77,32.81) and (185.98,25.69) .. (194.88,25.69) .. controls (203.78,25.69) and (211,32.81) .. (211,41.6) .. controls (211,50.38) and (203.78,57.5) .. (194.88,57.5) .. controls (185.98,57.5) and (178.77,50.38) .. (178.77,41.6) -- cycle ;
\draw [color={rgb, 255:red, 189; green, 16; blue, 224 }  ,draw opacity=1 ][line width=1.5]    (132.05,119.8) -- (83.47,119.43) ;
\draw [shift={(136.05,119.83)}, rotate = 180.44] [fill={rgb, 255:red, 189; green, 16; blue, 224 }  ,fill opacity=1 ][line width=0.08]  [draw opacity=0] (8.75,-4.2) -- (0,0) -- (8.75,4.2) -- (5.81,0) -- cycle    ;
\draw  [color={rgb, 255:red, 0; green, 0; blue, 0 }  ,draw opacity=1 ][fill={rgb, 255:red, 74; green, 144; blue, 226 }  ,fill opacity=0.7 ][dash pattern={on 5.63pt off 4.5pt}][line width=1.5]  (51.24,120.43) .. controls (51.24,111.64) and (58.46,104.52) .. (67.36,104.52) .. controls (76.25,104.52) and (83.47,111.64) .. (83.47,120.43) .. controls (83.47,129.21) and (76.25,136.33) .. (67.36,136.33) .. controls (58.46,136.33) and (51.24,129.21) .. (51.24,120.43) -- cycle ;
\draw [color={rgb, 255:red, 189; green, 16; blue, 224 }  ,draw opacity=1 ][line width=1.5]    (151.96,185.24) -- (151.96,136.94) ;
\draw [shift={(151.96,189.24)}, rotate = 270] [fill={rgb, 255:red, 189; green, 16; blue, 224 }  ,fill opacity=1 ][line width=0.08]  [draw opacity=0] (8.75,-4.2) -- (0,0) -- (8.75,4.2) -- (5.81,0) -- cycle    ;
\draw  [color={rgb, 255:red, 0; green, 0; blue, 0 }  ,draw opacity=1 ][fill={rgb, 255:red, 74; green, 144; blue, 226 }  ,fill opacity=0.7 ][dash pattern={on 5.63pt off 4.5pt}][line width=1.5]  (135.84,120.83) .. controls (135.84,112.05) and (143.06,104.92) .. (151.96,104.92) .. controls (160.85,104.92) and (168.07,112.05) .. (168.07,120.83) .. controls (168.07,129.61) and (160.85,136.73) .. (151.96,136.73) .. controls (143.06,136.73) and (135.84,129.61) .. (135.84,120.83) -- cycle ;
\draw  [color={rgb, 255:red, 0; green, 0; blue, 0 }  ,draw opacity=1 ][fill={rgb, 255:red, 74; green, 144; blue, 226 }  ,fill opacity=0.7 ][dash pattern={on 5.63pt off 4.5pt}][line width=1.5]  (135.84,205.36) .. controls (135.84,196.57) and (143.06,189.45) .. (151.96,189.45) .. controls (160.85,189.45) and (168.07,196.57) .. (168.07,205.36) .. controls (168.07,214.14) and (160.85,221.26) .. (151.96,221.26) .. controls (143.06,221.26) and (135.84,214.14) .. (135.84,205.36) -- cycle ;
\draw [color={rgb, 255:red, 189; green, 16; blue, 224 }  ,draw opacity=1 ][line width=1.5]    (221.05,121.83) -- (171.84,121.55) ;
\draw [shift={(167.84,121.52)}, rotate = 0.33] [fill={rgb, 255:red, 189; green, 16; blue, 224 }  ,fill opacity=1 ][line width=0.08]  [draw opacity=0] (8.75,-4.2) -- (0,0) -- (8.75,4.2) -- (5.81,0) -- cycle    ;
\draw [color={rgb, 255:red, 189; green, 16; blue, 224 }  ,draw opacity=1 ][line width=1.5]    (167.86,205.07) -- (234.09,140.73) ;
\draw [shift={(236.96,137.94)}, rotate = 135.83] [fill={rgb, 255:red, 189; green, 16; blue, 224 }  ,fill opacity=1 ][line width=0.08]  [draw opacity=0] (8.75,-4.2) -- (0,0) -- (8.75,4.2) -- (5.81,0) -- cycle    ;
\draw  [color={rgb, 255:red, 0; green, 0; blue, 0 }  ,draw opacity=1 ][fill={rgb, 255:red, 74; green, 144; blue, 226 }  ,fill opacity=0.7 ][dash pattern={on 5.63pt off 4.5pt}][line width=1.5]  (220.84,121.83) .. controls (220.84,113.05) and (228.06,105.92) .. (236.96,105.92) .. controls (245.85,105.92) and (253.07,113.05) .. (253.07,121.83) .. controls (253.07,130.61) and (245.85,137.73) .. (236.96,137.73) .. controls (228.06,137.73) and (220.84,130.61) .. (220.84,121.83) -- cycle ;
\draw [color={rgb, 255:red, 189; green, 16; blue, 224 }  ,draw opacity=1 ][line width=1.5]    (136.05,205.65) -- (70.17,139.17) ;
\draw [shift={(67.36,136.33)}, rotate = 45.26] [fill={rgb, 255:red, 189; green, 16; blue, 224 }  ,fill opacity=1 ][line width=0.08]  [draw opacity=0] (8.75,-4.2) -- (0,0) -- (8.75,4.2) -- (5.81,0) -- cycle    ;
\draw [color={rgb, 255:red, 0; green, 0; blue, 0 }  ,draw opacity=0.8 ][line width=1.5]    (152.62,64.38) -- (152.52,86) ;
\draw [shift={(152.5,90)}, rotate = 270.27] [fill={rgb, 255:red, 0; green, 0; blue, 0 }  ,fill opacity=0.8 ][line width=0.08]  [draw opacity=0] (6.97,-3.35) -- (0,0) -- (6.97,3.35) -- cycle    ;
\draw [shift={(152.62,64.38)}, rotate = 270.27] [color={rgb, 255:red, 0; green, 0; blue, 0 }  ,draw opacity=0.8 ][line width=1.5]    (0,4.02) -- (0,-4.02)   ;

\draw (110.36,43.1) node  [font=\scriptsize,color={rgb, 255:red, 255; green, 255; blue, 255 }  ,opacity=1 ]  {$p_{i}$};
\draw (194.88,43.1) node  [font=\scriptsize,color={rgb, 255:red, 255; green, 255; blue, 255 }  ,opacity=1 ]  {$p_{j}$};
\draw (67.36,121.43) node  [font=\scriptsize,color={rgb, 255:red, 255; green, 255; blue, 255 }  ,opacity=1 ]  {$p_{i}$};
\draw (151.96,120.83) node  [font=\scriptsize,color={rgb, 255:red, 255; green, 255; blue, 255 }  ,opacity=1 ]  {$p_{i'}$};
\draw (151.96,205.36) node  [font=\scriptsize,color={rgb, 255:red, 255; green, 255; blue, 255 }  ,opacity=1 ]  {$p_{j'}$};
\draw (236.96,121.83) node  [font=\scriptsize,color={rgb, 255:red, 255; green, 255; blue, 255 }  ,opacity=1 ]  {$p_{j}$};
\draw (152.62,43.1) node [anchor=south] [inner sep=3pt]  [font=\scriptsize]  {$c$};
\draw (151.96,163.09) node [anchor=west] [inner sep=2pt]  [font=\scriptsize]  {$c$};
\draw (194.45,122.68) node [anchor=south] [inner sep=3pt]  [font=\scriptsize]  {$\infty $};
\draw (109.76,121.63) node [anchor=south] [inner sep=3pt]  [font=\scriptsize]  {$\infty $};
\draw (101.7,170.99) node [anchor=north east] [inner sep=0.75pt]  [font=\scriptsize]  {$\infty $};
\draw (202.41,171.51) node [anchor=north west][inner sep=0.75pt]  [font=\scriptsize]  {$\infty $};

\end{tikzpicture}